\documentclass[12pt,a4paper,oneside]{article}

\usepackage[T1]{fontenc}
\usepackage[utf8]{inputenc}
\usepackage{lmodern}
\usepackage{microtype}
\microtypesetup{expansion=false}
\usepackage{amsmath,amssymb,amsthm,mathtools}
\usepackage{booktabs}
\usepackage{enumitem}
\usepackage{needspace}
\usepackage{xspace}
\usepackage[dvipsnames]{xcolor}
\definecolor{linkmaroon}{HTML}{800000}
\usepackage[a4paper,margin=28mm,headheight=15pt]{geometry}
\usepackage[colorlinks=true,linkcolor=linkmaroon,
            citecolor=linkmaroon,urlcolor=linkmaroon]{hyperref}
\usepackage[nameinlink,capitalise]{cleveref}

\definecolor{ink}{HTML}{18212B}
\color{ink}

\newtheorem{theorem}{Theorem}[section]
\newtheorem{lemma}[theorem]{Lemma}
\newtheorem{proposition}[theorem]{Proposition}
\newtheorem{corollary}[theorem]{Corollary}
\newtheorem{fact}[theorem]{Fact}
\theoremstyle{definition}
\newtheorem{definition}[theorem]{Definition}

\theoremstyle{remark}
\newtheorem{remark}[theorem]{Remark}

\newcommand*{\dif}{\mathop{}\!\mathrm{d}}
\newcommand*{\GainGap}{\operatorname{GainGap}}

\setlist[itemize]{leftmargin=1.5em,itemsep=0.25em,topsep=0.4em}
\setlist[enumerate]{leftmargin=1.7em,itemsep=0.25em,topsep=0.4em}

\title{Sybil-Proofness of Generalized Adaptive Clinching and Implications for  
Blockchain Mechanism Design\thanks{Author order is randomized with the
\href{https://www.aeaweb.org/journals/policies/random-author-order/search?RandomAuthorsSearch\%5Bsearch\%5D=GeSW7Qj5fXbF}{archived entry}.
}}
\author{
Chengkai Zhang\thanks{Rutgers University. \texttt{zhangchengkai2002@gmail.com}.}
\and Yuhao Li\thanks{Columbia University and Carnegie Mellon University. \texttt{yuhaoli@cs.columbia.edu}.}
\and Elaine Shi\thanks{Carnegie Mellon University. \texttt{elaineshi@cmu.edu}.}
\and Mengqian Zhang\thanks{Carnegie Mellon University. \texttt{mengqianzhang@cmu.edu}.}
\and Lirong Xia\thanks{Rutgers University. \texttt{lirong.xia@rutgers.edu}.}
\and Qishen Han\thanks{Rutgers University. \texttt{Hnickc2017@gmail.com}.}
}
\date{}

\begin{document}

\begin{titlepage}

\maketitle

\begin{abstract}
The adaptive clinching auction of Dobzinski, Lavi, and Nisan (FOCS'08) is a cornerstone of mechanism design with budget-constrained bidders, and it has since inspired a broad range of extensions to substantially richer market settings. Adaptive clinching enjoys compelling guarantees: it is Pareto optimal, and truthful bidding is a dominant strategy whenever budgets cannot be overreported. In many real-world environments, however, identities are cheap to create, making resilience to Sybil attacks an additional requirement. Whether the guarantees of adaptive clinching survive when a bidder can split its budget across fake identities has remained open. We introduce a generalized framework for adaptive clinching and prove that it satisfies Sybil-robust incentive compatibility: a bidder cannot benefit from splitting into any finite number of identities with arbitrary value reports, provided their aggregate reported budget does not exceed the bidder's enforceable true budget. Sybil resilience thus joins the list of ideal properties of adaptive clinching.

Our work is motivated by practical mechanism-design problems in blockchain-based decentralized finance, where two-asset automated market makers (AMMs) are among the most widely used applications. We design a new \emph{AMM adaptive clinching mechanism} that extends classical adaptive clinching in two ways: it supports 
a two-sided market and protects against two-sided strategies, and it allows the supply function to depend on the auction state. Its one-sided special case is captured by our generalized adaptive clinching framework, and we further extend our proof techniques to show that the full two-sided mechanism simultaneously satisfies Sybil-robust incentive compatibility and a relaxed notion of Pareto optimality. In light of known impossibility results, together with a new impossibility proved in this paper, we argue that these guarantees are essentially the best possible. More broadly, our work reveals new connections between classical algorithmic game theory and modern decentralized blockchains.
\end{abstract}
\thispagestyle{empty}
\end{titlepage}

\tableofcontents

\newpage

\section{Introduction}
\label{sec:intro}

Automated Market Makers (AMMs) are arguably the most widely used application in
decentralized finance (DeFi). An AMM is a smart contract that stands ready to
trade with anyone: rather than matching buyers against sellers through an order
book, it holds an inventory of assets and quotes an exchange rate determined by
a publicly known pricing curve. By far the most common deployment is the
two-asset AMM, exemplified by Uniswap~\cite{uniswap}.

In a typical two-asset AMM, users trade one of two assets, denoted $X$ and $Y$,
against a smart contract, which we henceforth also refer to as the {\it pool}.
Initially, the state of the pool is denoted by $(x_0,y_0)$, meaning that the
pool holds $x_0$ units of $X$ and $y_0$ units of $Y$. The exchange rate for
each trade is determined by a predefined pricing curve. For example, the widely
adopted constant-product rule requires
\[
    (x_0-x)(y_0+y)=C
\]
for some positive constant $C$. Thus, starting from state $(x_0,y_0)$, a user
who contributes $y$ units of $Y$ to the pool receives $x$ units of $X$.
Following the trade, the pool transitions to state $(x_0-x,y_0+y)$, preserving
the invariant that the product of the reserves remains equal to $C$.

\paragraph{Widespread incentive attacks.}
Because the price a user receives depends on the pool state when its trade
executes, and that state is moved by every preceding trade in the same block,
an AMM is acutely sensitive to who controls the contents and the ordering of a
block. A block builder can exploit this directly. Suppose a user submits an
order to buy $X$ with $Y$. Immediately {\it before} that order, the builder
buys $X$ cheaply for itself, driving the price up to the user's declared
{\it slippage tolerance}---the worst price at which it is still willing to
trade. The user's order then executes at that price and pushes the price
higher still, and the builder sells immediately {\it after} it. The builder
has thus bought low and sold high within a single block, at essentially no
risk, since it sees the user's order before committing to either leg. The two
legs are known respectively as {\it front-running} and {\it back-running}, and
the pair as a {\it sandwich attack}; a substantial body of
work~\cite{highfreq-dex,park2023conceptual,theorymev} documents how prevalent
and how profitable such attacks are in practice.

Such incentive attacks are commonly referred to as Miner (or Maximum) Extractable Value
(MEV), and the harm they do extends well beyond the users who are sandwiched.
Because the profit accrues to whoever controls the contents and ordering of a
block, MEV fosters opaque off-chain markets in which block builders, traders,
and specialized searchers compete to identify and exploit these opportunities,
and it rewards most those participants with the best access to order flow and
the deepest capital. These forces push the ecosystem steadily toward
centralization. Indeed, a recent measurement study by Yang et
al.~\cite{buildercentral} found that the top two block builders produce more
than 85\% of Ethereum blocks, severely undermining the decentralized vision of
blockchain systems.

\paragraph{A theory on AMM mechanism design.}
Motivated by these attacks, a line of
work~\cite{ammmechdesign,ammtrilemma,otter} has formulated AMM as a
mechanism-design problem. In this formulation, the AMM mechanism receives a
{\it batch} of bids (typically all trade requests contained in one
block). It then decides, for every bid, how much of each asset that bid sends and
receives. Each bid comes in one of two directions. A {\it sell-$Y$} bid
specifies a value $v$ and a budget $Q$, meaning that it wants to buy $X$ using
at most $Q$ units of $Y$, and that it values each unit of $X$ it acquires at
$v$ units of $Y$. 
A {\it sell-$X$} bid is
defined symmetrically. 

Two structural features distinguish the problem from traditional
auctions. First, the market is {\it two-sided}: a batch may contain both
sell-$X$ and sell-$Y$ bids, which can trade with each other as well as with
the pool. Second, the mechanism must {\it exactly settle}: every token paid
out is funded either by an opposite-side bid or by the pool, no outside
subsidy is used, no token is discarded, and the end state of the pool lands
back on the prescribed pricing curve. 

Within this formulation, one may ask for a {\it dream} AMM mechanism, namely
one that simultaneously achieves the following two desiderata.
\begin{enumerate}[leftmargin=6mm,itemsep=1pt]
\item {\bf Sybil-robust incentive compatibility.} The mechanism should encourage truthful
bidding by an individual trader, and it should moreover be strategy-proof
against a block builder who may itself be a trader with intrinsic value and
demand. Concretely, a trader or builder can misreport both its value and its
budget, and can inject any finite number of fake bids under pseudonymous
identities, carrying arbitrary values and budgets. 
Further, a builder can arbitrarily sequence the bids contained in the batch. 
We refer to the resulting incentive notion as {Sybil-robust incentive compatibility}. 
We stress that Sybil-robustness is critical in blockchain environments
since creating arbitrary pseudonyms is essentially free. 

\item {\bf Pareto optimality.} The outcome should be Pareto optimal: no incremental
trade among a subset of traders can make some of them better off without making
any of them worse off, and no trader can improve its own utility by trading
further with the pool.
\end{enumerate}

Throughout the paper, we assume that a strategic builder cannot censor other
traders' bids, meaning that every submitted bid must be included in the
mechanism's input. This matches the usual assumption in classical mechanism
design: a strategic participant may manipulate its own reports, or even inject fake bids (Sybil attacks), but cannot
suppress the reports submitted by other participants. In practice, this is ensured in blockchain systems when the underlying consensus
protocol provides censorship resilience~\cite{latencyCR,focil,mcp}, a
guarantee that major blockchain projects such as Ethereum are committed to
offering at the consensus layer~\cite{focil}.
Such an assumption is also unavoidable: it is known that once censorship is
admissible, 
Sybil-robust incentive compatibility is unattainable in AMM mechanism design
even without imposing any efficiency requirement~\cite{otter}.

\paragraph{Impossibility of a dream AMM mechanism.}
Unfortunately, such a dream AMM mechanism is unattainable. Shi et
al.~\cite{otter} and Li et al.~\cite{ammtrilemma} show that Pareto optimality
and Sybil-robust incentive compatibility are mutually incompatible for
exact-settlement AMM mechanisms.
Further, this impossibility holds in a very strong sense, 
even when the strategy space consists only 
of single-identity valuation misreports, and even 
if the mechanism additionally has access to unforgeable auxiliary information
attached to each bid (e.g., trusted timestamps 
on transactions determined by a decentralized
sequencer~\cite{decentral-seq,orderfair00,orderfair01,orderfair02,orderfair03,orderfair04}).

\subsection{Our Results and Contributions}

In this paper we identify two natural relaxations of the dream specification
under which feasibility is restored, and we present an AMM adaptive clinching mechanism that attains it.
Both relaxations are needed: neither one on its own restores feasibility, as we
explain later in this section, after stating our main feasibility result.

\begin{enumerate}[leftmargin=7mm,itemsep=1pt]
\item {\it No aggregate budget overreporting.}
Our first relaxation restricts the strategy space. A trader may still split
itself into arbitrarily many pseudonymous identities and report arbitrary
values, but the budgets that its identities declare {\it in its own trading
direction} must sum to at most its true budget. No restriction whatsoever is
placed on the budgets that it declares in the {\it opposite} direction: a
trader who genuinely wishes to sell $Y$ may inject sell-$X$ identities carrying
arbitrary finite budgets.

This assumption is well motivated, because a declared budget must be backed by
escrowed funds, %
and that escrow cannot be replenished
out of a sibling identity's proceeds during the batch. A trader whose
identities collectively declare more than it owns would therefore have to fund
the excess up front, which by definition it cannot do. In practice a trader is
of course escrow-limited on the opposite side as well, but our guarantees do
not need this: they hold even against a deviator that
can escrow arbitrarily large opposite-direction budgets.

\item {\it Weak Pareto optimality.}
Our second relaxation weakens the efficiency requirement to {\it weak Pareto
optimality for eligible bids}, a notion suggested by Li et
al.~\cite{ammtrilemma} and Shi et al.~\cite{otter} for overcoming impossibilities
pertaining to full Pareto optimality. A bid is {\it eligible} if it is willing to trade at
the pool's initial spot price. We require that, under truthful reporting, (a)
no subset of eligible traders can re-trade among themselves, subject to their
post-batch holdings, in a way that Pareto improves their joint utilities; and
(b) no eligible trader can strictly improve its utility by making an additional trade
with the pool {\it in the direction of its own type}---a trader who came to
sell $Y$ may only sell more $Y$.

Intuitively, the impossibility results of Li et al.~\cite{ammtrilemma}
and Shi et al.~\cite{otter} suggest that incentive compatibility is easier to 
achieve if the mechanism can raise the price as competition for the pool's output intensifies, and close the batch at a marginal price exceeding the values of the traders it has served --- 
provided that the average price each trader actually pays remains below its value. Weak Pareto optimality opens up our design space in this direction, 
by not counting such reverse-direction trades with the pool as Pareto improvements.
\end{enumerate}

\paragraph{Our main result.}
Under these two relaxations, we obtain the following.

\begin{theorem}[Sybil-robust and weakly Pareto optimal AMM]
\label{thm:intro-main}
Assume that a strategic trader cannot overreport its budget in aggregate
across the Sybil identities it submits in its own trading direction. Then
there exists an AMM mechanism that is individually rational, exactly settles
against the pool, and simultaneously satisfies Sybil-robust incentive
compatibility and weak Pareto optimality for eligible bids.
\end{theorem}

As argued below, both relaxations are necessary to  
obtain \Cref{thm:intro-main}, 
thus the combination of properties we achieve is essentially {\it best possible}. 
Moreover, as discussed later in this section,  
the techniques used to achieve \Cref{thm:intro-main}
reveal exciting connections to classical 
algorithmic game theory.

\paragraph{Necessity of both relaxations.}
Both relaxations are necessary.  If one insists on full Pareto optimality,
impossibility persists even when aggregate same-direction budget overreporting
is prohibited~\cite{otter}.\footnote{More precisely, the proof of Theorem~26
in Shi et al.~\cite{otter} uses no deviation that overreports a budget.  The
only required adjustment is to replace their implication from incentive
compatibility to no free lunch with
\Cref{fact:sybil-ic-implies-no-free-lunch}, which establishes the corresponding
implication for our restricted strategy space.}
Conversely, 
we prove that 
if aggregate same-direction budget overreporting is permitted,
impossibility persists 
even after relaxing full Pareto optimality to our weak notion, 
as stated in 
the following theorem:

\begin{theorem}[Impossibility when overreporting budget is permitted]
Suppose that 
traders can arbitrarily misreport their values and budgets. 
Then, 
no deterministic AMM mechanism can simultaneously satisfy
Sybil-robust incentive compatibility 
and weak Pareto optimality.
\label{thm:intro-imp}
\end{theorem}

In fact, this impossibility holds in a very strong sense.
Its proof needs only deviations that replace a single sell-\(Y\) report with
another, and does not rely on Sybil identities or opposite-direction
orders.  Moreover, this impossibility holds even
when the mechanism has access to unforgeable auxiliary information
attached to each bid.

Zhang et al.~\cite{ammchar} prove a similar impossibility. In comparison with our
\Cref{thm:intro-imp}, their impossibility holds for a much weaker efficiency notion,
but it critically relies on the builder's capabilities to reorder
bids within a batch, and it ceases to hold 
if the underlying consensus can provide 
trusted timestamps 
on the submitted bids, e.g., generated by 
a decentralized 
sequencer~\cite{decentral-seq,orderfair00,orderfair01,orderfair02,orderfair03,orderfair04} ---
indeed, a feasibility result~\cite{ammmechdesign}
is known in this setting satisfying their weaker efficiency notion.

\paragraph{Connections to classical algorithmic game theory.}
To obtain \Cref{thm:intro-main}, we generalize the adaptive clinching auction
of Dobzinski, Lavi, and Nisan~\cite{adaptiveclinch}, the classical mechanism
for multi-unit auctions with hard budget limits. In the classical setting, a
{\it fixed} divisible supply is offered to budget-constrained bidders by an ascending
price clock, and a bidder {\it clinches} a unit precisely when the residual
demand of the other bidders no longer covers the remaining supply.

The challenge in our setting is that an AMM has no fixed, exogenous supply:
its pricing curve makes the supply available for trade depend on the
mechanism's current state.  Our central observation is that adaptive clinching
extends to this setting by replacing fixed supply with a {\it carefully chosen
state-dependent target}.

Our work reveals exciting connections between classical algorithmic game
theory (AGT) and modern blockchain mechanism design.  On one hand,
our results give a compelling demonstration of how  
classical AGT ideas contribute to the formation of a much needed scientific foundation for
incentive-compatible blockchain markets.  
On the other hand, the proof techniques
developed for our AMM setting advance the state-of-the-art
understanding of classical AGT --- in particular, 
we show that the fixed-supply adaptive clinching auction of Dobzinski, Lavi, and
Nisan~\cite{adaptiveclinch} satisfies Sybil proofness,
as stated in the following theorem.

\begin{theorem}[Sybil-proofness of fixed-supply adaptive clinching]
\label{thm:intro-ord-adaptive-clinch}
In the adaptive clinching auction for a fixed divisible supply, truthful
participation as a single identity weakly dominates every strategy that
replaces the bidder with finitely many identities having arbitrary reported
values and aggregate declared budget at most the bidder's true budget.
\end{theorem}

\Cref{thm:intro-ord-adaptive-clinch} strengthens the incentive guarantees
previously established for adaptive clinching with fixed divisible
supply~\cite{adaptiveclinch,budgetelicitation}.  The previous analyses allow value
misreports and, in the case of Bhattacharya et al.~\cite{budgetelicitation}, budget under-reporting, but
do not prove any guarantees in the presence 
of pseudonymous splitting.

The fixed-supply auction is not obtained directly as a special case of our AMM
mechanism.  Instead, our proofs introduce an abstract target-based clinching
framework that embodies both mechanisms as special cases.  We establish
Sybil-robust incentive compatibility at this abstract level, thus extending the
guarantees to any auction that falls within this target-based clinching framework.  
This new framework, which can be viewed as a generalization
of the ordinary adaptive clinching, 
may be of independent interest. For example, to demonstrate its usefulness, 
we show that 
the framework also captures 
the adaptive clinching mechanism with exogenous
online supply introduced by Goel, Mirrokni, and Paes Leme~\cite{clinchonline}, yielding
the following consequence: 

\begin{theorem}[Sybil-proofness with exogenous online supply]
\label{thm:intro-online-adaptive-clinch}
At every positive-supply step of the adaptive clinching mechanism with an
exogenous online supply path, truthful participation as a single identity
weakly dominates, in cumulative utility, every strategy that replaces the
bidder with finitely many identities having arbitrary reported values and
aggregate declared budget at most the bidder's true budget.
\end{theorem}

\subsection{Additional Related Work}

\paragraph{AMM mechanism design.}
Our study builds directly on the AMM mechanism-design framework introduced by
Chan et al.~\cite{ammmechdesign}.  They also construct a Sybil-robust
incentive-compatible mechanism, but their result has two limitations.  First,
their mechanism does not achieve weak Pareto optimality (wPO).  It instead
satisfies weak local efficiency (wLE) for eligible users, which corresponds to
only the pool-trade clause of our wPO definition.  \Cref{thm:intro-imp} shows
that this loss of efficiency is unavoidable in their strategy space, which
permits a trader to overreport its aggregate budget.
Second, they require an overly stringent assumption, that the consensus layer assigns a
trustworthy arrival time to every bid.  Although a line of
work~\cite{decentral-seq,orderfair00,orderfair01,orderfair02,orderfair03,orderfair04}
has explored how to provide such a trusted timestamp, 
existing schemes are known to be manipulatable, e.g., by a client who deliberately
creates Condorcet cycles~\cite{condorcetattack}.

Subsequent work~\cite{ammtrilemma,otter} further delineates the feasibility
frontier in this model.  Li et al.~\cite{ammtrilemma} prove a trilemma among
incentive compatibility, weak local efficiency, and uniform pricing: any two
of these properties can be achieved, but no mechanism can satisfy all three.
Thus, contrary to a common intuition, clearing a batch at a uniform price does
not by itself yield strong MEV resilience while preserving even weak local
efficiency.  

Li et al.~\cite{ammtrilemma} and Shi et al.~\cite{otter} also show
that full Pareto optimality is fundamentally incompatible with Sybil-robust incentive
compatibility. 
Shi et al.~\cite{otter} circumvent this barrier by relaxing exact settlement
to permit ``burning''.
Their model
allows the pool to release more units than are allocated to traders, with the
surplus diverted for the benefit of the broader community.  For example, the
surplus can subsidize transaction fees, reward liquidity providers, or be
returned to the pool to reduce the price impact of future trades.  Under this
relaxation, they construct a simple and computationally efficient ``dream''
AMM that achieves both full Pareto optimality and Sybil-robust incentive
compatibility.

\paragraph{Sybil proofness in mechanism design.}
Sybil-proofness, also known as false-name-proofness, has been studied in a
variety of settings, including voting~\cite{anonproofvoting,falsenameweightedvoting},
resource exchange~\cite{sybilresourceexchange}, and combinatorial
auctions~\cite{falsename00,falsename01,falsename02,falsenameworstcase}.
A recurring theme is that auctions that are truthful when identities cannot
be faked may cease to be truthful once Sybil attacks are possible; most
notably, the VCG mechanism is not false-name-proof in general~\cite{falsename00,falsename02,sybilvcg}.
To the best of our knowledge, the Sybil resilience of adaptive-clinching-style
auctions for divisible goods under budget constraints has not been analyzed
before.  The closest related works are those of Iwasaki
et al.~\cite{falsenameascending} and Sakurai et al.~\cite{falsenamebudget}.
Both works 
aim to design Sybil-resilient auctions for indivisible goods. 
 Iwasaki
et al.~\cite{falsenameascending}
focuses ona setting without budget constraints, and 
Sakurai et al.~\cite{falsenamebudget} considers
non-quasi-linear utilities which include budget constraints.
Since their settings and auction mechanics 
are fundamentally different from ours, their Sybil-proofness
arguments do not carry over to our setting. 
Moreover, both Iwasaki
et al.~\cite{falsenameascending} and Sakurai et al.~\cite{falsenamebudget}
sacrifice efficiency  
to obtain Sybil resilience.  
In fact, 
in the setting of
Iwasaki et al., Pareto efficiency is known to be incompatible with
Sybil-proofness~\cite{falsename02}. 
Thus Iwasaki et al.\ evaluate efficiency only empirically, through
simulations, and Sakurai et al.\ leave the analysis of efficiency loss to
future work.  In contrast, we identify meaningful relaxations under which we
provide provable efficiency guarantees.

\paragraph{Adaptive clinching and extensions.}
The original work by Dobzinsi et al.~\cite{adaptiveclinch} 
that proposed the adaptive clinching auction assumed that the budgets are public (i.e., unforgeable).
The subsequent work of Bhattacharya et al.~\cite{budgetelicitation} show
that, for a divisible good, underreporting a budget is never profitable, so
adaptive clinching is incentive compatible
as long as budgets cannot be overreported.
It is also known that with private budgets, 
simultaneously achieving incentive compatibility 
and Pareto optimality is impossible~\cite{adaptiveclinch}.
Our feasibility and infeasibility pair (\Cref{thm:intro-main} and \Cref{thm:intro-imp})
can thus be viewed as  
AMM analogue of classical budget-constrained auctions of divisible goods.

Adaptive clinching~\cite{adaptiveclinch} has been extended to polymatroidal environments, including
the AdWords polytope for sponsored search with multiple keywords and
slots~\cite{polyhedralclinch}, as well as to other sponsored-search and
combinatorial settings~\cite{clinchkeywords,clinchcombinatorial},
heterogeneous items~\cite{clinchheterogeneous}, two-sided
markets~\cite{clinchtwosided}, online supply~\cite{clinchonline}, and payment
constraints more general than hard budgets~\cite{clinchbeyondbudgets}.  Most
of these works assume public budgets, and none considers Sybil identities.
More fundamentally, in all of them the feasible allocations are fixed
exogenously or, with online supply, revealed independently of the bids.  In an
AMM, the supply is instead created by the very payments the auction collects,
along the pool's cost curve.  Our target-based framework captures this
endogenous supply, and it recovers fixed and online supply as constant targets
(\Cref{app:constant-targets}).  Our two-sided setting also differs from that
of Hirai and Sato~\cite{clinchtwosided}, whose sellers are assumed to be
truthful and are paid at least their reserve prices.  In contrast, our pool
must receive exactly the amount that returns it to its curve, and traders in
both directions are strategic.

Even with public budgets, no deterministic mechanism attains the guarantees of
adaptive clinching once identical items may be
complements~\cite{publicbudgetimpossibility},
marginal values decrease~\cite{polyhedralclinch,clinchkeywords}, or items are
heterogeneous with multi-dimensional valuations~\cite{clinchheterogeneous}.
Another line of work replaces Pareto optimality with approximations of liquid
welfare: Dobzinski and Paes Leme~\cite{liquidwelfare} show that adaptive
clinching 2-approximates it when budgets are public, and Lu and
Xiao~\cite{liquidwelfareimproved} improve the ratio to $(1+\sqrt{5})/2$ with a
different auction.  Devanur et al.~\cite{priorfreebudget} use clinching as a
building block for prior-free revenue and welfare guarantees when agents share
a common budget.  We instead retain an exact, albeit weaker, Pareto guarantee.

\paragraph{Other related work.}
Ferreira and Parkes~\cite{credible-ex} suggest having the consensus
layer enforce verifiable, application-specific sequencing rules that limit how
a block builder may order transactions, and they give such rules for AMMs
with weak but provable guarantees against MEV extraction.
Their approach has several limitations.  First, it places
application-specific logic in the consensus layer, a shared infrastructure
layer that ought to remain application agnostic.  Second, their framework is overly
restrictive and requires
transactions to be executed one by one in the order in which they appear in
the block, thus leading to strong impossibilities. In particular, risk-free
profit cannot be eliminated in their model.

Among other work on AMMs, Milionis, Moallemi, and
Roughgarden~\cite{myersonian} characterize profit-maximizing demand curves for
a market maker among those that elicit truthful reports from traders.  Since
traders in their model trade directly with the market maker, strategic
ordering and attacks such as front-running and back-running fall outside its
scope.  Bartoletti et al.~\cite{maximizemev} compute a block builder's optimal
MEV-extraction strategy --- their goal is to characterize MEV extraction rather than
to prevent it.

Batch trading has also been studied as an approach to mitigating MEV.
Canidio and Fritsch~\cite{batchamm00,batchamm01} propose a function-maximizing AMM and analyze the elimination of arbitrage profits under idealized perfect competition among arbitrageurs sharing a common external reference price. CoW Protocol~\cite{cow} deploys batch trading through competition among solvers. Canidio and Henneke~\cite{fairbatchauction} propose fair combinatorial auctions for trade intents, studying the trade-off between execution fairness and total value returned to traders under strategic solver bidding. Ramseyer et al.~\cite{ramseyer2023augmenting} study equilibrium computation and design trade-offs for batch exchanges incorporating CFMMs. These works do not establish incentive compatibility for finite populations of strategic traders, which is our focus. Within the AMM mechanism-design framework discussed above, Li et al.~\cite{ammtrilemma} show that uniform pricing cannot be combined with both incentive compatibility and weak local efficiency.

In a smart-contract implementation, the escrow that prevents budget
overreporting must be protected against flash loans~\cite{attackdefi}, which
would let a trader 
borrow funds and repay the loan within the same transaction. 
 The attack can be avoided by splitting
execution into two phases: traders escrow their budgets in one block, and the
batch is settled in a later block, so that no flash loan can span both.

\section{Definitions}
\label{sec:model}

\subsection{AMM Reserve and Cost Curves}
\label{subsec:pricing-curve}

There are two divisible assets, \(X\) and \(Y\).  The AMM pool has initial
reserves \((x_0,y_0)\in\mathbb R_{++}^2\).  Its reserve curve is
\begin{equation*}
 \mathcal{R}_{\varphi}:=\{(x,\varphi(x)):x>0\}.
\end{equation*}
During a batch, every pool trade settles on this reserve curve.

\begin{definition}[Admissible AMM]
\label{def:two-sided-amm}
The AMM is \emph{admissible} if
\(\varphi:(0,\infty)\to(0,\infty)\) is continuous, strictly decreasing, and
convex, satisfies
\begin{equation*}
 \lim_{x\downarrow0}\varphi(x)=\infty,
 \qquad
 \lim_{x\uparrow\infty}\varphi(x)=0,
\end{equation*}
and is differentiable at \(x_0\), where \(y_0=\varphi(x_0)\).  The {\it spot price}
of one unit of \(X\), measured in units of \(Y\), is
\begin{equation*}
 p_0:=-\varphi'(x_0)>0.
\end{equation*}
\end{definition}

For a net release \(z<x_0\) of asset \(X\), define the signed pool cost
\begin{equation*}
 C(z):=\varphi(x_0-z)-y_0.
\end{equation*}
Here \(z>0\) means that the pool releases \(z\) units of \(X\), whereas
\(z<0\) means that it absorbs \(-z\) units.  Correspondingly, \(C(z)\) is
the net amount of \(Y\) received by the pool and may be negative.  The
function \(C\) is continuous, strictly increasing, and convex, with
\begin{equation*}
 C(0)=0,
 \qquad
 \lim_{z\uparrow x_0}C(z)=+\infty,
 \qquad
 \lim_{z\to-\infty}C(z)=-y_0,
 \qquad
 C'(0)=p_0.
\end{equation*}
Thus sending \(y\ge0\) units of \(Y\) into the pool releases
\(C^{-1}(y)\in[0,x_0)\) units of \(X\), while sending \(x\ge0\) units of
\(X\) into the pool releases \(-C(-x)\in[0,y_0)\) units of \(Y\).  These two
compensation maps are continuous, strictly increasing, and concave.  Their
initial marginal exchange rates are \(1/p_0\) and \(p_0\), respectively.

\paragraph{Constant product as a special case.}
If \(\varphi(x)=x_0y_0/x\), then
\begin{equation*}
 C(z)=\frac{y_0z}{x_0-z},
 \qquad
 C^{-1}(y)=\frac{x_0y}{y_0+y},
 \qquad
 -C(-x)=\frac{y_0x}{x_0+x}.
\end{equation*}
Consequently \(p_0=y_0/x_0\),
the first compensation map has initial derivative \(x_0/y_0=1/p_0\), and
the second has initial derivative \(y_0/x_0=p_0\).
Thus constant product satisfies \cref{def:two-sided-amm}.

\subsection{AMM Mechanism}
\label{subsec:amm-mechanism-model}

A physical trader has one of two types, according to the asset it wishes to
sell.

\paragraph{Sell-\(Y\) orders.}
A sell-\(Y\) trader \(k\) has true type
\begin{equation*}
 (V_k,Q_k)\in\mathbb R_{\ge0}^2,
\end{equation*}
where \(V_k\) is its value for one unit of \(X\), measured in units of \(Y\),
and \(Q_k\) is its true \(Y\)-budget.  A sell-\(Y\) identity \(i\) reports
\begin{equation*}
 (v_i,q_i)\in\mathbb R_{\ge0}^2,
\end{equation*}
where \(v_i\) has the same units as \(V_k\) and \(q_i\) is its declared
\(Y\)-budget.  An outcome for this identity consists of \(S_{Y,i}\) units of
\(Y\) sold and \(R_{X,i}\) units of \(X\) received, with
\begin{equation}
 0\le S_{Y,i}\le q_i,\qquad R_{X,i}\ge0.
 \label{eq:sell-y-outcome}
\end{equation}
Its surplus evaluated at the report is \(v_iR_{X,i}-S_{Y,i}\).  A single
sell-\(Y\) report is truthful when \((v_i,q_i)=(V_k,Q_k)\).

All values in this paper are denominated in units of \(Y\) per unit of \(X\).
Thus \(V_k\) and \(v_i\) are reciprocals of a positive reservation value for
one unit of \(Y\) expressed in units of \(X\).  This convention puts both
types of values on the same price scale.

\paragraph{Sell-\(X\) orders.}
A sell-\(X\) trader \(k\) has true type
\begin{equation*}
 (U_k,R_k)\in\mathbb R_{\ge0}^2,
\end{equation*}
where \(U_k\) is its reservation value for one unit of \(X\), measured in
units of \(Y\), and \(R_k\) is its true \(X\)-budget.  A sell-\(X\) identity
\(i\) reports
\begin{equation*}
 (u_i,r_i)\in\mathbb R_{\ge0}^2.
\end{equation*}
It sells \(S_{X,i}\) units of \(X\) and receives \(R_{Y,i}\) units of \(Y\),
where
\begin{equation}
 0\le S_{X,i}\le r_i,\qquad R_{Y,i}\ge0.
 \label{eq:sell-x-outcome}
\end{equation}
Its surplus evaluated at the report is \(R_{Y,i}-u_iS_{X,i}\).  A single
sell-\(X\) report is truthful when \((u_i,r_i)=(U_k,R_k)\).

A sell-\(Y\) order is \emph{eligible} if \(v_i\ge p_0\), and a sell-\(X\)
order is eligible if \(u_i\le p_0\).  
Unless stated otherwise, the truthful types under consideration
satisfy \(V_k\ge p_0\) for sell-\(Y\) and \(U_k\le p_0\) for sell-\(X\).
At equality with \(p_0\), either direction may be designated as the trader's
type because it is indifferent at the spot price.

\paragraph{AMM mechanism.}
Let \(\mathcal I_{\rm Y}\) and \(\mathcal I_{\rm X}\) be the finite sets of
reported sell-\(Y\) and sell-\(X\) identities.  An \emph{AMM mechanism}
maps the reported orders to an outcome
\begin{equation*}
 \bigl((S_{Y,i},R_{X,i})\bigr)_{i\in\mathcal I_{\rm Y}},
 \qquad
 \bigl((S_{X,i},R_{Y,i})\bigr)_{i\in\mathcal I_{\rm X}},
\end{equation*}
satisfying \eqref{eq:sell-y-outcome} and \eqref{eq:sell-x-outcome}.
The mechanism may match the two sides internally and trade the residual
imbalance with the AMM pool.

For an outcome, define the net flows into the pool by
\begin{equation}
 X_{\mathrm{pool}}
 :=\sum_{i\in\mathcal I_{\rm Y}}R_{X,i}
   -\sum_{i\in\mathcal I_{\rm X}}S_{X,i},
 \qquad
 Y_{\mathrm{pool}}
 :=\sum_{i\in\mathcal I_{\rm Y}}S_{Y,i}
   -\sum_{i\in\mathcal I_{\rm X}}R_{Y,i}.
 \label{eq:net-pool-flows}
\end{equation}
Positive values mean that the pool releases \(X\) and receives \(Y\);
negative values describe the reverse trade.  Trades matched internally
cancel from both expressions.

\begin{definition}[Exact-settlement feasibility]
\label{def:two-sided-exact-settlement}
An outcome is \emph{exact-settlement feasible} if it satisfies the
per-identity constraints \eqref{eq:sell-y-outcome} and
\eqref{eq:sell-x-outcome}, and the terminal pool reserves satisfy
\begin{equation*}
 (x_0-X_{\mathrm{pool}},y_0+Y_{\mathrm{pool}})
 \in\mathcal{R}_{\varphi}.
\end{equation*}
Equivalently, its net pool flows satisfy
\begin{equation*}
 X_{\mathrm{pool}}<x_0,\qquad
 Y_{\mathrm{pool}}=C(X_{\mathrm{pool}}).
\end{equation*}
Thus every unit paid to an identity is supplied by an opposite-side identity
or by the pool, and every unit delivered by an identity is received by an
opposite-side identity or by the pool.  The aggregate user flows are exactly
the negatives of the pool flows: no outside subsidy is used and no token is
discarded.  A mechanism is \emph{exact-settlement feasible} if it produces
such an outcome for every reported profile.
\end{definition}

\begin{definition}[Individual rationality]
\label{def:two-sided-ir}
A mechanism is \emph{individually rational} if, whenever every physical
trader submits its single truthful order, every sell-\(Y\) trader \(k\)
satisfies
\begin{equation*}
 0\le S_{Y,k}\le Q_k,\qquad
 V_kR_{X,k}-S_{Y,k}\ge0,
\end{equation*}
and every sell-\(X\) trader \(k\) satisfies
\begin{equation*}
 0\le S_{X,k}\le R_k,\qquad
 R_{Y,k}-U_kS_{X,k}\ge0.
\end{equation*}
Equivalently, these inequalities say that truthful participation gives every
trader finite, nonnegative true utility.  A trader that does not trade has
utility zero.
\end{definition}

Every declared budget is a hard per-identity trading limit.  It is escrowed
before the batch, cannot be transferred between identities during the batch,
and cannot be replenished from another identity's batch proceeds.  Transfers
and deliveries occur at final settlement.

\subsection{Strategy Space and Utility}
\label{subsec:strategy-utility}

\paragraph{Strategy space.}
We consider a strategy space where not only misreporting
value and budget is possible,
but also the creation of Sybil identities --- as long as 
reports in the same-direction of the trader's true type 
do not exceed the user's actual budget in aggregate. 
No cap is imposed on opposite-direction Sybil bids, i.e., 
the trader's opposite direction budget
may be arbitrarily large. 
The mechanism observes only the reported
orders; it does not know which identities are controlled by the same physical
trader.

More formally, 
fix a physical trader \(k\), and let
\(\mathcal I_{\rm Y}(k)\subseteq\mathcal I_{\rm Y}\) and
\(\mathcal I_{\rm X}(k)\subseteq\mathcal I_{\rm X}\) denote its sell-\(Y\)
and sell-\(X\) identities.

\begin{definition}[Admissible two-sided Sybil strategy]
\label{def:two-sided-strategy}
An admissible strategy of trader \(k\) is a pair of finite collections
\begin{equation*}
 \bigl((v_i,q_i)\bigr)_{i\in\mathcal I_{\rm Y}(k)},
 \qquad
 \bigl((u_i,r_i)\bigr)_{i\in\mathcal I_{\rm X}(k)},
\end{equation*}
with nonnegative reports, subject to the following direction-dependent
constraint.
\begin{enumerate}
\item If \(k\) has true sell-\(Y\) type \((V_k,Q_k)\), then
\begin{equation*}
 \sum_{i\in\mathcal I_{\rm Y}(k)}q_i\le Q_k.
\end{equation*}
The opposite-direction budgets \(r_i\) may be arbitrary finite nonnegative
numbers.
\item If \(k\) has true sell-\(X\) type \((U_k,R_k)\), then
\begin{equation*}
 \sum_{i\in\mathcal I_{\rm X}(k)}r_i\le R_k.
\end{equation*}
The opposite-direction budgets \(q_i\) may be arbitrary finite nonnegative
numbers.
\end{enumerate}
The reported values \(v_i\) and \(u_i\) may also be arbitrary finite
nonnegative numbers.
\end{definition}

\begin{remark}[On builder strategies]
\label{rmk:sequence}
We may view the builder as a special trader enpowered with one additional strategy:  
a builder can choose how to sequence the bids
within the batch.  We omit this sequencing choice from the formal strategy
space because our mechanism is \emph{order-insensitive}: it treats the batch
as an unordered set of reports, and its rules depend on reported values and
budgets but not on positions within the block.  Consequently, every ordering
of a fixed collection of reports produces the same allocations and payments.
Thus, our incentive guarantees (\cref{def:incentive-compatibility}), based on 
\cref{def:two-sided-strategy}, 
 already cover such sequencing strategies. 

A builder may have intrinsic value and demand, in which case we may model it as a
trader (with sequencing capability); 
or it may participate solely to seek risk-free profit.  
\Cref{fact:sybil-ic-implies-no-free-lunch} shows that our  
incentive notion (\cref{def:incentive-compatibility}) 
in fact protects against both types of
builders, without requiring a separate builder-specific condition 
to capture the latter category.
\end{remark}

\paragraph{Aggregate net flow and utility.}
Define trader \(k\)'s aggregate gross token flows by
\begin{align*}
 S_Y(k)&:=\sum_{i\in\mathcal I_{\rm Y}(k)}S_{Y,i},
 &R_X(k)&:=\sum_{i\in\mathcal I_{\rm Y}(k)}R_{X,i},
\\
 S_X(k)&:=\sum_{i\in\mathcal I_{\rm X}(k)}S_{X,i},
 &R_Y(k)&:=\sum_{i\in\mathcal I_{\rm X}(k)}R_{Y,i}.
\end{align*}
Its net changes in token holdings are
\begin{equation}
 \Delta_{X,k}:=R_X(k)-S_X(k),
 \qquad
 \Delta_{Y,k}:=R_Y(k)-S_Y(k).
 \label{eq:physical-net-flows}
\end{equation}
Thus \(-\Delta_{X,k}\) and \(-\Delta_{Y,k}\) are its net sales of \(X\) and
\(Y\), respectively.  Utility is evaluated from these aggregate net flows,
not by adding the identities' reported surpluses.

For a true sell-\(Y\) type \((V_k,Q_k)\), define
\begin{equation}
 \operatorname{Util}^{\rm Y}_k
 :=
 \begin{cases}
  -\infty,
  &\text{if }-\Delta_{Y,k}>Q_k
    \text{ or }
    \bigl(-\Delta_{Y,k}<0\text{ and }\Delta_{X,k}<0\bigr),\\
  V_k\Delta_{X,k}+\Delta_{Y,k},
  &\text{otherwise}.
 \end{cases}
 \label{eq:two-sided-physical-utility-y}
\end{equation}
The first catastrophic event means that the trader's net sale of \(Y\)
exceeds its true budget.  The second means that its aggregate trade is
strictly opposite to its type: it gains \(Y\) while losing \(X\).

For a true sell-\(X\) type \((U_k,R_k)\), define
\begin{equation}
 \operatorname{Util}^{\rm X}_k
 :=
 \begin{cases}
  -\infty,
  &\text{if }-\Delta_{X,k}>R_k
    \text{ or }
    \bigl(-\Delta_{X,k}<0\text{ and }\Delta_{Y,k}<0\bigr),\\
  \Delta_{Y,k}+U_k\Delta_{X,k},
  &\text{otherwise}.
 \end{cases}
 \label{eq:two-sided-physical-utility-x}
\end{equation}
Here the catastrophic events mean that the trader net-sells more \(X\) than
its true budget, or that it gains \(X\) while losing \(Y\).

These utility functions do not penalize a trader merely for receiving a net
gain in its budget token when it does not lose the other token.  Moreover,
the catastrophic budget tests concern net rather than gross sales: receipts
through opposite-side identities offset deliveries of the same token.

The truthful merged strategy of a sell-\(Y\) type \((V_k,Q_k)\) is the single
sell-\(Y\) order \((V_k,Q_k)\); its utility reduces to
\(V_kR_X(k)-S_Y(k)\).  The truthful merged strategy of a sell-\(X\) type
\((U_k,R_k)\) is the single sell-\(X\) order \((U_k,R_k)\); its utility
reduces to \(R_Y(k)-U_kS_X(k)\).

\subsection{Desired Properties}

\paragraph{Sybil-robust incentive compatibility.}
Sybil-robust incentive compatibility says that the truthful strategy
is profit-maximizing in light of
the strategy space defined earlier.

\begin{definition}[Sybil-robust incentive compatibility]
\label{def:incentive-compatibility}
A mechanism is \emph{Sybil-robust incentive compatible} if, for every
physical trader, every true type (including ineligible types), every fixed
strategy profile of the other traders, and every admissible strategy in
\cref{def:two-sided-strategy}, the true utility from the truthful merged
strategy is at least the true utility from the deviation.  The applicable
utility is \eqref{eq:two-sided-physical-utility-y} for a sell-\(Y\) type and
\eqref{eq:two-sided-physical-utility-x} for a sell-\(X\) type.
\end{definition}

Since the strategy space for a block builder 
 --- possibly having intrinsic value and demand --- is the same
as for an ordinary trader, 
\Cref{def:incentive-compatibility}
simultaneously captures
dominant strategy incentive compatibility (DSIC)
for individual traders, as well as strategy-proofness
for a block builder.

\paragraph{Weak Pareto optimality for eligible bids.}
Fix a completed, exact-settlement-feasible outcome under truthful reporting.
For each eligible user \(i\), define its post-batch holdings and value
coefficient by
\begin{equation}
 (H_i^X,H_i^Y,t_i):=
 \begin{cases}
  (R_{X,i},Q_i-S_{Y,i},V_i),&i\text{ is a sell-\(Y\) user},\\
  (R_i-S_{X,i},R_{Y,i},U_i),&i\text{ is a sell-\(X\) user}.
 \end{cases}
 \label{eq:post-batch-holdings}
\end{equation}
Thus a change \((\delta x_i,\delta y_i)\) in user \(i\)'s holdings changes
its utility by \(t_i\delta x_i+\delta y_i\).

\begin{definition}[Weak Pareto optimality]
\label{def:weak-pareto}
An exact-settlement-feasible two-sided outcome is \emph{weakly Pareto optimal
for eligible bids} if both of the following hold.
\begin{enumerate}
\item There is no subset \(\mathcal J\) of eligible users and no additional
internal trade satisfying
\begin{equation*}
 \sum_{i\in\mathcal J}\delta x_i=0,\qquad
 \sum_{i\in\mathcal J}\delta y_i=0,\qquad
 \delta x_i\ge-H_i^X,\qquad
 \delta y_i\ge-H_i^Y\quad(i\in\mathcal J)
\end{equation*}
that weakly increases \(t_i\delta x_i+\delta y_i\) for every
\(i\in\mathcal J\) and strictly increases it for at least one.
\item No eligible user can make a strictly utility-improving bilateral trade
with the pool in the same direction as that user's type, using its post-batch
holding of the token it sells.  A sell-\(Y\) user may tender additional \(Y\)
to the pool for \(X\), whereas a sell-\(X\) user may tender additional \(X\)
to the pool for \(Y\).
\end{enumerate}
The second condition deliberately excludes a reverse-direction pool trade.
\end{definition}

\paragraph{Sybil-robust incentive compatibility implies ``no free lunch''.}
In practice, a strategic builder may have no intrinsic value or demand and
participate solely to seek a risk-free profit.  By risk-free profit, we mean
a net gain in one asset without a net loss in the other---that is, a free
lunch.  We show below that Sybil-robust incentive compatibility
(\cref{def:incentive-compatibility}) rules out such profits, even when the
builder may submit identities with arbitrarily large finite budgets in either
asset.  We begin by formalizing this guarantee as the \emph{no-free-lunch}
property.

\begin{definition}[No free lunch]
\label{def:no-free-lunch}
For any reported profile and any subset of identities
\(\mathcal J\subseteq\mathcal I_{\rm Y}\cup\mathcal I_{\rm X}\), define
their aggregate net token gains by
\begin{align*}
 \Delta_X(\mathcal J)
 &:=\sum_{i\in\mathcal J\cap\mathcal I_{\rm Y}}R_{X,i}
   -\sum_{i\in\mathcal J\cap\mathcal I_{\rm X}}S_{X,i},\\
 \Delta_Y(\mathcal J)
 &:=\sum_{i\in\mathcal J\cap\mathcal I_{\rm X}}R_{Y,i}
   -\sum_{i\in\mathcal J\cap\mathcal I_{\rm Y}}S_{Y,i}.
\end{align*}
A mechanism satisfies \emph{no free lunch} if, for every such profile and
subset,
\begin{equation*}
 \Delta_X(\mathcal J)>0\ \Longrightarrow\ \Delta_Y(\mathcal J)<0,
 \qquad
 \Delta_Y(\mathcal J)>0\ \Longrightarrow\ \Delta_X(\mathcal J)<0.
\end{equation*}
Thus no subset of identities can gain strictly positively in one token
without losing the other token in aggregate.
\end{definition}

\begin{fact}[Sybil-robust IC implies no free lunch]
\label{fact:sybil-ic-implies-no-free-lunch}
Every mechanism satisfying Sybil-robust incentive compatibility
(\cref{def:incentive-compatibility}) satisfies no free lunch.
\end{fact}

\begin{proof}
Fix a reported profile and a subset \(\mathcal J\), and abbreviate its net
gains as \(\Delta_X,\Delta_Y\).  Keep all reports outside \(\mathcal J\)
fixed.  Since the mechanism does not observe ownership, we may interpret
\(\mathcal J\) as the identities of one physical trader.  Set
\begin{equation*}
 Q:=\sum_{i\in\mathcal J\cap\mathcal I_{\rm Y}}q_i,
 \qquad
 R:=\sum_{i\in\mathcal J\cap\mathcal I_{\rm X}}r_i.
\end{equation*}
By \cref{def:two-sided-strategy}, submitting these identities is admissible
for a sell-\(Y\) trader with budget \(Q\), or a sell-\(X\) trader with
budget \(R\), at any nonnegative true value.  Suppose towards a
contradiction that \(\Delta_X,\Delta_Y\ge0\), with at least one inequality
strict.  Neither catastrophic utility condition applies to these net gains.

If \(\Delta_Y>0\), consider a sell-\(Y\) trader of true type \((0,Q)\).
Its truthful merged utility is \(-S_Y^{\rm tr}\le0\), whereas submitting
\(\mathcal J\) gives utility \(\Delta_Y>0\).  This contradicts
\cref{def:incentive-compatibility}.

If \(\Delta_X>0\), for each \(u\ge0\) let \(s(u)\) and \(r(u)\) denote
the amounts of \(X\) sold and \(Y\) received by the single sell-\(X\)
report \((u,R)\), against the same fixed outside reports.  This report is
an admissible deviation for a true type \((0,R)\), and its utility for
that type is \(r(u)\), since \(0\le s(u)\le R\).  Incentive compatibility
for type \((0,R)\) therefore gives
\begin{equation*}
 r(u)\le r(0)=:K<\infty\qquad\text{for every }u\ge0.
\end{equation*}
Choose a finite \(U>K/\Delta_X\).  The truthful utility of type \((U,R)\)
is at most
\begin{equation*}
 r(U)-Us(U)\le r(U)\le K,
\end{equation*}
whereas its admissible strategy \(\mathcal J\) gives
\(\Delta_Y+U\Delta_X>K\).  This again contradicts
\cref{def:incentive-compatibility}, proving the claim.
\end{proof}

Shi et al.~\cite{otter} also show that their notion of incentive compatibility
implies no free lunch, paralleling \Cref{fact:sybil-ic-implies-no-free-lunch}.
Their result does not apply directly here, however, because their strategy
space permits budget overreporting, whereas ours caps a trader's aggregate
same-direction reports by its true budget.

\section{Warmup: The One-Sided AMM Clinching Mechanism}
\label{sec:mechanism}

As a warmup, we first present an AMM clinching mechanism
assuming a one-sided market.  
Specifically, 
throughout this section, every submitted identity is in the sell-\(Y\)
direction: it tenders \(Y\) and demands \(X\).  
The full, two-sided mechanism will be 
introduced later in \Cref{sec:two-sided-mechanism}.

\subsection{The One-Sided Model and Feasibility Requirements}

More specficially,
we specialize the model of \cref{subsec:amm-mechanism-model,subsec:strategy-utility}
to a finite set of sell-\(Y\) identities \(\mathcal I_{\rm Y}\), with
\(\mathcal I_{\rm X}=\varnothing\).  Thus identity \(i\) reports
\((v_i,q_i)\in\mathbb R_{\ge0}^2\), where \(v_i\) is its value for one unit
of \(X\), measured in units of \(Y\), and \(q_i\) is its \(Y\)-budget.  It
is eligible exactly when \(v_i\ge p_0\); an ineligible identity receives
zero.  As before, the mechanism observes neither true types nor which
identities have a common owner.

Write
\[
 x_i:=R_{X,i},\qquad y_i:=S_{Y,i}.
\]
The mechanism chooses \(x_i\ge0\) and \(0\le y_i\le q_i\): identity \(i\)
receives \(x_i\) units of \(X\), sells \(y_i\) units of \(Y\), and has
reported surplus \(v_ix_i-y_i\).

On a truthful profile, identify each physical trader with its unique identity
\(i\), write its eligible true type as
\((V_i,Q_i)\in[p_0,\infty)\times\mathbb R_{\ge0}\), and set
\((v_i,q_i)=(V_i,Q_i)\).  Its true utility is then \(V_ix_i-y_i\).
Indeed, \(y_i\le q_i=Q_i\), and the opposite-direction catastrophe in
\eqref{eq:two-sided-physical-utility-y} cannot occur because
\(\Delta_{X,i}=x_i\ge0\).  Same-direction Sybil strategies and their
utilities retain the definitions of
\cref{def:two-sided-strategy,subsec:strategy-utility}, specialized to
\(\mathcal I_{\rm X}=\varnothing\).

On the one-sided domain \(0\le z<x_0\), \(C(z)\) is the amount of \(Y\) the
pool must receive to release \(z\) units of \(X\).  Define its right marginal
ask by
\begin{equation*}
 a(z):=C'_+(z),
\end{equation*}
so, in particular, \(a(0)=p_0\).

For any reported profile, define the aggregate allocation and payment by
\begin{equation}
 X:=\sum_{i\in\mathcal I_{\rm Y}}x_i,
 \qquad
 Y:=\sum_{i\in\mathcal I_{\rm Y}}y_i.
 \label{eq:one-sided-aggregate-outcome}
\end{equation}
Because \(\mathcal I_{\rm X}=\varnothing\), the two-sided net pool flows
\eqref{eq:net-pool-flows} specialize to
\(X_{\mathrm{pool}}=X\) and \(Y_{\mathrm{pool}}=Y\).  Consequently,
\cref{def:two-sided-exact-settlement} requires precisely
\begin{equation}
 0\le X<x_0,
 \qquad
 Y=C(X).
 \label{eq:one-sided-exact-settlement}
\end{equation}
The terminal pool reserves are then
\begin{equation*}
 (x_0-X,y_0+Y)\in\mathcal{R}_{\varphi}.
\end{equation*}

Likewise, for a truthful one-sided bidder,
\(S_Y(i)=y_i\) and \(R_X(i)=x_i\).  The two-sided individual-rationality
requirement in \cref{def:two-sided-ir} therefore specializes to
\begin{equation*}
 0\le y_i\le Q_i,
 \qquad
 V_ix_i-y_i\ge0.
\end{equation*}
The mechanism below is designed to satisfy these specializations of the
two-sided requirements.  Allocations and payments are generated by an
ascending price clock \(p\), starting at the spot price \(p_0\).

\subsection{State, Chord Target, and Demand}

At clock price \(p\), identity \(i\) has cumulative allocation \(x_i\),
cumulative payment \(y_i\), and effective remaining budget \(b_i\).  While the
identity remains in the auction, \(b_i=q_i-y_i\).  When it exits without
spending all of its budget, the unspent amount is cancelled and \(b_i\) is set
to zero.  Together with the aggregate allocation \(X\) and payment \(Y\)
defined in \eqref{eq:one-sided-aggregate-outcome}, let
\begin{equation*}
 B:=\sum_{i\in\mathcal I_{\rm Y}}b_i
\end{equation*}
denote the aggregate effective remaining budget.

The mechanism starts from
\begin{equation}
 p=p_0,\qquad x_i=y_i=0,\qquad X=Y=0.
 \label{eq:initial-state}
\end{equation}
Ineligible reports are removed as specified above.  Every remaining identity
initially has \(b_i=q_i\).

After the target below is computed, all events at \(p_0\) are processed
before the clock rises.  In particular, if \(T^{\rm amm}(p_0,0,0)>0\), that
initial target is released as a supply lump, and identities with \(v_i=p_0\)
are treated as exiting identities.  For the AMM, a positive initial target
occurs precisely when \(C\) has a nontrivial initial linear piece of slope
\(p_0\).  Thus the general event order below also covers the initial state.

For any state satisfying \(Y\ge C(X)\), define
\begin{equation*}
 g_p(z):=C(z)-Y-p(z-X),
 \qquad z\in[X,x_0).
\end{equation*}
The \emph{AMM chord target} is
\begin{equation}
 T^{\rm amm}(p,X,Y)
 :=\max\{z\in[X,x_0):g_p(z)\le0\}.
 \label{eq:amm-chord-target}
\end{equation}
The maximum exists: the set is nonempty because \(g_p(X)\le0\), it is an
interval by convexity, and it is bounded away from \(x_0\) because
\(C(z)\to\infty\).  At its right endpoint,
\begin{equation*}
 C(T^{\rm amm})=Y+p(T^{\rm amm}-X),
 \qquad a(T^{\rm amm})\ge p.
\end{equation*}
The residual supply offered to the auction at the current price is
\begin{equation*}
 S:=T^{\rm amm}-X.
\end{equation*}
Geometrically, \(T^{\rm amm}\) is the rightmost intersection of \(C\) with the line of
slope \(p\) through \((X,Y)\).  Notice that a purchase of \(\delta\le S\) at
price \(p\) changes the state to
\begin{equation*}
 X^+=X+\delta,\qquad Y^+=Y+p\delta,
\end{equation*}
but leaves \(Y-pX\), and hence \(T^{\rm amm}\), unchanged.  The purchase
reduces residual supply from \(S\) to \(S-\delta\).

At the current price \(p\), define
\begin{equation}
 {\sf Active}:=\{i:v_i>p,\ b_i>0\},
 \qquad
 {\sf Exit}:=\{i:v_i=p,\ b_i>0\}.
 \label{eq:active-exit-sets}
\end{equation}
Thus \({\sf Active}\) is the set of active identities and \({\sf Exit}\) is
the set of exiting identities that have not yet been processed.  Until an
exiting identity is processed, its budget remains part of \(B\).  Write
\begin{equation*}
 D_i:=\frac{b_i}{p},
 \qquad
 D_{-i}:=\sum_{j\ne i:\,b_j>0}\frac{b_j}{p}
\end{equation*}
for its demand and the demand of the other unprocessed identities.  The
mechanism maintains the supply invariant
\begin{equation}
 S\le D_{-i}
 \qquad\text{for every }i\in{\sf Active}.
 \label{eq:supply-invariant}
\end{equation}
The set of clinching identities is
\begin{equation}
 {\sf Clinch}:=\{i\in{\sf Active}:S=D_{-i}\}.
 \label{eq:clinch-set}
\end{equation}

It is convenient to introduce the water level
\begin{equation*}
 h:=B-pS.
\end{equation*}
Since \(pD_{-i}=B-b_i\), the invariant and clinching condition are
equivalently
\begin{equation}
 b_i\le h
 \quad\text{for every }i\in{\sf Active},
 \qquad
 {\sf Clinch}=\{i\in{\sf Active}:b_i=h\}.
 \label{eq:water-invariant}
\end{equation}

\subsection{Continuous Evolution}
\label{subsec:continuous-dynamics}

We use the clock price itself as the time variable, so \(\dif p/\dif p=1\).
Consider an open interval containing no valuation exit, no discontinuity of
the target, no change in the clinching set, and no terminal event.  The clock
and state paths are absolutely continuous on such an interval, and the target
is always reevaluated from the algebraic identity
\begin{equation}
 T^{\rm amm}=T^{\rm amm}(p,X,Y),
 \qquad S=T^{\rm amm}-X,
 \qquad h=B-pS.
 \label{eq:continuous-state-identities}
\end{equation}

The continuous evolution is defined by the following primitive conditions.
An identity outside \({\sf Clinch}\) does not trade:
\begin{equation}
 \frac{\dif x_i}{\dif p}
 =\frac{\dif y_i}{\dif p}
 =\frac{\dif b_i}{\dif p}=0,
 \qquad i\notin{\sf Clinch}.
 \label{eq:nonclincher-primitive-dynamics}
\end{equation}
Every identity in \({\sf Clinch}\) receives a nonnegative allocation at the
current price, spends exactly that payment from its effective budget, and
remains tight at the water level:
\begin{equation}
 \frac{\dif x_i}{\dif p}\ge0,
 \qquad
 \frac{\dif y_i}{\dif p}
   =p\frac{\dif x_i}{\dif p},
 \qquad
 \frac{\dif b_i}{\dif p}
   =-\frac{\dif y_i}{\dif p},
 \qquad
 \frac{\dif}{\dif p}(b_i-h)=0,
 \quad i\in{\sf Clinch}.
 \label{eq:clincher-primitive-dynamics}
\end{equation}
Equations \eqref{eq:continuous-state-identities}--
\eqref{eq:clincher-primitive-dynamics} form a differential-algebraic
definition of the continuous auction path.  They say directly that trade is
charged at the clock price and that an identity which is just able to clinch
remains just able to clinch.  They apply without modification when \(S=0\):
zero residual supply is a boundary state of the same dynamics, not a separate
dynamic regime.  No trading-rate formula is assumed.  In
\cref{subsec:derived-dynamics} we prove that these conditions determine a
unique rate almost everywhere and derive its explicit form.  Upward target
jumps are handled by the event rules below.

\subsection{Fixed-Price Operations and Events}
\label{subsec:events}

For use in every event rule, define
\[
\operatorname{Buy}(i,\delta;p),
\qquad 0\le\delta\le\min\{S,b_i/p\},
\]
by
\begin{equation*}
 x_i\leftarrow x_i+\delta,
 \qquad y_i\leftarrow y_i+p\delta,
 \qquad b_i\leftarrow b_i-p\delta.
\end{equation*}
The aggregate variables \(X,Y,B\) are updated accordingly.  A fixed-price
purchase leaves both \(T^{\rm amm}\) and
\(h=B-p(T^{\rm amm}-X)\) unchanged.

The continuous dynamics stop whenever one of the following events is reached.
\begin{enumerate}
\item \emph{A bidder joins the clinching set.}  There is no state jump.  Add
the identity to \({\sf Clinch}\) and continue with
\eqref{eq:continuous-state-identities}--
\eqref{eq:clincher-primitive-dynamics}.

\item \emph{A valuation exit.}  The identities in \({\sf Exit}\) are marked
as exiting.  If the stopping allocation is executed, simultaneous exits are
served by the proportional rule specified below; otherwise, their remaining
budgets are cancelled simultaneously.

\item \emph{A supply lump.}  A lump is an upward jump
\(\Delta T^{\rm amm}>0\) of the target at a fixed price.  Incorporate the
entire jump by setting
\begin{equation}
 T^{\rm amm}\leftarrow T^{\rm amm}+\Delta T^{\rm amm},
 \qquad S\leftarrow S+\Delta T^{\rm amm}.
 \label{eq:lump}
\end{equation}
For the AMM, such a jump can occur only from a settled state \((X,C(X))\)
when \(C\) has a linear piece of slope \(p\) beginning at \(X\).  If
\([X,X^*]\) is the remainder of the maximal such piece, then
\(\Delta T^{\rm amm}=X^*-X\) and the updated target is
\(T^{\rm amm}=X^*\).
At a price that is both a lump and an exit, incorporate the entire lump before
processing any exiting identity.
\end{enumerate}

After incorporating a possible lump and identifying every exit at \(p\),
update \({\sf Active}\) and \({\sf Exit}\).  Before cancelling any budget or
making any restoration purchase, test
\begin{equation}
 \sum_{i\in{\sf Active}}\frac{b_i}{p}\le S.
 \label{eq:stopping-test}
\end{equation}

If \eqref{eq:stopping-test} holds, perform the \emph{stopping allocation}.
First give every active identity its full demand by applying
\(\operatorname{Buy}(i,b_i/p;p)\) for every \(i\in{\sf Active}\).  Let
\begin{equation}
 L:=S-\sum_{i\in{\sf Active}}\frac{b_i}{p},
 \qquad
 B_E:=\sum_{i\in{\sf Exit}}b_i
 \label{eq:exit-remnant}
\end{equation}
be the residual quantity and the total exiting budget, both evaluated before
these purchases.  If \(B_E>0\), give each exiting identity the proportional
quantity
\begin{equation}
 \delta_i
 :=\frac{b_i}{B_E}\min\left\{L,\frac{B_E}{p}\right\},
 \qquad i\in{\sf Exit},
 \label{eq:proportional-exit-allocation}
\end{equation}
at price \(p\).  Cancel every unspent exiting budget and end the run.  At an
ordinary exit the current chord is cleared.  At a lump, the run may instead
end partway through the linear piece: the unallocated suffix is not released
by the pool.  The proportional rule makes a simultaneous exit independent of
how a physical order is partitioned among equal-valued identities.

Suppose instead that \eqref{eq:stopping-test} fails.  With the entire lump
already incorporated, cancel all exiting budgets simultaneously:
\begin{equation}
 b_j\leftarrow0
 \qquad\text{for every }j\in{\sf Exit},
 \label{eq:cancel-exiting-budgets}
\end{equation}
and update \(B\).  Define the post-event water level
\begin{equation*}
 h^\star:=B-pS,
\end{equation*}
where \(B\) is evaluated after all cancellations.  For each
\(i\in{\sf Active}\) with \(b_i>h^\star\), apply exactly once
\begin{equation}
 \operatorname{Buy}\left(i,\frac{b_i-h^\star}{p};p\right).
 \label{eq:one-shot-restoration}
\end{equation}
These purchases reduce every affected active budget directly to \(h^\star\).
Each purchase leaves \(B-pS\) unchanged, so all purchases use the same
post-event water level and their order is immaterial.

Set \(h\leftarrow h^\star\), recompute \(T^{\rm amm},S,B\), update
\({\sf Active}\), \({\sf Exit}\), and \({\sf Clinch}\) using
\eqref{eq:active-exit-sets} and \eqref{eq:water-invariant}, and resume the
continuous dynamics.  The stopping test need not be repeated: cancelling
exiting budgets changes neither active demand nor residual supply, while each
fixed-price purchase reduces them by the same amount.  This one-shot
restoration is equivalent to restoring after the lump and again after each
exit, because every intermediate restoration purchase leaves the water level
unchanged.

\subsection{Continuation and Termination}

Unless one of the terminal conditions below holds, raise the clock and follow
the continuous evolution of \cref{subsec:continuous-dynamics} until the next
event.  This continuation rule applies for every \(S\ge0\).  In particular,
clearing the current chord does not by itself terminate the mechanism: if
\(S=0\) and an active identity remains, the same primitive equations continue
to govern the auction while the target is reevaluated as the clock rises.

The mechanism ends in any of the following circumstances:
\begin{enumerate}
\item the stopping allocation \eqref{eq:stopping-test} is executed;
\item after all fixed-price operations at the current price have been
processed, no active identity remains;
\item all effective budgets are exhausted.
\end{enumerate}
Upon termination, cancel every unspent effective budget and make no further
allocation or payment.  The no-active-identity test is made only after a
target lump at the same price and the valuation-exit procedure have been
processed, so it does not preempt either operation.

There is no additional terminal condition associated with the AMM capacity
\(x_0\).  At every finite clock price and finite auction state, the chord
target satisfies \(T^{\rm amm}<x_0\), and hence the cumulative allocation
satisfies \(X\le T^{\rm amm}<x_0\).  Thus neither the target nor the
allocation reaches the capacity boundary at a finite state.  The equality
\(S=0\) only says that the current chord has been cleared; if an active
identity remains, raising the clock eventually produces a target strictly
above \(X\).  %
Exact settlement, individual rationality, and the incentive
properties of this outcome are proved in \cref{sec:one-sided-proofs}.

Observe that the rules above refer to the AMM only through the target
\(T^{\rm amm}\), its evolution as the clock rises, and the capacity
\(x_0\); the cost curve \(C\) and its marginal ask \(a\) appear only to
explain what those rules mean for the AMM.  \Cref{sec:one-sided-proofs}
exploits this by proving the
guarantees for a general target, and then verifying that \(T^{\rm amm}\)
qualifies.

\section{\texorpdfstring{Proofs for the One-Sided AMM Clinching\\Mechanism}
{Proofs for the One-Sided AMM Clinching Mechanism}}
\label{sec:one-sided-proofs}

We now prove that the AMM adaptive clinching mechanism
satisfies Sybil-robust incentive compatibility
as well as weak Pareto optimality for eligible bids. 
Before we begin our proof, we first
define an abstract
target-based clinching mechanism
 in \Cref{subsec:abstract-target-mechanism}.
This more generalized abstract mechanism allows us to 
isolate the proof arguments from the particular choice of supply function. 
Specifically, later in 
\Cref{app:constant-targets}, 
we point out 
that the classical ordinary 
adaptive clinching auction~\cite{adaptiveclinch,budgetelicitation},
and
its version with exogenous online supply~\cite{clinchonline}, are the special cases in which the
target is constant.
Thus, our proofs directly imply that these classical adaptive clinching auctions
are sybil-proof too. 

After defining the abstract target-based mechanism, 
we derive the explicit continuous rates implied by its
primitive conditions.  We then verify the abstract requirements directly for
the AMM chord target, and establish the structural lemmas for a single run.
We next 
compare a merged run with a Sybil run to prove the desired incentive compatibility properties.

\subsection{Abstract Target-Based Clinching Mechanism}
\label{subsec:abstract-target-mechanism}

The AMM clinching mechanism is driven by its chord target
\(T^{\rm amm}(p,X,Y)\), which determines how much cumulative supply is
available at each auction state.  We now abstract this role from the geometry
of the AMM.  The resulting target-based mechanism uses the same clinching
rules, but allows a general admissible target function.
Choosing \(T^{\rm amm}\) as its target recovers the AMM mechanism.

\begin{definition}[Abstract target-based mechanism]
\label{def:abstract-target-based-mechanism}
An \emph{abstract target-based mechanism} is a five-tuple
\begin{equation*}
 \mathfrak T
 :=(p_\star,\sigma,T_\sigma,
    \mathcal D_\sigma,\overline X_\sigma),
\end{equation*}
whose components have the following meanings.
\begin{enumerate}
\item \(p_\star>0\) is the initial clock price.
\item \(\sigma\) is an environment description fixed independently of the
submitted bids and common to every bid profile being compared.  It contains
any non-strategic information that the target rule may consult as the clock
advances.  If no such information is needed like in the AMM clinching case, \(\sigma = \emptyset\) and its
subscript may be omitted.
\item \(T_\sigma\) is the cumulative target rule.  At a state
\((p,X,Y)\), it specifies the cumulative quantity that may be allocated.
\item \(\mathcal D_\sigma\) is the domain of auction states on which the
target rule is defined.
\item \(\overline X_\sigma\in(0,+\infty]\) is the capacity bound.
\end{enumerate}
Thus \(T_\sigma\) is a map from \(\mathcal D_\sigma\) into
\([0,\overline X_\sigma]\); its explicit price argument selects any
clock-indexed information contained in \(\sigma\).
\end{definition}

At price \(p\), aggregate allocation \(X\), and aggregate payment \(Y\), set
\begin{equation*}
 S:=T_\sigma(p,X,Y)-X.
\end{equation*}
For an admissible abstract target-based mechanism, \(T_\sigma(p,X,Y)\ge X\), so
\(S\) is the part of the cumulative target not yet allocated.  For the AMM
application, \(p_\star=p_0\), \(\sigma\) is trivial, and
\(T_\sigma=T^{\rm amm}\); the full specification is verified in
\cref{subsec:amm-target-admissible}.

\begin{definition}[Admissible abstract target-based mechanism]
\label{def:admissible-target}
An abstract target-based mechanism \(\mathfrak T\) is \emph{admissible} if the
following properties hold.
\begin{enumerate}
\item \emph{Feasibility and partition independence.}  For every
\((p,X,Y)\in\mathcal D_\sigma\),
\begin{equation*}
 X\le T_\sigma(p,X,Y)\le\overline X_\sigma,
\end{equation*}
where the specified capacity \(\overline X_\sigma\) is common to all bid
profiles.  The target depends on bidders only through \((X,Y)\), and not on
their names or on the partition of an aggregate budget among identities.

\item \emph{Fixed-price invariance.}  If \(0\le\delta\le S\), then the
updated state belongs to the domain and
\begin{equation}
 T_\sigma(p,X+\delta,Y+p\delta)
 =T_\sigma(p,X,Y).
 \label{eq:abstract-fixed-price-invariance}
\end{equation}

\item \emph{Target decreases with mark-to-market gain.}  At a common
\((p,\sigma)\), for any two
states in the domain,
\begin{equation}
 pX-Y\le pX'-Y'
 \quad\Longrightarrow\quad
 T_\sigma(p,X,Y)\ge T_\sigma(p,X',Y').
 \label{eq:target-mark-to-market-order}
\end{equation}
Thus a state with the smaller aggregate mark-to-market gain has the weakly
larger target.  Equality of the gains forces equality of the targets.

\item \emph{Clock regularity.}  If \((X,Y)\) is held fixed, the map
\(p\mapsto T_\sigma(p,X,Y)\) is nondecreasing and right-continuous, is
locally of bounded variation, and has only upward jumps.  Away from those
jumps it is locally absolutely continuous; write
\begin{equation}
 \lambda_\sigma(p,X,Y)
 :=\left.\frac{\dif}{\dif r}
 T_\sigma(r,X,Y)\right|_{r=p}\ge0
 \label{eq:target-clock-rate}
\end{equation}
for its almost-everywhere frozen-state clock rate.

\item \emph{Eventual supply release.}  For every
\((p,X,Y)\in\mathcal D_\sigma\) such that
\(T_\sigma(p,X,Y)=X<\overline X_\sigma\), there is a price \(p'>p\) for
which
\begin{equation}
 (p',X,Y)\in\mathcal D_\sigma
 \qquad\text{and}\qquad
 T_\sigma(p',X,Y)>X.
 \label{eq:no-subcapacity-absorption}
\end{equation}
Thus the target may remain equal to \(X\) over a waiting interval, but it
cannot remain there forever while \(X\) is strictly below the specified
capacity.
\end{enumerate}
The domain is also required to contain \((p_\star,0,0)\) and to be upward
closed in its price coordinate: if \((p,X,Y)\in\mathcal D_\sigma\), then
\((p',X,Y)\in\mathcal D_\sigma\) for every \(p'\ge p\).  These are requirements
on one abstract target-based mechanism; none refers to an execution, a merged run,
or a Sybil run.
A state with \(X=\overline X_\sigma<+\infty\) is called
\emph{capacity-saturated}.  Feasibility then forces
\(T_\sigma(p,X,Y)=X\) at every such state.
\end{definition}

The \emph{target-based clinching mechanism} for an admissible
\(\mathfrak T\) is obtained from the AMM mechanism of \cref{sec:mechanism}
by replacing its AMM-specific variables as follows:
\begin{equation}
 (p_0,T^{\rm amm},x_0)
 \longmapsto
 (p_\star,T_\sigma,\overline X_\sigma).
 \label{eq:abstract-target-substitution}
\end{equation}
Thus reports are eligible when \(v_i\ge p_\star\), residual supply is
\(S=T_\sigma-X\), and the target is bounded by
\(\overline X_\sigma\).  An upward jump of \(T_\sigma\), whether inherent in
its price dependence or encoded by \(\sigma\), is the supply lump in
\eqref{eq:lump}.  Every state identity, primitive differential condition,
event order, stopping allocation, restoration rule, and continuation rule is
otherwise unchanged.  The analytical quantity \(\lambda_\sigma\) in
\eqref{eq:target-clock-rate} is derived from \(T_\sigma\) and is not an
additional component of \(\mathfrak T\).  The substitution therefore changes
the source of target supply, not the clinching mechanics.

\paragraph{Convention at capacity saturation.}
Earlier in our AMM clinching mechanism, we never saturate the capacity, that is, 
\(T^{\rm amm}<x_0\) at every finite state.  
However, our generalized target-based mechanism  
may reach capacity saturation at a finite state, so 
we must explicitly 
specify what happens afterward.  We adopt the following completion
convention: attaining capacity does not itself terminate the execution.  The
clock and ordinary valuation exits continue under the usual rules.
\Cref{lem:capacity-zero-trade-tail} shows that this continuation involves no
further trade, so we call it the \emph{zero-trade completion convention}.
Allocations and payments are therefore already economically final; the
continuation only retains effective budgets until their ordinary valuation
exits.  This prevents unused budgets from being treated as exhausted merely
because the available capacity has been fully allocated.

\subsection{Derived Form of the Continuous Dynamics}
\label{subsec:derived-dynamics}

The preceding construction defines the auction through primitive
differential conditions, without postulating a trading-rate formula.  We now
solve those conditions once for an arbitrary admissible abstract target-based
mechanism.  The resulting identities are consequences of the mechanism and will be
used throughout the proofs below.

\begin{lemma}[Fixed-price accounting]
\label{lem:fixed-price-accounting}
Suppose \(\delta\) units are allocated at price \(p\) and charged
\(p\delta\).  Then \(T_\sigma\) and \(h=B-pS\) are unchanged, while
\(S\) and the buyer's effective budget decrease by \(\delta\) and
\(p\delta\), respectively.
\end{lemma}

\begin{proof}
Target invariance is \eqref{eq:abstract-fixed-price-invariance}.  The updates
are \(X^+=X+\delta\), \(B^+=B-p\delta\), and
\(S^+=T_\sigma-X^+=S-\delta\).  Therefore
\[
 B^+-pS^+=B-p\delta-p(S-\delta)=B-pS=h.
\]
\end{proof}

\begin{lemma}[Explicit continuous dynamics]
\label{lem:derived-continuous-rate}
Consider an event-free interval on which the frozen-state target rate
\(\lambda_\sigma\) exists, and write \(c=|{\sf Clinch}|\).  The primitive
conditions \eqref{eq:continuous-state-identities}--
\eqref{eq:clincher-primitive-dynamics}, after the target substitution
\eqref{eq:abstract-target-substitution}, uniquely determine the state
derivatives almost everywhere.  Define the derived rate
\begin{equation}
 \gamma:=S+p\lambda_\sigma.
 \label{eq:derived-clinching-rate}
\end{equation}
Then every identity follows
\begin{align}
 \frac{\dif x_i}{\dif p}&=
 \begin{cases}
   \gamma/p,&i\in{\sf Clinch},\\
   0,&i\notin{\sf Clinch},
 \end{cases}
 &
 \frac{\dif y_i}{\dif p}&=
 \begin{cases}
   \gamma,&i\in{\sf Clinch},\\
   0,&i\notin{\sf Clinch},
 \end{cases}
 \label{eq:derived-allocation-payment-ode}\\
 \frac{\dif b_i}{\dif p}&=
 \begin{cases}
   -\gamma,&i\in{\sf Clinch},\\
   0,&i\notin{\sf Clinch}.
 \end{cases}
 \label{eq:derived-budget-ode}
\end{align}
The aggregate variables satisfy
\begin{equation}
 \frac{\dif h}{\dif p}=-\gamma,
 \qquad
 \frac{\dif T_\sigma}{\dif p}=\lambda_\sigma,
 \qquad
 \frac{\dif S}{\dif p}
 =\lambda_\sigma-\frac{c\gamma}{p}.
 \label{eq:abstract-aggregate-dynamics}
\end{equation}
In particular, \(h\) is nonincreasing.  If, after all same-price events have
been processed, \(S=0\) and there is exactly one active identity \(i\), then
\(i\) is the sole clincher and
\begin{equation}
 \frac{\dif x_i}{\dif p}=\lambda_\sigma,
 \qquad
 \frac{\dif y_i}{\dif p}=p\lambda_\sigma,
 \qquad
 \frac{\dif b_i}{\dif p}=-p\lambda_\sigma,
 \qquad
 \frac{\dif S}{\dif p}=0.
 \label{eq:single-active-boundary-dynamics}
\end{equation}
Thus the sole active identity purchases every continuously appearing target
increment at its appearance price.
\end{lemma}

\begin{proof}
First suppose \(c>0\).  The tightness condition in
\eqref{eq:clincher-primitive-dynamics} requires every clincher's budget to
have the same derivative as \(h\).  Hence their common payment rate can be
written
\[
 \rho:=-\frac{\dif h}{\dif p}
       =\frac{\dif y_i}{\dif p}
       =-\frac{\dif b_i}{\dif p},
 \qquad i\in{\sf Clinch}.
\]
Current-price charging then gives
\begin{equation}
 \frac{\dif B}{\dif p}=-c\rho,
 \qquad
 \frac{\dif X}{\dif p}=\frac{c\rho}{p},
 \qquad
 \frac{\dif Y}{\dif p}=c\rho.
 \label{eq:aggregate-rho-rates}
\end{equation}
The last two derivatives point in the fixed-price direction.  By
\eqref{eq:abstract-fixed-price-invariance}, they do not move the target; only
the frozen-state clock motion does.  Therefore
\[
 \frac{\dif T_\sigma}{\dif p}=\lambda_\sigma,
 \qquad
 \frac{\dif S}{\dif p}
 =\lambda_\sigma-\frac{c\rho}{p}.
\]
Differentiating \(h=B-pS\) and using \eqref{eq:aggregate-rho-rates} yields
\[
 -\rho
 =-c\rho-S-p\left(\lambda_\sigma-\frac{c\rho}{p}\right)
 =-S-p\lambda_\sigma.
\]
Consequently \(\rho=S+p\lambda_\sigma=\gamma\).  Since \(S\ge0\) and
\(\lambda_\sigma\ge0\), the resulting allocation rates are nonnegative, as
required by the primitive conditions.  This proves existence and uniqueness
of \eqref{eq:derived-allocation-payment-ode}--
\eqref{eq:abstract-aggregate-dynamics} when \(c>0\).

If \(c=0\), \eqref{eq:nonclincher-primitive-dynamics} fixes all individual
states.  Thus \(B'=X'=Y'=0\), while target evolution gives
\(T_\sigma'=\lambda_\sigma\).  Differentiating \(S=T_\sigma-X\) and
\(h=B-pS\) gives the same aggregate identities with \(c=0\).  Finally, when
\(S=0\) and exactly one active identity remains, its outside demand is zero,
so it is the sole clincher.  Setting \(c=1\) and \(S=0\) in the formulas just
proved gives \eqref{eq:single-active-boundary-dynamics}.
\end{proof}

\begin{lemma}[Capacity saturation is a zero-trade tail]
\label{lem:capacity-zero-trade-tail}
Suppose a run reaches a capacity-saturated state.  Thereafter its aggregate
and individual allocations and payments remain constant.  The clock and the
ordinary valuation-exit accounting may continue, but no identity makes a
further purchase.
\end{lemma}

\begin{proof}
At saturation, \(X=\overline X_\sigma\).  At every allowed future state,
target feasibility gives \(X\le T_\sigma\le\overline X_\sigma=X\).  Hence
\(T_\sigma=X\), \(S=0\), and \(\lambda_\sigma=0\).  The derived rate is therefore
\(\gamma=S+p\lambda_\sigma=0\), so the continuous dynamics make no trade.
There can be no future target lump.  At a valuation exit with another active
identity remaining, the stopping test fails, the exiting budgets are
cancelled, and \(h^\star=B\); since every remaining budget is at most \(B\),
the restoration rule makes no purchase.  When no active identity remains,
the run terminates with zero residual supply.  Thus no event in the completed
tail changes an allocation or payment.
\end{proof}

\subsection{The AMM Target Specification Is Admissible}
\label{subsec:amm-target-admissible}

For use in the proof, define the marginal-supply function
\begin{equation}
 Z(p):=\sup\{z\in[0,x_0):a(z)\le p\},
 \qquad p\ge p_0.
 \label{eq:marginal-supply}
\end{equation}
Thus \(Z(p)\) is the amount of pool supply available at a right marginal
ask of at most \(p\).  If \(a\) is strictly increasing and continuous, then
\(Z(p)=a^{-1}(p)\).  The generalized-inverse definition in
\eqref{eq:marginal-supply} continues to apply when \(a\) has flat segments or
jumps, because \(C\) is convex but need not be strictly convex or
differentiable everywhere.

For the AMM, use the state domain
\begin{equation*}
 \mathcal D^{\rm amm}
 :=\{(p,X,Y):p\ge p_0,\ 0\le X<x_0,\ Y\ge C(X)\}.
\end{equation*}
The inequality \(Y\ge C(X)\) says that the current auction state lies on or
above the pool cost curve.

\begin{proposition}[Admissibility of the AMM target specification]
\label{prop:amm-target-admissible}
Let \(\sigma^{\rm amm}\) be the trivial environment description.  If the
reserve curve satisfies \cref{def:two-sided-amm}, then
\begin{equation*}
 \mathfrak T^{\rm amm}
 :=(p_0,\sigma^{\rm amm},T^{\rm amm},
    \mathcal D^{\rm amm},x_0)
\end{equation*}
is an admissible abstract target-based mechanism with unattained capacity.  Its
target-based clinching mechanism is the mechanism of \cref{sec:mechanism}.
\end{proposition}

\begin{proof}
Fix \((p,X,Y)\in\mathcal D^{\rm amm}\) and write
\(g_p(z)=C(z)-Y-p(z-X)\).  The sublevel set
\begin{equation*}
 \{z\in[X,x_0):g_p(z)\le0\}
\end{equation*}
contains \(X\), is an interval by convexity, and is bounded away from
\(x_0\) because \(C(z)\to+\infty\).  It therefore has a right endpoint
\(T^{\rm amm}\in[X,x_0)\).  Continuity gives
\begin{equation}
 C(T^{\rm amm})=Y+p(T^{\rm amm}-X).
 \label{eq:amm-target-equality-proof}
\end{equation}
The equality also holds when the target equals \(X\), because
\(g_p(X)<0\) would make a right neighborhood of \(X\) feasible.  The right
derivative of \(g_p\) at its right endpoint is nonnegative, hence
\(a(T^{\rm amm})\ge p\).  This proves target feasibility and the capacity
bound.  The formula uses no individual reports, proving partition
independence.

For a fixed-price purchase \(0\le\delta\le T^{\rm amm}-X\), every candidate
\(z\ge X+\delta\) satisfies
\begin{align*}
 &C(z)-(Y+p\delta)-p\bigl(z-(X+\delta)\bigr)\notag\\
 &\hspace{35mm}=C(z)-Y-p(z-X).
\end{align*}
Thus the right endpoint is unchanged.  Moreover, convexity and the endpoint
inequalities imply \(g_p(z)\le0\) for every
\(z\in[X,T^{\rm amm}]\); consequently the updated state still satisfies
\(Y+p\delta\ge C(X+\delta)\).  This proves fixed-price invariance and
preservation of the domain.

To prove that the target decreases with mark-to-market gain, put
\begin{equation*}
 \psi_p(z):=pz-C(z).
\end{equation*}
The chord condition is
\(\psi_p(z)\ge pX-Y\), where \(pX-Y\) is the aggregate mark-to-market gain.
The function \(\psi_p\) is concave, reaches its rightmost maximum at
\(Z(p)\), and is strictly decreasing on
\([Z(p),x_0)\).  The right endpoint of every nonempty superlevel set lies in
that interval.  Raising the required mark-to-market gain therefore weakly
moves the right
endpoint to the left, which is exactly
\eqref{eq:target-mark-to-market-order}.

If the state is held fixed and \(p\) rises, then for every \(z>X\) the
quantity \(g_p(z)\) weakly decreases.  Hence the chord target is
nondecreasing.  Convexity gives right-continuity and local bounded variation.
On a continuous nondegenerate portion with \(S>0\) and
\(a(T^{\rm amm})>p\), differentiating
\eqref{eq:amm-target-equality-proof} with \((X,Y)\) fixed gives
\begin{equation}
 \lambda^{\rm amm}
 =\frac{T^{\rm amm}-X}{a(T^{\rm amm})-p}
 =\frac{S}{a(T^{\rm amm})-p}.
 \label{eq:amm-target-clock-rate}
\end{equation}
The only upward discontinuity can occur from a settled state
\((X,C(X))\) when the clock reaches \(a(X)\) and \(C\) is linear from
\(X\) to the right.  The target then jumps to the right endpoint \(Z(p)\)
of that linear piece.  Indeed, the tangent line of slope \(p=a(X)\) agrees
with \(C\) precisely on the maximal such piece and lies below \(C\)
afterward.  These are exactly the supply lumps specified in
\eqref{eq:lump}.

Finally, no state with \(X<x_0\) exhausts all future AMM supply: because
\(C\) is finite at \(X\) and diverges only as \(z\uparrow x_0\), sufficiently
large prices produce a target strictly above \(X\).  The capacity \(x_0\)
is never attained.  For the AMM, substituting
\eqref{eq:amm-target-clock-rate} into the derived rate
\eqref{eq:derived-clinching-rate} gives
\begin{equation*}
 \gamma=S+p\frac{S}{a(T^{\rm amm})-p}
 =S\frac{a(T^{\rm amm})}{a(T^{\rm amm})-p}.
\end{equation*}
At a settled state with one active identity, the target-increment identity
\eqref{eq:single-active-boundary-dynamics} traces the generalized
marginal-supply path \(Z\), including its linear-piece jumps.  Hence every
AMM interpretation stated in \cref{sec:mechanism} follows from the uniform
substitution \eqref{eq:abstract-target-substitution}.
\end{proof}

\subsection{Single-Run Structure}
\label{subsec:single-run-structure}

The next lemmas concern one execution of the target-based clinching mechanism
associated with an admissible abstract target-based mechanism.  They are proved
before any two executions are coupled.

\begin{lemma}[Event restoration and stopping]
\label{lem:event-restoration}
At an event price, include the entire target jump before evaluating the
stopping test.  Let \(B_A\) and \(B_E\) be the total active and exiting
budgets at that moment.  Then:
\begin{enumerate}
\item the stopping test is equivalent to \(B_E\ge h\);
\item if it fails, cancelling all exiting budgets gives
\(h^\star=h-B_E>0\), and the one-shot purchases
\((b_i-h^\star)_+/p\) are feasible and restore \(b_i\le h^\star\) for every
active identity;
\item the truth value of the stopping test does not change during the
fixed-price purchases; and
\item the exit of a clinching identity is a stopping event.
\end{enumerate}
\end{lemma}

\begin{proof}
At the test, \(B=B_A+B_E\), so
\[
 B_A\le pS=B-h
 \quad\Longleftrightarrow\quad B_E\ge h.
\]
If the test fails, \(B_E<h\).  Cancellation changes \(B\), and hence \(h\),
by \(-B_E\), giving \(h^\star=h-B_E>0\).  Reducing every budget above
\(h^\star\) to \(h^\star\) restores the water-level inequalities.  The total
quantity required is feasible because
\[
 \sum_i(b_i-h^\star)_+
 \le B_A-h^\star=pS.
\]
The inequality is immediate if no budget exceeds \(h^\star\); otherwise it
follows by subtracting at least one copy of \(h^\star\) from a subsum of
\(B_A\).  By \cref{lem:fixed-price-accounting}, every such purchase leaves
\(h^\star\) fixed and reduces active demand and residual supply by the same
amount, proving the second and third statements.  Finally, immediately
before exiting, a clincher has \(b_i=h\).  A simultaneous target jump can only
lower the newly computed \(h\), so the total exiting budget is at least the
post-jump water level.  The first statement then implies stopping.
\end{proof}

\begin{lemma}[Allocation structure and clincher persistence]
\label{lem:allocation-and-persistence}
In a continuing run, the following statements hold.
\begin{enumerate}
\item The supply invariant is equivalent to
\begin{equation}
 b_i\le h\quad(i\in{\sf Active}),
 \qquad
 {\sf Clinch}=\{i\in{\sf Active}:b_i=h\}.
 \label{eq:abstract-water-characterization}
\end{equation}
\item An identity that has never clinched has received no allocation and made
no payment.
\item Once an identity clinches, it remains a clincher until the run ends.
\item No positive effective budget is exhausted in a continuing run.  Hence
every identity with \(v_i>p\) and positive reported budget is active, and
\begin{equation}
 b_i=\min\{q_i,h\}
 \qquad(i\in{\sf Active}).
 \label{eq:active-budget-water-formula}
\end{equation}
\end{enumerate}
\end{lemma}

\begin{proof}
Because \(pD_{-i}=B-b_i\), the inequality \(S\le D_{-i}\) is exactly
\(b_i\le B-pS=h\), with equality precisely for a clincher.  This proves the
first statement.

The mechanism allocates only to a current clincher, to an active identity in
a restoration purchase, or as part of a stopping allocation.  A restoration
buyer finishes with budget equal to the water level and therefore becomes a
clincher; a stopping allocation ends the run.  When \(S=0\) and one active
identity remains, that identity is a clincher because the demand of all other
identities is zero.  Thus a bidder that has never clinched has not traded in
any continuing run, proving the second statement.

Between events, \cref{lem:derived-continuous-rate} decreases a clincher's
budget and \(h\) at the same rate.  At a nonstopping event, the new level is
positive and every active budget above it is reduced to it; in particular,
every prior clincher is returned to equality.  At an event where a clincher
exits, \cref{lem:event-restoration} says that the run stops.  When \(S=0\)
and one active identity remains, the primitive tightness condition likewise
preserves equality because that identity holds the entire effective budget.
This proves persistence.

On a continuous interval a clincher can reach budget zero only when \(h=0\),
at which point every active budget is zero.  At a nonstopping event the
post-event level is strictly positive.  Therefore exhaustion of a positive
effective budget ends the run.  If an active identity has never clinched,
the second statement gives \(b_i=q_i\le h\); if it has clinched, persistence
gives \(b_i=h\le q_i\).  This proves
\eqref{eq:active-budget-water-formula} and the fourth statement.
\end{proof}

\begin{lemma}[Exited identities in a continuing run]
\label{lem:exited-zero-payment}
If a run continues after price \(p\), every identity \(j\) with \(v_j\le p\)
has \(x_j=y_j=0\), and its entire unspent reported budget was cancelled when
it exited.
\end{lemma}

\begin{proof}
If \(j\) had clinched before its exit, persistence would make it a clincher
at that exit, which would stop the run by
\cref{lem:event-restoration}.  If it first received allocation in the
stopping allocation at its exit, the run would again have ended.  The
allocation-structure statement in \cref{lem:allocation-and-persistence}
therefore gives \(x_j=y_j=0\), and the nonstopping exit rule cancels its full
budget.
\end{proof}

\begin{lemma}[Sole survivor for the AMM target]
\label{lem:amm-sole-survivor}
If a reachable AMM run continues at a state with \(S=0\) and \(p>p_0\), then
exactly one identity has positive effective budget.  In particular, at most
one active identity remains.
\end{lemma}

\begin{proof}
Since \(S=0\), the definition of the water level gives \(h=B\).  If some
identity clinches, its budget is \(b_i=h=B\) by
\eqref{eq:abstract-water-characterization}, so every other effective budget
is zero.

Suppose instead that nobody clinches.  Clincher persistence implies that
nobody has previously clinched, and the allocation-structure statement then
gives \(X=Y=0\).  But at this state and any \(p>p_0\), the function
\(C(z)-pz\) has right derivative \(p_0-p<0\) at zero.  It is therefore
negative for all sufficiently small \(z>0\), so
\(T^{\rm amm}(p,0,0)>0\).  This contradicts \(S=0\).  Thus a clincher exists
and is the unique identity with positive effective budget.  A continuing run
has at least one such identity, completing the proof.
\end{proof}

\subsection{Feasibility and Efficiency}
\label{subsec:one-sided-feasibility-efficiency}

We now specialize the single-run structure to the AMM target.  All statements
in this subsection concern truthful one-sided sell-\(Y\) participation.

\begin{theorem}[Exact settlement and individual rationality]
\label{thm:one-sided-settlement-ir}
Every completed one-sided AMM clinching run is exact-settlement feasible.  On
a truthful profile it is individually rational: for every identity \(i\),
\begin{equation}
 0\le y_i\le Q_i,
 \qquad
 V_ix_i-y_i\ge0.
 \label{eq:proved-one-sided-ir}
\end{equation}
\end{theorem}

\begin{proof}
We first prove exact settlement.  Consider a stopping allocation and use the
notation \(B_A,B_E,h\) immediately after a possible lump has been included.
If \(h\ge0\), the stopping condition \(B_E\ge h\) and the proportional exit
rule allocate
\[
 \frac{B_A}{p}+\frac{B_E-h}{p}
 =\frac{B-h}{p}=S
\]
units, so the final chord is cleared in full.  Fixed-price invariance and the
chord equality give
\begin{equation*}
 X_f=T^{\rm amm},
 \qquad
 Y_f=Y+p(T^{\rm amm}-X)=C(T^{\rm amm})=C(X_f).
\end{equation*}
If \(h<0\), the stopping event must follow a target lump: away from an upward
target jump, \(h\) starts nonnegative, is continuous and nonincreasing by
\cref{lem:derived-continuous-rate}, and reaching zero exhausts all active
budgets; a nonstopping exit leaves the positive level described in
\cref{lem:event-restoration}.  At a lump, the selected chord is a linear
piece of \(C\) of slope \(p\).  Every active and exiting identity exhausts
its budget, leaving the suffix \(-h/p\) unallocated, but every purchased
prefix still satisfies
\[
 Y_f=Y+p(X_f-X)=C(X_f).
\]

Every other terminal case has \(S=0\).  The target equality
\eqref{eq:amm-target-equality-proof} then gives \(Y_f=C(X_f)\).  In
particular, \eqref{eq:single-active-boundary-dynamics} shows that the
one-active boundary path moves directly along the pool curve.  Thus every
terminal state satisfies \eqref{eq:one-sided-exact-settlement}; all pool
output is assigned and all payments enter the pool.

For individual rationality, the mechanism never charges more than the
reported budget, so truthful reporting gives \(y_i\le q_i=Q_i\).  Every
increment received while active is priced below \(v_i=V_i\), and every
increment received at its exit is priced exactly \(V_i\).  Therefore
\[
 y_i=\int p\,\dif x_i\le V_i\int\dif x_i=V_ix_i.
\]
This includes fixed-price event purchases and Stieltjes jumps.  Hence
\eqref{eq:proved-one-sided-ir} holds.
\end{proof}

For a truthful bidder, let
\begin{equation*}
 r_i:=Q_i-y_i
\end{equation*}
denote its terminal unspent \(Y\)-budget, and let
\(a(X_f)=C'_+(X_f)\) be the terminal right marginal pool ask.

\begin{lemma}[Terminal certificate]
\label{lem:one-sided-terminal-certificate}
Every completed truthful one-sided run satisfies:
\begin{enumerate}
\item if \(x_j>0\) and \(V_i>V_j\), then \(r_i=0\); and
\item if \(r_i>0\), then \(V_i\le a(X_f)\).
\end{enumerate}
\end{lemma}

\begin{proof}
For the first statement, consider how identity \(j\) first receives
allocation.  If it receives allocation while clinching or in a restoration
purchase, it remains a clincher until the run ends.  If the run reaches its
exit at \(V_j\), that exit stops the run; if the run ends earlier, both
\(i\) and \(j\) are still active.  In either case the stopping rule gives
identity \(i\) its full remaining demand and exhausts its budget.  If \(j\) first
receives allocation in that stopping event, the same conclusion follows
directly from the stopping rule.  If \(j\) receives allocation in the
one-active boundary regime, every other effective budget, including that of
the higher-valued identity \(i\), is already zero.  These cases exhaust the
allocation rules and prove item~1.

Let \(p_f\) be the terminal clock price.  We first note that
\begin{equation}
 p_f\le a(X_f).
 \label{eq:terminal-price-below-ask}
\end{equation}
At a fully cleared chord this follows from
\(a(T^{\rm amm})\ge p_f\); at a partial lump it follows because the final
point lies on a linear piece of slope \(p_f\).  In the one-active boundary
regime, the bidder either waits below the next ask or trades at the marginal
ask, so the same inequality holds.

If identity \(i\) exited, then \(V_i\le p_f\), and
\eqref{eq:terminal-price-below-ask} proves item~2.  If it did not exit and
\(r_i>0\), suppose for contradiction that \(V_i>a(X_f)\).  Exact settlement
places the state at \((X_f,C(X_f))\).  By right continuity of the convex
cost, a sufficiently small additional amount of \(X\) has average incremental
cost strictly below \(V_i\), and \(i\) has enough value and positive residual
budget to purchase some such amount.  The continuation rule would not have
terminated at \(X_f\), a contradiction.
\end{proof}

\begin{theorem}[Weak Pareto optimality in the sell-\(Y\) direction]
\label{thm:one-sided-weak-pareto}
The truthful one-sided AMM clinching outcome is weakly Pareto optimal for
eligible bids in the sense of \cref{def:weak-pareto}.  
\end{theorem}

\begin{proof}
First consider an additional internal trade among eligible users.  If every
user has \(r_i=0\), nobody holds \(Y\) that can be tendered, while every
holder values its \(X\) positively; a Pareto improvement is impossible.
Otherwise set
\begin{equation*}
 \lambda:=\max\{V_i:r_i>0\}.
\end{equation*}
Item~1 of \cref{lem:one-sided-terminal-certificate} implies that every user
holding positive \(X\) has value at least \(\lambda\), while every user able
to tender \(Y\) has value at most \(\lambda\).

Represent a proposed internal trade by changes
\((\delta x_i,\delta y_i)_i\) in user holdings.  Conservation gives
\(\sum_i\delta x_i=\sum_i\delta y_i=0\), and the utility change is
\(V_i\delta x_i+\delta y_i\).  Let
\(Q=\sum_{\delta x_i>0}\delta x_i\).  A user paying \(Y\) must have
\(r_i>0\), and weak improvement requires
\begin{equation*}
 -\delta y_i\le V_i\delta x_i\le\lambda\delta x_i.
\end{equation*}
Thus buyers can pay at most \(\lambda Q\).  Every supplier of \(X\) has
positive initial \(X\)-holdings after the mechanism and hence value at least
\(\lambda\); weak improvement requires
\begin{equation*}
 \delta y_j\ge V_j(-\delta x_j)
                 \ge\lambda(-\delta x_j).
\end{equation*}
Suppliers therefore require at least \(\lambda Q\).  Conservation forces
equality throughout, so every participant's utility change is zero.  No
strict Pareto improvement exists.

It remains to consider an additional same-direction pool trade.  Exact
settlement starts it from \((X_f,C(X_f))\).  Buying \(\delta>0\) more units
of \(X\) costs at least
\begin{equation*}
 C(X_f+\delta)-C(X_f)\ge a(X_f)\delta.
\end{equation*}
Any bidder able to fund the trade has \(r_i>0\), and therefore
\(V_i\le a(X_f)\) by item~2 of the terminal certificate.  Its utility change
is nonpositive.  This proves both clauses of the one-sided restriction of
\cref{def:weak-pareto}.
\end{proof}

\subsection{Sybil-Robust Incentive Compatibility}
\label{subsec:one-sided-sybil-proof}

Fix a physical sell-\(Y\) bidder with true type \((V,Q)\), where
\(V\ge p_\star\).  In the \emph{merged run} \(\mathcal M\), it submits the
single report \((V,Q)\).  In the \emph{Sybil run} \(\mathcal S\), it submits a
finite set \(G\) of sell-\(Y\) identities
\begin{equation}
 (v_i,q_i)_{i\in G},
 \qquad
 q_G:=\sum_{i\in G}q_i\le Q.
 \label{eq:one-sided-sybil-deviation}
\end{equation}
The set \(N\) of outside identities, their reports, and the exogenous target
state are identical in the two runs.  Every unit allocated to an identity in
\(G\) has the physical bidder's true value \(V\).  Thus the two utilities to
be compared are
\begin{equation*}
 U_G^{\mathcal M}=VX_G^{\mathcal M}-Y_G^{\mathcal M},
 \qquad
 U_G^{\mathcal S}=VX_G^{\mathcal S}-Y_G^{\mathcal S}.
\end{equation*}

\begin{lemma}[Capping value overreports]
\label{lem:cap-one-sided-overreports}
Replacing every Sybil value report \(v_i\) by \(\min\{v_i,V\}\) weakly
increases the physical bidder's true utility.
\end{lemma}

\begin{proof}
The original and capped executions coincide below price \(V\), because value
reports affect the mechanism only when identities exit.  Allocations made at
price \(V\) have zero true marginal surplus, whereas allocations made above
\(V\) have negative true marginal surplus \(V-p\).  Capping all overreports
simultaneously removes only zero- or negative-surplus trades and cannot remove
any positive-surplus trade.
\end{proof}

For any fixed set of identities \(\mathcal I\), let \(X_{\mathcal I}(p)\) and
\(Y_{\mathcal I}(p)\) denote its
cumulative allocation and payment through clock price \(p\).  If a run ends,
extend both quantities constantly.  Define its mark-to-market gain by
\begin{equation}
 \Gamma_{\mathcal I}(p):=pX_{\mathcal I}(p)-Y_{\mathcal I}(p).
 \label{eq:mark-to-market-gain}
\end{equation}

\begin{lemma}[Mark-to-market identity]
\label{lem:mark-to-market-identity}
Under current-price charging,
\begin{equation*}
 \Gamma_{\mathcal I}(p)=\int_{p_\star}^{p}X_{\mathcal I}(s)\,\dif s.
\end{equation*}
In particular, \(\Gamma_{\mathcal I}\) is continuous across fixed-price
purchases and is locally absolutely continuous, with derivative
\(X_{\mathcal I}\) almost everywhere.  If every identity in \(\mathcal I\)
reports at most \(V\), then its aggregate true utility is
\(\Gamma_{\mathcal I}(V)\).
\end{lemma}

\begin{proof}
Current-price charging gives the Stieltjes identity
\(Y_{\mathcal I}(p)=\int_{[p_\star,p]}s\,\dif X_{\mathcal I}(s)\).
Integration by parts yields
\[
 pX_{\mathcal I}(p)-Y_{\mathcal I}(p)
 =p_\star X_{\mathcal I}(p_\star)-Y_{\mathcal I}(p_\star)
   +\int_{p_\star}^{p}X_{\mathcal I}(s)\,\dif s.
\]
Every allocation at the initial price is charged \(p_\star\), so the first
two terms cancel.  A jump of \(X_{\mathcal I}\) at price \(p\) is accompanied
by the payment jump \(p\,\Delta X_{\mathcal I}\), making
\(\Gamma_{\mathcal I}\) continuous.  If all reports are at most \(V\), no
member of \(\mathcal I\) buys above \(V\), and evaluating
\eqref{eq:mark-to-market-gain} at \(V\) gives its true utility.
\end{proof}

For a run \(\mathcal R\in\{\mathcal M,\mathcal S\}\), write
\(B_N^{\mathcal R}\) for
the outsiders' total effective remaining budget and define
\begin{equation}
 W_N^{\mathcal R}:=X_N^{\mathcal R}+\frac{B_N^{\mathcal R}}p.
 \label{eq:outsider-quantity-wealth}
\end{equation}
Thus \(W_N^{\mathcal R}\) is the outsiders' quantity-denominated wealth in
run \(\mathcal R\): the \(X\) they have already received plus the amount of
\(X\) their remaining effective budget could purchase at the current price.
At and after a terminal state, all effective budgets are set to zero for this
accounting convention.  Put
\begin{align*}
 \GainGap_G(p)&:=\Gamma_G^{\mathcal M}(p)-\Gamma_G^{\mathcal S}(p),
 &\GainGap_N(p)&:=\Gamma_N^{\mathcal S}(p)-\Gamma_N^{\mathcal M}(p).
\end{align*}
Thus the group gain gap is oriented as merged minus Sybil, whereas the
outsider gain gap is oriented as Sybil minus merged.  By
\cref{lem:mark-to-market-identity}, \(\GainGap_G\) and \(\GainGap_N\) are
locally absolutely continuous and
\begin{equation*}
 \GainGap_G'=X_G^{\mathcal M}-X_G^{\mathcal S},
 \qquad
 \GainGap_N'=X_N^{\mathcal S}-X_N^{\mathcal M}
 \quad\text{almost everywhere}.
\end{equation*}

\begin{lemma}[Group allocation gap]
\label{lem:abstract-allocation-gap}
At every price below \(V\) at which the merged run continues,
\begin{equation}
 X_G^{\mathcal M}-X_G^{\mathcal S}
 \ge\min\{0,\ W_N^{\mathcal S}-W_N^{\mathcal M}
                   +T^{\mathcal M}-T^{\mathcal S}\}.
 \label{eq:abstract-allocation-gap}
\end{equation}
Here a target of a run that has already ended is evaluated at its fixed
terminal state and the current common price and environment.
\end{lemma}

\begin{proof}
The merged identity is active with positive budget.  Its supply invariant is
\[
 T^{\mathcal M}-X^{\mathcal M}\le\frac{B_N^{\mathcal M}}p,
\]
and therefore
\begin{equation}
 X_G^{\mathcal M}\ge T^{\mathcal M}-W_N^{\mathcal M}.
 \label{eq:merged-group-lower-bound}
\end{equation}
If the Sybil group has never clinched and its run continues, then
\(X_G^{\mathcal S}=0\) by
\cref{lem:allocation-and-persistence}, so
\(X_G^{\mathcal M}-X_G^{\mathcal S}\ge0\).  Otherwise, if the Sybil run
continues, some group identity \(i\) is currently a clincher.  Its
clinching equality gives
\[
 T^{\mathcal S}-X^{\mathcal S}=D_{-i}^{\mathcal S}
 \ge\frac{B_N^{\mathcal S}}p,
\]
so \(X_G^{\mathcal S}\le T^{\mathcal S}-W_N^{\mathcal S}\).  If the Sybil
run has ended, then \(B_N^{\mathcal S}=0\), and target feasibility at its
terminal state gives \(T^{\mathcal S}\ge X^{\mathcal S}\), hence the same
inequality.  Subtracting this
upper bound from \eqref{eq:merged-group-lower-bound} proves
\eqref{eq:abstract-allocation-gap} in the latter cases, while the first case
gives its zero branch.
\end{proof}

\begin{lemma}[Outsider accounting]
\label{lem:abstract-outsider-accounting}
Fix a price \(p<V\) at which the merged run continues, and let
\begin{equation*}
 \Phi_p(t):=\sum_{j\in N:\,v_j>p}\min\{q_j,t\}.
\end{equation*}
If the Sybil run also continues, then
\begin{align}
 Y_N^{\mathcal S}-Y_N^{\mathcal M}
   &=\Phi_p(h^{\mathcal M})-\Phi_p(h^{\mathcal S}),
 \label{eq:outsider-payment-water}\\
 p\bigl(W_N^{\mathcal S}-W_N^{\mathcal M}\bigr)&=\GainGap_N,
 \label{eq:outsider-wealth-gain}
\end{align}
and
\begin{equation}
 \beta^{\mathcal M}-\beta^{\mathcal S}
 =Q-q_G+d_G^{\mathcal S}\ge0,
 \qquad \beta:=Y+B,
 \label{eq:committed-wealth-gap}
\end{equation}
where \(d_G^{\mathcal S}\) is the total Sybil budget cancelled so far, including
budgets removed as ineligible at the initial price.  If the Sybil run has
ended, then
\begin{equation}
 Y_N^{\mathcal S}\ge Y_N^{\mathcal M},
 \qquad
 p\bigl(W_N^{\mathcal S}-W_N^{\mathcal M}\bigr)\ge\GainGap_N.
 \label{eq:ended-sybil-outsider-order}
\end{equation}
\end{lemma}

\begin{proof}
Suppose both runs continue.  By \cref{lem:exited-zero-payment}, every
outsider with \(v_j\le p\) has paid zero and had its full budget cancelled in
both runs.  Every outsider with \(v_j>p\) is active and, by
\eqref{eq:active-budget-water-formula}, has budget
\(\min\{q_j,h\}\).  Summing its payment
\(q_j-\min\{q_j,h\}\) proves
\eqref{eq:outsider-payment-water}.  Since
\(\beta=Y+B\) equals the sum of initial budgets minus cancelled budgets, the
outsider cancellations agree in the two runs, while the merged identity has
no cancellation below \(V\).  This proves
\eqref{eq:committed-wealth-gap}.  Finally,
\[
 pW_N=pX_N+B_N
      =\Gamma_N+Y_N+B_N
      =\Gamma_N+\sum_{j\in N:\,v_j>p}q_j
\]
in both continuing runs, proving \eqref{eq:outsider-wealth-gain}.

Now suppose the Sybil run ended at \(p_1\le p\).  Capacity saturation is not
a terminal event under the completion convention.
Every outsider with \(v_j>p_1\) was active at its terminal event and exhausted
its budget.  An outsider with \(v_j\le p_1\) pays zero in the continuing
merged run and a nonnegative amount in the Sybil run.  Thus
\(Y_N^{\mathcal S}\ge Y_N^{\mathcal M}\).  Because terminal effective budgets are set
to zero,
\[
 p\bigl(W_N^{\mathcal S}-W_N^{\mathcal M}\bigr)
 =\GainGap_N+Y_N^{\mathcal S}-Y_N^{\mathcal M}-B_N^{\mathcal M}.
\]
The last two terms are nonnegative identity by identity: an outsider active
in the merged run has paid its full budget in the ended run, and an outsider
already exited from the merged run has zero effective budget and zero
payment there.  Hence
\(p(W_N^{\mathcal S}-W_N^{\mathcal M})\ge\GainGap_N\).
\end{proof}

\begin{theorem}[Sybil-proofness for an admissible target-based mechanism]
\label{thm:abstract-one-sided-sybil}
In any 
admissible abstract target-based mechanism \(\mathfrak T\), 
truthful
participation as the merged identity \((V,Q)\) weakly dominates every
same-direction Sybil strategy \eqref{eq:one-sided-sybil-deviation}:
\begin{equation}
 VX_G^{\mathcal M}-Y_G^{\mathcal M}
 \ge VX_G^{\mathcal S}-Y_G^{\mathcal S}.
 \label{eq:abstract-sybil-conclusion}
\end{equation}
\end{theorem}

\begin{proof}
By \cref{lem:cap-one-sided-overreports}, it is enough to consider
\(v_i\le V\) for every \(i\in G\).  Then
\cref{lem:mark-to-market-identity} identifies the two utilities with
\(\Gamma_G^{\mathcal M}(V)\) and \(\Gamma_G^{\mathcal S}(V)\), so it remains to prove
\(\GainGap_G(V)\ge0\).

Let \(p_1\le V\) be the price at which the merged run ends, taking
\(p_1=V\) if it continues up to the merged identity's exit.  We first show
that, for almost every \(p\in[p_\star,p_1)\),
\begin{align}
 \GainGap_G\le\GainGap_N&\quad\Longrightarrow\quad
 \GainGap_G'\ge\min\{0,\GainGap_N/p\},
 \label{eq:group-gain-gap-differential-inequality}\\
 \GainGap_G\ge\GainGap_N&\quad\Longrightarrow\quad
 \GainGap_N'\ge\GainGap_N/p.
 \label{eq:outsider-gain-gap-differential-inequality}
\end{align}

The total mark-to-market gain difference is
\begin{equation}
 \GainGap_G-\GainGap_N
 =(pX^{\mathcal M}-Y^{\mathcal M})
  -(pX^{\mathcal S}-Y^{\mathcal S}).
 \label{eq:total-mark-to-market-gap}
\end{equation}
If \(\GainGap_G\le\GainGap_N\), the fact that the target decreases with
mark-to-market gain
gives
\(T^{\mathcal M}\ge T^{\mathcal S}\).  By
\cref{lem:abstract-outsider-accounting},
\(W_N^{\mathcal S}-W_N^{\mathcal M}\ge\GainGap_N/p\), whether the Sybil run
continues or has already ended.  The group allocation bound then gives
\[
 \GainGap_G'=X_G^{\mathcal M}-X_G^{\mathcal S}
 \ge\min\{0,W_N^{\mathcal S}-W_N^{\mathcal M}
              +T^{\mathcal M}-T^{\mathcal S}\}
 \ge\min\{0,\GainGap_N/p\},
\]
proving \eqref{eq:group-gain-gap-differential-inequality}.

Suppose \(\GainGap_G\ge\GainGap_N\).  Then the fact that the target decreases with
mark-to-market gain gives
\(T^{\mathcal M}\le T^{\mathcal S}\).  If the Sybil run has ended,
\eqref{eq:ended-sybil-outsider-order} and
\begin{equation}
 p \cdot \GainGap_N'=p(X_N^{\mathcal S}-X_N^{\mathcal M})
     =\GainGap_N+Y_N^{\mathcal S}-Y_N^{\mathcal M}
 \label{eq:outsider-gain-derivative}
\end{equation}
gives \eqref{eq:outsider-gain-gap-differential-inequality}.  If both runs
continue, use
\begin{equation*}
 h=\beta+(pX-Y)-pT_\sigma.
\end{equation*}
Together with \eqref{eq:committed-wealth-gap},
\eqref{eq:total-mark-to-market-gap}, and
\(T^{\mathcal M}\le T^{\mathcal S}\), this gives
\begin{align*}
 h^{\mathcal M}-h^{\mathcal S}
 &=(\beta^{\mathcal M}-\beta^{\mathcal S})
   +(\GainGap_G-\GainGap_N)
   -p(T^{\mathcal M}-T^{\mathcal S})\ge0.
\end{align*}
The function \(\Phi_p\) is nondecreasing, so
\eqref{eq:outsider-payment-water} gives
\(Y_N^{\mathcal S}\ge Y_N^{\mathcal M}\).  Equation
\eqref{eq:outsider-gain-derivative} again proves
\eqref{eq:outsider-gain-gap-differential-inequality}.

We next show that both gain gaps are nonnegative on \([p_\star,p_1]\):
\begin{equation*}
 \GainGap_G\ge0,
 \qquad
 \GainGap_N\ge0.
\end{equation*}
Let
\begin{equation*}
 \mu:=\min\{\GainGap_G,\GainGap_N\}.
\end{equation*}
It is locally absolutely continuous and \(\mu(p_\star)=0\).  At almost every
price where the relevant functions are differentiable, if \(\mu<0\), then
\begin{equation*}
 \mu'\ge\frac{\mu}{p}.
\end{equation*}
Indeed, if \(\GainGap_G\le\GainGap_N\), then
\eqref{eq:group-gain-gap-differential-inequality} gives
\[
 \GainGap_G'
 \ge\min\{0,\GainGap_N/p\}
 \ge\mu/p.
\]
If instead \(\GainGap_N\le\GainGap_G\), then
\eqref{eq:outsider-gain-gap-differential-inequality} gives
\[
 \GainGap_N'\ge\GainGap_N/p=\mu/p.
\]
At equality points the same conclusion follows from the almost-everywhere
derivative of the minimum.  If \(\mu\) ever became negative, take the last
preceding price \(p_2\) at which \(\mu(p_2)=0\).  On the subsequent negative
interval,
\[
 \left(\frac{\mu(p)}p\right)'
 =\frac{p\mu'(p)-\mu(p)}{p^2}\ge0,
\]
which contradicts \(\mu(p_2)/p_2=0\).  Thus
\(\GainGap_G,\GainGap_N\ge0\) through \(p_1\).

If \(p_1=V\), this already gives the result.  Suppose \(p_1<V\).  Capacity
saturation is followed by zero-trade completion rather than termination, so
the merged run has ended through one of the ordinary terminal rules.  Its
identity is still below its value and hence has not exited: a stopping
allocation gives it its full remaining demand, while either of the other
terminal rules is compatible with it remaining active only if its effective
budget is zero.  Because none of its budget is cancelled below \(V\), it has
therefore exhausted its budget:
\(Y_G^{\mathcal M}=Q\).  At \(p_1\), let \(b_G^{\mathcal S}\) and
\(d_G^{\mathcal S}\) be the Sybil group's remaining and cancelled budgets.
Since \(Y_G^{\mathcal S}=q_G-d_G^{\mathcal S}-b_G^{\mathcal S}\),
\begin{equation*}
 p_1\bigl(X_G^{\mathcal M}(p_1)-X_G^{\mathcal S}(p_1)\bigr)
   -b_G^{\mathcal S}
 =\GainGap_G(p_1)+Q-q_G+d_G^{\mathcal S}\ge0.
\end{equation*}
Every unit subsequently obtained by the Sybil group costs at least \(p_1\),
so its additional allocation is at most
\[
 \frac{b_G^{\mathcal S}}{p_1}
 \le X_G^{\mathcal M}(p_1)-X_G^{\mathcal S}(p_1).
\]
Consequently
\(X_G^{\mathcal M}(p)-X_G^{\mathcal S}(p)\ge0\) for every
\(p\in[p_1,V]\).  Because
\(\GainGap_G'=X_G^{\mathcal M}-X_G^{\mathcal S}\) almost everywhere,
\(\GainGap_G\) is nondecreasing on that interval.  Therefore
\(\GainGap_G(V)\ge\GainGap_G(p_1)\ge0\), proving
\eqref{eq:abstract-sybil-conclusion}.
\end{proof}

\begin{theorem}[One-sided AMM incentive compatibility]
\label{thm:one-sided-amm-incentives}
Fix all reports of the other sell-\(Y\) traders.  A physical sell-\(Y\)
trader of true type \((V,Q)\) cannot improve its utility by replacing its
truthful merged identity \((V,Q)\) with finitely many sell-\(Y\) identities
having arbitrary value reports and aggregate reported budget at most \(Q\).
Thus the one-sided AMM clinching mechanism is %
Sybil-robust 
incentive compatible.
\end{theorem}

\begin{proof}
By \cref{prop:amm-target-admissible}, \(\mathfrak T^{\rm amm}\) is an
admissible abstract target-based mechanism.  Apply
\cref{thm:abstract-one-sided-sybil}.  Since the aggregate Sybil payment is at
most the aggregate reported budget, which is at most \(Q\), both sides of
\eqref{eq:abstract-sybil-conclusion} are finite utilities under
\eqref{eq:two-sided-physical-utility-y} restricted to sell-\(Y\) identities.
\end{proof}

Combining \cref{thm:one-sided-settlement-ir,thm:one-sided-weak-pareto,%
thm:one-sided-amm-incentives}, the one-sided AMM clinching mechanism is exact-settlement feasible,
individually rational, weakly Pareto optimal for eligible bids, and
Sybil-robust 
incentive compatible.

\section{Full Two-Sided AMM Clinching Mechanism}
\label{sec:two-sided-mechanism}

We now extend our one-sided mechanism in \Cref{sec:mechanism}
to allow both sell-$Y$ and sell-$X$ bids. 
Types, eligibility, outcomes, and
settlement retain their definitions from \cref{sec:model}.  

The idea is to
fill the minority side in full at the spot price.  Its orders then supply an
initial quantity to the other side, which competes for this quantity and
additional pool output through the one-sided clinching mechanism.
The main additional proof challenge is the treatment of  
cross-directional strategies, that is,  
a sell-$Y$-type user can potentially post 
sell-$X$-type bids 
under its Sybil identities.

\subsection{Dominant and Minority Sides}
\label{subsec:two-sided-extended-cost}

Let \(\mathcal E_{\rm Y}\subseteq\mathcal I_{\rm Y}\) and
\(\mathcal E_{\rm X}\subseteq\mathcal I_{\rm X}\) be the eligible identities.
Measure both sides at the spot price in units of \(X\):
\begin{equation*}
 D_Y:=\sum_{i\in\mathcal E_{\rm Y}}\frac{q_i}{p_0},
 \qquad
 D_X:=\sum_{i\in\mathcal E_{\rm X}}r_i,
 \qquad
 M:=\min\{D_Y,D_X\}.
\end{equation*}
The side with the smaller quantity is the \emph{minority}; the other is the
\emph{dominant side}.  Here \(M\) is the minority fill in units of \(X\).
If the quantities are equal, designate
sell-\(X\) as the minority.  As shown below, either designation gives full
internal matching at \(p_0\) and leaves the pool unchanged.

\subsection{Allocation and Settlement}
\label{subsec:two-sided-allocation}

No matter whether 
sell-$Y$ or sell-$X$ is dominant, 
ineligible identities receive zero outcomes.  
The eligible identities are processed as follows --- note that 
minority-fill
proceeds are transferred only at final settlement and cannot finance another
order in the same batch.

\paragraph{Sell-\(Y\) is dominant.}
Suppose \(D_Y\ge D_X\), so sell-\(X\) is the minority and \(M=D_X\).  Fold
the minority supply directly into the cost faced by the dominant sell-\(Y\)
auction:
\begin{equation}
 \widetilde C_M(z):=
 \begin{cases}
  p_0z,&0\le z\le M,\\
  p_0M+C(z-M),&M\le z<M+x_0.
 \end{cases}
 \label{eq:two-sided-extended-cost}
\end{equation}
The first \(M\) units cost \(p_0\) each and represent the sell-\(X\) minority
orders, whose total compensation is \(p_0M\).  After those units are
allocated, the cost of further output is the original pool cost.  Thus the
dominant auction treats internally matched supply and pool supply as a single
cost curve.  This curve records aggregate supply costs; individual dominant
trades are still charged at the auction clock price.

For a state with \(Y\ge\widetilde C_M(X)\), define its chord target by the
same construction as \eqref{eq:amm-chord-target}:
\begin{equation*}
 \widetilde T_M(p,X,Y)
 :=\max\{z\in[X,M+x_0):
              \widetilde C_M(z)\le Y+p(z-X)\}.
\end{equation*}
In particular, even if the original pool cost is strictly convex, the
extended cost has a linear segment whenever \(M>0\).  This is why the
one-sided construction includes supply lumps and does not assume strict
convexity.
Execute the following rules.
\begin{enumerate}
\item Assign each eligible minority identity its full spot-price trade:
\begin{equation*}
 S_{X,i}=r_i,\qquad R_{Y,i}=p_0r_i
 \quad(i\in\mathcal E_{\rm X}).
\end{equation*}
\item Run the one-sided mechanism of \cref{sec:mechanism} on the reports
\((v_i,q_i)_{i\in\mathcal E_{\rm Y}}\), replacing
\((C,T^{\rm amm},x_0)\) by
\((\widetilde C_M,\widetilde T_M,M+x_0)\).  The clock still starts at
\(p_0\).  Use the same continuous dynamics, initial lump, event order,
proportional exit rule, and continuation rule.  If its output is
\((x_i,y_i)\), set
\begin{equation*}
 R_{X,i}=x_i,\qquad S_{Y,i}=y_i,
 \qquad X_f:=\sum_{i\in\mathcal E_{\rm Y}}x_i,
 \qquad Y_f:=\sum_{i\in\mathcal E_{\rm Y}}y_i.
\end{equation*}
\item Settle all transfers simultaneously.  The pool releases
\(X_f-M\) units of \(X\) and receives \(Y_f-p_0M\) units of \(Y\).
\end{enumerate}
In \cref{thm:two-sided-settlement-ir} we prove that
\(X_f\ge M\) and \(Y_f-p_0M=C(X_f-M)\), so all these transfers are funded
and the net pool trade lies on the original reserve curve.

\paragraph{Sell-\(X\) is dominant.}
Suppose \(D_Y<D_X\), so sell-\(Y\) is the minority and \(M=D_Y\).  Exchange
the roles of the assets.  In the exchanged coordinates the reserve curve is
\(\varphi^{-1}\), the initial spot price is \(1/p_0\), and the one-sided
pool cost is
\begin{equation*}
 C^{\rm rev}(w):=-C^{-1}(-w),\qquad 0\le w<y_0.
\end{equation*}
The extended cost, now measured as units of \(X\) needed to release \(w\)
units of \(Y\), is
\begin{equation*}
 \widetilde C^{\rm rev}_{M}(w):=
 \begin{cases}
  w/p_0,&0\le w\le p_0M,\\
  M+C^{\rm rev}(w-p_0M),&p_0M\le w<p_0M+y_0.
 \end{cases}
\end{equation*}
At a reverse-auction state, let \(p^{\rm rev}\) be the price in units of \(X\) per
unit of \(Y\), let \(X^{\rm rev}\) be the cumulative amount of \(Y\)
allocated, and let \(Y^{\rm rev}\) be the cumulative amount of \(X\) paid.
Define
\begin{equation*}
 \widetilde T^{\rm rev}_M(p^{\rm rev},X^{\rm rev},Y^{\rm rev})
 :=\max\{w\in[X^{\rm rev},p_0M+y_0):
       \widetilde C^{\rm rev}_M(w)
       \le Y^{\rm rev}+p^{\rm rev}(w-X^{\rm rev})\}.
\end{equation*}
Assign each eligible minority sell-\(Y\) identity its full spot-price trade:
it pays \(q_i\) and receives \(q_i/p_0\).  Run the one-sided mechanism on
the dominant sell-\(X\) identities with reported values\footnote{
For a report \(u_i=0\), interpret \(1/u_i=+\infty\): that identity never
exits because of its value.  
Equivalently, replace \(+\infty\) by any finite value larger than the maximum price the 
reverse clock can reach; Lemma \cref{lem:two-sided-finite-clock} shows that this does not change the outcome.
} \(1/u_i\), budgets
\(r_i\), and clock \(p^{\rm rev}\) starting at \(1/p_0\), replacing the cost, target,
and capacity by
\((\widetilde C^{\rm rev}_M,\widetilde T^{\rm rev}_M,p_0M+y_0)\).  An allocation
\(x_i^{\rm rev}\) and payment \(y_i^{\rm rev}\) in this auction mean
\(R_{Y,i}=x_i^{\rm rev}\) and \(S_{X,i}=y_i^{\rm rev}\).  

\section{\texorpdfstring{Proofs for the Two-Sided AMM Clinching\\Mechanism}
{Proofs for the Two-Sided AMM Clinching Mechanism}}
\label{sec:two-sided-proofs}

We first justify the extended auction and settlement.  For incentives, the
one-sided Sybil theorem already handles splitting and value reports while
the fill is fixed.  The additional issue is that opposite-direction orders
can enlarge that fill or change which side is dominant.  A bound on the
value of extra spot-price supply handles the former; the initial water level
handles the latter.  We then use the one-sided terminal certificate to prove
weak Pareto optimality for eligible bids from both sides.

\subsection{Extended Targets, Termination, and Exact Settlement}

Once the minority fill is folded into the dominant side's cost curve, that
curve and its domain fully determine the resulting one-sided AMM clinching
auction.  The curve's right derivative at zero gives the initial clock
price, the endpoint of its domain gives the capacity, the condition that an
auction state lies on or above the curve gives the state domain, and the
chord construction gives the target.  Thus the one-sided clinching auction when sell-\(Y\) is
dominant is determined by \(\widetilde C_M\), while the reverse auction when
sell-\(X\) is dominant is determined by \(\widetilde C^{\rm rev}_M\).

\begin{lemma}[The extended auction inherits the one-sided results]
\label{lem:two-sided-extended-admissible}
For every fixed \(M\ge0\), the one-sided AMM clinching auction generated by
\(\widetilde C_M\) is an admissible target-based mechanism with unattained
capacity \(M+x_0\).  Its cost curve is continuous, strictly increasing, and
convex on \([0,M+x_0)\), starts at zero with right derivative \(p_0\), and
diverges at \(M+x_0\).  The reverse one-sided AMM clinching auction generated
by \(\widetilde C^{\rm rev}_M\) is likewise admissible; its cost curve has
the corresponding properties on \([0,p_0M+y_0)\), with initial right
derivative \(1/p_0\) and unattained capacity \(p_0M+y_0\).

Consequently, the single-run structural lemmas, exact-settlement theorem,
and terminal certificate from \cref{sec:one-sided-proofs} apply to the
one-sided auction generated by \(\widetilde C_M\) and, separately, to the
reverse one-sided auction generated by \(\widetilde C^{\rm rev}_M\).  
Same-direction Sybil-robust incentive compatibility (\cref{thm:abstract-one-sided-sybil}) also
applies separately to these two auctions.
\end{lemma}

\begin{proof}
The two branches of \eqref{eq:two-sided-extended-cost} agree at \(M\), and
their slopes meet at \(p_0=C'_+(0)\).  The right derivative is \(p_0\) on
\([0,M)\) and \(a(z-M)\) on \([M,M+x_0)\); it is positive and
nondecreasing.  The other asserted cost properties follow directly from
those of \(C\).  The proof of \cref{prop:amm-target-admissible} uses precisely
these properties of the cost, so it applies with \(C\) replaced by
\(\widetilde C_M\).  The structural lemmas and
\cref{thm:abstract-one-sided-sybil} then apply to the resulting admissible
target.  Likewise, the proofs of
\cref{thm:one-sided-settlement-ir,lem:one-sided-terminal-certificate} use
only this chord construction and the same cost properties.

The inverse of a strictly decreasing convex reserve curve is again strictly
decreasing and convex.  It has the same reserve-limit properties and is
differentiable at \(y_0\), with derivative \(-1/p_0\).  Hence exchanging
the assets gives an admissible AMM, and the argument applies to
\(\widetilde C^{\rm rev}_M\) as well.
\end{proof}

\Cref{lem:two-sided-finite-clock} below 
addresses a boundary case in the reverse auction.  When
sell-\(X\) is dominant, a sell-\(X\) report with reservation value \(u_i=0\)
has transformed one-sided value \(1/u_i=+\infty\).  
Since our earlier analysis of the one-sided case 
allows only finite-value reports, it is convenient to replace this
infinite transformed value by a sufficiently large finite one
without affecting the outcome.

\begin{lemma}[Infinite bids admit finite truncation]
\label{lem:two-sided-finite-clock}
Let \(\widehat C\) be either extended cost, let \(p_\star\) be its initial
clock price, and let \(B_0\) be the total initial budget in that auction.
Set
\begin{equation*}
 z_B:=\widehat C^{-1}(B_0),\qquad
 P_B:=\widehat C'_+(z_B)<\infty.
\end{equation*}
The run terminates at a finite clock price.  Any value report above
\(P_B\) can be replaced by any other value above \(P_B\), or by
\(+\infty\), without changing the outcome.  In particular, an identity
with infinite clock value spends its entire budget.
\end{lemma}

\begin{proof}
At a continuing state after all events at the current price, the structural
lemmas give \(h>0\).  With \(\beta:=Y+B\le B_0\), the chord identity gives
\begin{equation*}
 \widehat C(T)=Y+p(T-X)=\beta-h\le B_0,
 \qquad
 p\le\widehat C'_+(T)\le\widehat C'_+(z_B)=P_B.
\end{equation*}
Here \(z_B\) lies below the capacity and its right derivative is finite,
including when \(B_0=z_B=0\).  Choose a finite
\(H>P_B\) for every infinite report.  With these finite reports the usual
termination rule stops the run no later than its largest value report; the
bound above forces termination before \(H\).  Indeed, a continuing clock
could not pass through the interval \((P_B,H)\).  Reports above \(P_B\)
therefore never exit before termination, and their exact values do not
affect any allocation rule.  This both defines and proves termination for
the infinite-report convention.  Such an identity is still active when the
run ends, unless its budget was already zero; the stopping and exhaustion
rules spend its entire budget.
\end{proof}

\begin{theorem}[Two-sided settlement and individual rationality]
\label{thm:two-sided-settlement-ir}
The two-sided mechanism terminates, is exact-settlement feasible, and is
individually rational.  Specifically, when sell-\(Y\) is dominant,
\begin{equation}
 X_f\ge M,\qquad
 Y_f=\widetilde C_M(X_f)=p_0M+C(X_f-M),
 \qquad 0\le X_f-M<x_0.
 \label{eq:two-sided-proved-settlement}
\end{equation}
For every reported profile, executed prices also satisfy
\begin{equation}
 S_{Y,i}\ge p_0R_{X,i}\quad(i\in\mathcal I_{\rm Y}),
 \qquad
 R_{Y,i}\le p_0S_{X,i}\quad(i\in\mathcal I_{\rm X}).
 \label{eq:two-sided-executed-prices}
\end{equation}
Thus a sell-\(Y\) identity buys \(X\) at prices at least \(p_0\), and a
sell-\(X\) identity sells \(X\) at prices at most \(p_0\).
\end{theorem}

\begin{proof}
When sell-\(Y\) is dominant, \(D_Y\ge D_X\).  The initial target is
\begin{equation}
 \widetilde T_M(p_0,0,0)=M+Z(p_0)=:\Delta_0,
 \label{eq:two-sided-initial-target}
\end{equation}
the length of the extended cost's initial linear segment.  The total
dominant budget is \(B_0=p_0D_Y\).  If the run stops at \(p_0\), the
stopping rule allocates
\(X_f=\min\{\Delta_0,B_0/p_0\}\): it clears the whole initial target
when the budget suffices, and otherwise spends every budget on a prefix of
that linear segment.  Both \(\Delta_0\) and \(B_0/p_0=D_Y\) are at least
\(M\), so \(X_f\ge M\).

If the run continues beyond \(p_0\), its terminal allocation is at least
\(\Delta_0\).  To see this, recall the two settlement cases in the proof
of \cref{thm:one-sided-settlement-ir}.  A fully cleared terminal chord has
allocation at the rightmost chord endpoint, which is at or beyond the
initial linear segment.  A terminal point in a partially allocated lump
lies on a linear segment of slope \(p_f>p_0\), which is also beyond that
initial segment.  The boundary case \(S=0\) is a fully cleared chord.
Thus the minority fill is absorbed in every case.

By \cref{lem:two-sided-extended-admissible,lem:two-sided-finite-clock}, the
dominant run terminates and satisfies
\(Y_f=\widetilde C_M(X_f)\).  Substituting \(X_f\ge M\) gives
\eqref{eq:two-sided-proved-settlement}.  The minority supplies \(M\) units
of \(X\) and receives \(p_0M\) units of \(Y\); the remaining flows are
\(X_{\mathrm{pool}}=X_f-M\) and
\(Y_{\mathrm{pool}}=C(X_f-M)\).  This is exactly
\cref{def:two-sided-exact-settlement}, with no deficit in either asset.

Dominant identities obey their reported budgets and pay at most their
reported values per unit, by the one-sided rules.  Minority identities
receive \(p_0r_i\ge u_ir_i\).  These facts give individual rationality on
truthful profiles.  They also show that buyers pay clock prices at least
\(p_0\), while sellers receive exactly \(p_0\), proving
\eqref{eq:two-sided-executed-prices} in this case.

Exchanging the assets proves the corresponding settlement and budget claims
when sell-\(X\) is dominant.  Sellers then pay at most their reciprocal
value in the exchanged auction, so their original surplus
\(R_{Y,i}-u_iS_{X,i}\) is nonnegative.  For \(u_i=0\) this is immediate.
The ascending clock \(p^{\rm rev}\ge1/p_0\) pays sellers original prices
\(1/p^{\rm rev}\le p_0\), while minority buyers pay exactly \(p_0\).

Finally, if \(D_Y=D_X=M\), the entire dominant budget buys exactly \(M\)
units at \(p_0\).  The initial stopping rule spends every identity's
budget, including those exiting at \(p_0\), and the pool is unchanged.
This also proves that the tie designation is immaterial.
\end{proof}

\subsection{The Initial Fill and the Imbalance}

Recall that \(Z(p_0)\) is the value at \(p_0\) of the marginal-supply function defined
in \eqref{eq:marginal-supply}.  

\begin{lemma}[Guaranteed allocation at the spot price]
\label{lem:two-sided-spot-fill}
Suppose sell-\(Y\) is dominant, and let \(\eta=D_Y-M\).  
At the initial price \(p_0\), first set the
residual supply to the initial target
\(S=\widetilde T_M(p_0,0,0)=M+Z(p_0)\).  Before processing any identities
that exit at \(p_0\), the water level is
\begin{equation}
 h_0=p_0D_Y-p_0\bigl(M+Z(p_0)\bigr)
     =p_0\bigl(\eta-Z(p_0)\bigr)\le p_0\eta.
 \label{eq:two-sided-initial-water}
\end{equation}
Every eligible dominant identity receives at least
\begin{equation*}
 x_i(p_0)\ge\left(\frac{q_i}{p_0}-\eta\right)_+.
\end{equation*}
Consequently a truthful identity \((V,Q)\), \(V\ge p_0\), obtains utility
at least \((V-p_0)(Q/p_0-\eta)_+\).
\end{lemma}

\begin{proof}
Equation \eqref{eq:two-sided-initial-water} follows from
\eqref{eq:two-sided-initial-target}.  Let \(B_E\) be the total budget of
identities exiting at \(p_0\).  If the run continues, restoration lowers
each active budget to at most
\(h^\star=h_0-B_E\le p_0\eta\), so its spot-price purchase is
\((q_i-h^\star)_+/p_0\).  An exiting identity receives nothing in this
case, but \(q_i\le B_E<h_0\le p_0\eta\), so its claimed lower bound is
zero.  If the run stops, every active budget is spent.  The total unspent
exiting budget is \(\max\{h_0,0\}\), so each exiting identity spends at
least \((q_i-\max\{h_0,0\})_+\).  This proves the allocation bound in all
cases.  All later truthful purchases have nonnegative surplus, proving the
utility bound.
\end{proof}

Thus the imbalance measures how much of a dominant bidder's spot-price
demand can remain unfilled.  The pool's own initial linear segment can only
improve this guarantee.  This observation will bound the gain from changing
which side is the minority.

\subsection{The Value of Additional Spot-Price Supply}
\label{subsec:two-sided-fill-shift}

The same-direction Sybil theorem, \cref{thm:abstract-one-sided-sybil},
already shows that, for a fixed extended cost curve, a trader cannot benefit
by splitting its sell-\(Y\) participation among multiple identities or
misreporting their values.  Thus the new issue in the two-sided incentive
proof is the effect of opposite-direction sell-\(X\) identities on the cost
curve faced by the trader's merged sell-\(Y\) identity.

The proof in \cref{subsec:two-sided-incentives} first considers a strategic
trader with true sell-\(Y\) type; the sell-\(X\) case follows by exchanging
the assets.  Suppose the truthful profile has sell-\(Y\) as the dominant
side.  If the trader submits opposite-direction sell-\(X\) identities and
sell-\(Y\) remains dominant, those identities enlarge the minority fill from
\(M\) to \(M+E\), thereby lengthening the initial \(p_0\)-price segment of
the extended cost curve.  This additional supply may benefit the trader's
merged sell-\(Y\) identity, but supplying the \(E\) units of \(X\) through
its opposite-direction identities costs the trader \((V-p_0)E\) in surplus.
\Cref{lem:two-sided-fill-shift} shows that the gain from the enlarged fill
cannot exceed this cost.  A deviation that changes which side is dominant is
handled separately in \cref{subsec:two-sided-incentives}.

Fix outside sell-\(Y\) reports and a true value \(V\ge p_0\).  Let
\(\mathcal U(M,Q)\) denote the utility of the single truthful identity
\((V,Q)\) in the one-sided run on \(\widetilde C_M\).  In this notation
the fill is a parameter of the cost; no minority orders are part of this
one-sided comparison.

\begin{lemma}[Fill-shift bound]
\label{lem:two-sided-fill-shift}
For all \(M,E,Q\ge0\),
\begin{equation*}
 \mathcal U(M+E,Q+p_0E)
 \le \mathcal U(M,Q)+(V-p_0)E.
\end{equation*}
\end{lemma}

The comparison gives the bidder both \(E\) extra units of initial supply
and the \(p_0E\) needed to buy them.  The lemma says that its utility can
increase by no more than the surplus of buying those units at the spot
price.  Because the extra supply can change the whole auction path, the
proof compares mark-to-market gains, as in the one-sided incentive proof.

\begin{proof}
If \(V=p_0\), both utilities are zero.  If \(Q=0\), then
\(\mathcal U(M,0)=0\), and a budget of \(p_0E\) buying only at prices at
least \(p_0\) gives surplus at most \((V-p_0)E\).  Hence assume
\(V>p_0\) and \(Q>0\).

Call the smaller-fill run \(0\) and the larger-fill run \(1\).  Their
costs are \(C_0:=\widetilde C_M\) and
\(C_1:=\widetilde C_{M+E}\), and the focal identity \(k\) has budgets
\(Q\) and \(Q+p_0E\), respectively.  All outsider reports \(N\) are the
same.  Use \(T^r,h^r,\beta^r=Y^r+B^r\) for the state of run \(r\),
and use \(\Gamma_k^r,\Gamma_N^r,W_N^r\) as in
\eqref{eq:mark-to-market-gain} and \eqref{eq:outsider-quantity-wealth}.
After a run ends, allocations and payments stay fixed, effective budgets
are zero, and its target is evaluated at the current comparison price.

The costs satisfy the translation identity
\begin{equation}
 C_1(z+E)=C_0(z)+p_0E.
 \label{eq:two-sided-cost-translation}
\end{equation}
Consequently translating an auction state by \((E,p_0E)\) translates its
rightmost chord endpoint by \(E\).  Define the shift-adjusted gain gap of the
focal identity and the outsider gain gap by
\begin{equation*}
 \GainGap_k:=\Gamma_k^0-\Gamma_k^1+(p-p_0)E,
 \qquad
 \GainGap_N:=\Gamma_N^1-\Gamma_N^0.
\end{equation*}
Both functions are locally absolutely continuous and vanish at \(p_0\)
by \cref{lem:mark-to-market-identity}.  Moreover,
\begin{equation}
 \GainGap_k'=x_k^0-x_k^1+E,\qquad
 p\GainGap_N'=\GainGap_N+Y_N^1-Y_N^0,
 \label{eq:two-sided-shift-derivatives}
\end{equation}
almost everywhere.  At \(p=V\), the desired inequality is exactly
\(\GainGap_k(V)\ge0\).

Write \(\psi_p^r(z):=pz-C_r(z)\).  The chord equality and
\eqref{eq:two-sided-cost-translation} give
\begin{equation*}
 \GainGap_k-\GainGap_N
 =\psi_p^0(T^0)-\psi_p^1(T^1)+(p-p_0)E
 =\psi_p^1(T^0+E)-\psi_p^1(T^1).
\end{equation*}
Both arguments in the last expression are on the decreasing branch of
\(\psi_p^1\): they are rightmost chord endpoints, with the first one
translated as above.  The same strict monotonicity used in the proof of
\cref{prop:amm-target-admissible} therefore implies
\begin{equation}
 \begin{aligned}
  \GainGap_k\le \GainGap_N&\ \Longrightarrow\ T^1\le T^0+E,\\
  \GainGap_k\ge \GainGap_N&\ \Longrightarrow\ T^1\ge T^0+E.
 \end{aligned}
 \label{eq:two-sided-shift-target-implications}
\end{equation}

Let \(p_1\) be the terminal price of run \(0\), or \(V\) if that run
continues to the focal identity's exit.  Consider almost every
\(p\in(p_0,p_1)\).  The outsider accounting from
\cref{lem:abstract-outsider-accounting} uses only the common outsider
reports and the single-run structural lemmas, so it gives
\begin{equation}
 W_N^1-W_N^0\ge \GainGap_N/p.
 \label{eq:two-sided-shift-outsider-wealth}
\end{equation}
There is equality while both runs continue.  If run \(1\) has ended, the
same terminal accounting gives \(Y_N^1\ge Y_N^0\).

When both runs continue, identical outsiders have had identical budgets
cancelled, and the focal identity has not exited.  Thus
\(\beta^1-\beta^0=p_0E\).  Since \(h^r=\beta^r-C_r(T^r)\), the cost
translation gives
\begin{equation*}
 h^0-h^1=C_1(T^1)-C_1(T^0+E).
\end{equation*}
If \(\GainGap_k\ge \GainGap_N\), then \(T^1\ge T^0+E\), hence
\(h^1\le h^0\).
The outsider payment formula \eqref{eq:outsider-payment-water} gives
\(Y_N^1\ge Y_N^0\).  This also holds when run \(1\) has ended, so
\eqref{eq:two-sided-shift-derivatives} yields
\begin{equation}
 \GainGap_k\ge \GainGap_N
 \quad\Longrightarrow\quad
 \GainGap_N'\ge \GainGap_N/p.
 \label{eq:two-sided-shift-outsider-inequality}
\end{equation}

For the other inequality, the supply invariant of the active focal identity
in run \(0\) gives \(x_k^0\ge T^0-W_N^0\), exactly as in the proof of
\cref{lem:abstract-allocation-gap}.  If the focal identity in continuing
run \(1\) has never clinched, then \(x_k^1=0\) and \(\GainGap_k'\ge0\).
Otherwise clincher persistence, or target feasibility if that run has
ended, gives \(x_k^1\le T^1-W_N^1\).  Therefore, when
\(\GainGap_k\le \GainGap_N\),
\eqref{eq:two-sided-shift-target-implications} and
\eqref{eq:two-sided-shift-outsider-wealth} imply
\begin{equation}
 \GainGap_k'\ge\min\{0,\GainGap_N/p\}.
 \label{eq:two-sided-shift-focal-inequality}
\end{equation}
Indeed, in the latter case the sharper bound is
\(\GainGap_k'\ge T^0+E-T^1+W_N^1-W_N^0\ge \GainGap_N/p\).

Equations \eqref{eq:two-sided-shift-outsider-inequality} and
\eqref{eq:two-sided-shift-focal-inequality} are precisely the two
differential inequalities used in the first-crossing argument of
\cref{thm:abstract-one-sided-sybil}.  Applying that already proved argument
to \(\min\{\GainGap_k,\GainGap_N\}\), which starts at zero, gives
\(\GainGap_k,\GainGap_N\ge0\) through \(p_1\).

If \(p_1=V\), this proves the claim.  Otherwise run \(0\) ended below
\(V\), so its focal identity spent all of \(Q\).  No focal budget has
been cancelled in run \(1\) below \(V\); write \(b_k^1\) for its
remaining budget at \(p_1\), taking it to be zero if run \(1\) has ended.
The definition of \(\GainGap_k\) now gives
\begin{equation*}
 0\le \GainGap_k(p_1)
   =p_1\bigl(x_k^0(p_1)-x_k^1(p_1)+E\bigr)-b_k^1.
\end{equation*}
Every later unit in run \(1\) costs at least \(p_1\), so its focal
allocation can never exceed \(x_k^0(p_1)+E\).  Hence \(\GainGap_k'\ge0\)
from \(p_1\) to \(V\), and \(\GainGap_k(V)\ge0\) follows.
\end{proof}

\subsection{Sybil-Robust Incentive Compatibility Across Directions}
\label{subsec:two-sided-incentives}

\begin{theorem}[Two-sided Sybil-robust incentive compatibility]
\label{thm:two-sided-incentives}
The two-sided AMM clinching mechanism satisfies
Sybil-robust incentive compatibility
(\cref{def:incentive-compatibility}) for the full strategy space in
\cref{def:two-sided-strategy}.  
\end{theorem}

\begin{proof}
Fix a physical trader \(k\) and all reports of the other traders.  Truthful
utility is finite and nonnegative by
\cref{thm:two-sided-settlement-ir}, so a deviation giving utility
\(-\infty\) cannot improve it.  For all other deviations we compare the
linear surplus of the aggregate net flows, as required by
\eqref{eq:two-sided-physical-utility-y} and
\eqref{eq:two-sided-physical-utility-x}.

\paragraph{Eligible true sell-\(Y\) type.}
Suppose that trader \(k\) has type \((V,Q)\), with
\(V\ge p_0\).  Remove
ineligible identities from the deviation, since they have no effect.  Let
\[
 Q^{\rm dev}
 :=\sum_{i\in\mathcal I_{\rm Y}(k)\cap\mathcal E_{\rm Y}}q_i
 \le Q
\]
be trader \(k\)'s aggregate eligible sell-\(Y\) budget in the deviating
strategy, and let \(e:=S_X(k)\) be
the amount of \(X\) actually delivered by its sell-\(X\) identities.
Recall from \eqref{eq:physical-net-flows} that \(\Delta_{X,k}\) and
\(\Delta_{Y,k}\) are trader \(k\)'s aggregate net changes in its holdings of
\(X\) and \(Y\).  The executed-price inequalities
\eqref{eq:two-sided-executed-prices} imply
\begin{equation}
 V\Delta_{X,k}+\Delta_{Y,k}
 \le (V-p_0)\left(\frac{Q^{\rm dev}}{p_0}-e\right).
 \label{eq:two-sided-deviation-surplus-bound}
\end{equation}
For clarity, if its sell-\(Y\) identities spend \(s\le Q^{\rm dev}\), their
allocation is at most \(s/p_0\), so their surplus is at most
\((V/p_0-1)s\le(V-p_0)Q^{\rm dev}/p_0\).  Its opposite-direction identities
receive at most \(p_0e\) for an asset worth \(Ve\) to this same physical
trader.  Combining these two contributions proves the deviation-surplus
bound in \eqref{eq:two-sided-deviation-surplus-bound}.  In particular, the case
\(V=p_0\) is immediate: every deviation has nonpositive linear surplus
and truthfulness gives zero.

Suppose \(V>p_0\).  If sell-\(X\) is dominant in the truthful profile,
the focal identity spends all of \(Q\) at \(p_0\) and obtains
\((V-p_0)Q/p_0\).  Equation
\eqref{eq:two-sided-deviation-surplus-bound}, with \(Q^{\rm dev}\le Q\) and
\(e\ge0\), shows that no deviation can exceed this utility.

It remains to consider a truthful profile in which sell-\(Y\) is dominant.
Let \(M\) be the total eligible outside sell-\(X\) budget and let
\(\eta=D_Y-M\).  The truthful utility is
\(\mathcal U(M,Q)\), and \cref{lem:two-sided-spot-fill} gives
\begin{equation}
 \mathcal U(M,Q)\ge(V-p_0)(Q/p_0-\eta)_+.
 \label{eq:two-sided-truthful-utility-floor}
\end{equation}
There are two cases.

\begin{itemize}
\item \emph{The deviation leaves sell-\(Y\) dominant.}
All eligible opposite-direction identities are then filled in full at
\(p_0\).  Write \(E=e\) for their total \(X\)-budget.  The dominant
cost becomes \(\widetilde C_{M+E}\).  On this fixed cost,
\cref{thm:abstract-one-sided-sybil,lem:two-sided-extended-admissible}
bound the group's sell-\(Y\) surplus by that of a single truthful identity
with budget \(Q+p_0E\), since \(Q^{\rm dev}\le Q\le Q+p_0E\).  After accounting
for the opposite-direction deliveries, the full linear surplus is at most
\begin{equation*}
 \mathcal U(M+E,Q+p_0E)-(V-p_0)E
 \le\mathcal U(M,Q),
\end{equation*}
where the last step is \cref{lem:two-sided-fill-shift}.  The larger budget
here is only an upper bound used in the proof; the admissible deviation
still has the original cap \(Q^{\rm dev}\le Q\).

\item \emph{The deviation makes sell-\(X\) dominant.}
Its new total demand is
\(D_Y^{\rm dev}=D_Y-Q/p_0+Q^{\rm dev}/p_0\).  By the exchanged version of
\cref{thm:two-sided-settlement-ir}, the dominant sell-\(X\) identities
deliver at least \(D_Y^{\rm dev}\) units of \(X\) to cover that minority fill.
The outside sell-\(X\) identities can deliver at most \(M\), hence
\begin{equation*}
 e\ge D_Y^{\rm dev}-M=Q^{\rm dev}/p_0-Q/p_0+\eta.
\end{equation*}
Combining this with \eqref{eq:two-sided-deviation-surplus-bound} gives
linear surplus at most
\((V-p_0)(Q/p_0-\eta)\), which is bounded above by the truthful utility
in \eqref{eq:two-sided-truthful-utility-floor}.  This case also includes a
change of sides caused solely by withholding same-direction budget.
\end{itemize}

\paragraph{Eligible true sell-\(X\) type.}
First suppose \(0<U\le p_0\), and exchange the assets.  The trader's
truthful value and budget in the exchanged auction are \((1/U,R)\).
Admissibility becomes exactly the sell-\(Y\) constraint already considered:
the same-direction budgets sum to at most \(R\), while the opposite
budgets are unrestricted.  Its linear surplus in the exchanged numeraire is
\begin{equation*}
 \frac{\Delta_{Y,k}+U\Delta_{X,k}}{U},
\end{equation*}
a positive multiple of its original surplus.  Thus the preceding proof
applies.  Any zero outside reservation reports can be represented by a
common finite clock value above the bounds for the runs being compared,
using \cref{lem:two-sided-finite-clock}.

For \(U=0\), take true values \(\varepsilon>0\) tending to zero, with
the same budget \(R\) and fixed outside reports and deviation.  For all
sufficiently small \(\varepsilon\), the truthful report
\((\varepsilon,R)\) produces exactly the outcome of \((0,R)\): if it is
minority its fill is unchanged, and if it is dominant its reciprocal value
exceeds the finite clock bound.  The inequality just proved for linear
surpluses is
\(R_Y^{\rm tr}-\varepsilon S_X^{\rm tr}
\ge\Delta_{Y,k}^{\rm dev}+\varepsilon\Delta_{X,k}^{\rm dev}\).
Letting \(\varepsilon\downarrow0\) proves the required bound for every
finite-utility deviation of a zero-reservation type as well.

\paragraph{Ineligible true types.}
Finally, although our standing convention considers eligible true types,
the definition also quantifies over ineligible ones.  Summing
\eqref{eq:two-sided-executed-prices} over a physical trader's identities
gives, under every strategy,
\begin{equation}
 \Delta_{Y,k}\le-p_0\Delta_{X,k}.
 \label{eq:two-sided-net-price-bound}
\end{equation}
A sell-\(Y\) type with \(V<p_0\) has nonpositive linear surplus when
\(\Delta_{X,k}\ge0\).  When \(\Delta_{X,k}<0\), finite utility requires
\(\Delta_{Y,k}\le0\), so its surplus is again nonpositive.  A sell-\(X\)
type with \(U>p_0\) has surplus at most
\((U-p_0)\Delta_{X,k}\le0\) when \(\Delta_{X,k}\le0\); if
\(\Delta_{X,k}>0\), \eqref{eq:two-sided-net-price-bound} makes the trade
strictly opposite to its type and its utility is \(-\infty\).
Truthful ineligible reports receive zero.  This completes the proof for
every true type and for the stated net-flow utility functions.
\end{proof}

The two cases explain why opposite-direction identities do not help.  If
they remain on the minority side, their full delivery costs the trader at
least the extra utility made possible by the enlarged fill.  If they change
which side is dominant, conservation forces enough actual delivery to
offset the improvement in the trader's buying price.  Both comparisons
value every identity's trades at the physical trader's true value.

\subsection{Weak Pareto Optimality for Eligible Bids}
\label{subsec:two-sided-weak-pareto}

\begin{theorem}[Two-sided weak Pareto optimality]
\label{thm:two-sided-weak-pareto}
Under truthful reporting, the two-sided AMM clinching outcome is weakly
Pareto optimal for eligible bids in the sense of \cref{def:weak-pareto}.
\end{theorem}

\begin{proof}
Use the holdings and coefficients \((H_i^X,H_i^Y,t_i)\) already defined in
\eqref{eq:post-batch-holdings}.  We first show that there is a positive
cutoff \(\lambda\) such that
\begin{equation}
 H_i^Y>0\ \Longrightarrow\ t_i\le\lambda,
 \qquad
 H_i^X>0\ \Longrightarrow\ t_i\ge\lambda.
 \label{eq:two-sided-pareto-cutoff}
\end{equation}
Thus everyone who can pay \(Y\) values \(X\) no more highly than anyone
who can supply it.

When sell-\(Y\) is dominant, let
\begin{equation*}
 \lambda:=
 \begin{cases}
  \max\{V_i:Q_i-S_{Y,i}>0,\ i\in\mathcal E_{\rm Y}\},
       &\text{if this set is nonempty},\\
  p_0,&\text{otherwise}.
 \end{cases}
\end{equation*}
Minority sell-\(X\) users have holdings \((0,p_0R_i)\) and values
\(U_i\le p_0\le\lambda\).  Dominant users with remaining \(Y\) have
values at most \(\lambda\) by definition.  Every holder of \(X\) is a
dominant user; its value is at least \(\lambda\) by the first part of
\cref{lem:one-sided-terminal-certificate} applied to
\(\widetilde C_M\), or by eligibility if no dominant user has residual
budget.  This proves \eqref{eq:two-sided-pareto-cutoff} in this case.

When sell-\(X\) is dominant, let
\begin{equation*}
 \lambda:=
 \begin{cases}
  \min\{U_i:R_i-S_{X,i}>0,\ i\in\mathcal E_{\rm X}\},
       &\text{if this set is nonempty},\\
  p_0,&\text{otherwise}.
 \end{cases}
\end{equation*}
A zero-reservation dominant user spends all of its \(X\)-budget by
\cref{lem:two-sided-finite-clock}, so \(0<\lambda\le p_0\).
Minority sell-\(Y\) users have holdings \((Q_i/p_0,0)\) and values
\(V_i\ge p_0\ge\lambda\).  Every other holder of \(X\) has value at
least \(\lambda\) by definition.  A dominant user \(j\) holding positive
\(Y\) must have \(U_j\le U_i\) for every dominant user \(i\) with
residual \(X\): for positive reservation values this is the first
terminal-certificate statement with reciprocal values, and for \(U_j=0\)
it is immediate.  If no such residual user exists, eligibility gives
\(U_j\le p_0=\lambda\).  Thus
\eqref{eq:two-sided-pareto-cutoff} holds in this case too.

The internal-trade argument in \cref{thm:one-sided-weak-pareto} now applies
to all eligible users together.  To spell out its use of the cutoff, a
user paying \(Y\) can pay at most \(\lambda\) times the \(X\) it
receives while weakly improving its utility.  A user supplying \(X\)
requires at least \(\lambda\) times that quantity in \(Y\).  Conservation
of both tokens forces equality in all these bounds and leaves every utility
unchanged.  Hence there is no strict Pareto improvement through internal
trade, even when the trading subset includes both original directions.

It remains to check each user's permitted same-direction pool trade.  When
sell-\(Y\) is dominant, set \(z:=X_f-M\ge0\).  A dominant user with
unspent \(Y\) satisfies
\begin{equation*}
 V_i\le\widetilde C'_{M,+}(X_f)=a(z)
\end{equation*}
by the second terminal-certificate statement.  From the settled pool state,
receiving another \(\delta>0\) units of \(X\) costs
\(C(z+\delta)-C(z)\ge a(z)\delta\), so that user cannot gain.
Minority sell-\(X\) users have no remaining \(X\) to tender in their own
direction, because their endowments were filled in full.

When sell-\(X\) is dominant, write
\(w:=\sum_{i\in\mathcal E_{\rm X}}R_{Y,i}
       -\sum_{i\in\mathcal E_{\rm Y}}Q_i\ge0\)
for the net \(Y\) released by the pool, and put
\(a^{\rm rev}:=(C^{\rm rev})'_+\).  A dominant user with residual \(X\) has
\(U_i>0\), and the exchanged terminal certificate gives
\begin{equation*}
 1/U_i\le a^{\rm rev}(w).
\end{equation*}
Receiving an additional \(\delta>0\) units of \(Y\) costs at least
\(C^{\rm rev}(w+\delta)-C^{\rm rev}(w)\ge a^{\rm rev}(w)\delta\) units of \(X\).
Its utility change is therefore at most
\((1-U_i a^{\rm rev}(w))\delta\le0\).  Minority sell-\(Y\) users have spent
all of their \(Y\)-budgets, so they cannot make an additional trade in
their own direction.  This proves the second clause of
\cref{def:weak-pareto} and completes the proof.
\end{proof}

\section{Impossibility with Budget Overreporting}
\label{sec:budget-overreporting}

The positive incentive results above restrict a trader's aggregate declared
budget in its true direction to its actual budget.  We now show that this
restriction is essential.  Suppose instead that a trader may declare a
larger budget, with the deviation regarded as feasible whenever the amount
actually paid does not exceed its true budget.  Even without Sybil identities
or opposite-direction orders, this enlarged strategy space is incompatible
with the other desired properties.

Under the enlarged strategy space considered here, a true eligible type
\((V,Q)\) may replace its report by any single finite eligible report
\((v,q)\).  If that report produces
\((x,y)\), the deviation has finite utility precisely when \(y\le Q\).
Incentive compatibility with budget overreporting therefore requires
\[
 y\le Q
 \quad\Longrightarrow\quad
 Vx^{\rm tr}-y^{\rm tr}\ge Vx-y,
\]
where \((x^{\rm tr},y^{\rm tr})\) is the truthful outcome.  
In particular, 
when the mechanism does not require 
advance escrow of the entire declared budget, budget overreporting  
becomes possible as long as its realized payment remains affordable.

\begin{theorem}[Budget-overreporting impossibility]
\label{thm:budget-overreporting-impossibility}
For a constant-product AMM, there is no deterministic mechanism that, on
every profile of two eligible sell-\(Y\) users, simultaneously satisfies:
\begin{enumerate}
\item individual rationality (\cref{def:two-sided-ir});
\item exact settlement (\cref{def:two-sided-exact-settlement});
\item both clauses of weak Pareto optimality in \cref{def:weak-pareto}; and
\item incentive compatibility against single-identity eligible value and
budget misreports, including budget overreports whose realized payment does
not exceed the true budget.
\end{enumerate}
All true types and reports used in the proof are eligible and have finite
values and budgets.
\end{theorem}

We stress that this impossibility holds in a strong sense.  It assumes
neither symmetry nor resistance to Sybil or opposite-direction deviations,
and it already applies to the much smaller class of deviations that replace
one sell-\(Y\) report by another.
As a practical implication, the impossibility holds
even when the underlying consensus protocol provides fair sequencing~\cite{decentral-seq,orderfair00,orderfair01,orderfair02,orderfair03,orderfair04}
and censorship resilience~\cite{latencyCR,mcp,focil}.

\subsection{Normalization and Basic Consequences}

By rescaling, we may assume without loss of generality that
\(x_0=y_0=1\), and hence \(p_0=1\), throughout the proof.  The pool then releases
\[
 F(y):=C^{-1}(y)=\frac{y}{1+y},
 \qquad
 F'(y)=\frac{1}{(1+y)^2},
\]
and a sell-\(Y\) value is eligible exactly when \(V\ge1\).

For \(V>1\), define
\[
 \operatorname{Util}^{\rm pool}(V,y):=VF(y)-y,
 \qquad
 y^\star(V):=\sqrt V-1,
 \qquad
 \bar y(V):=V-1.
\]
Here \(\operatorname{Util}^{\rm pool}(V,y)\) is the utility of a
value-\(V\) user that receives all pool output generated by a payment
\(y\).  It increases up to its unique
maximizer \(y^\star(V)\), then decreases; its two zeros are \(0\) and
\(\bar y(V)\).  In particular,
\[
 0<y^\star(V)<\bar y(V),
 \qquad
 \operatorname{Util}^{\rm pool}\bigl(V,y^\star(V)\bigr)=(\sqrt V-1)^2.
\]
Set \(y^\star(1)=\bar y(1)=0\).

The first lemma is a strict priority property.  A higher-valued user that
still has some \(Y\) available could profitably purchase any allocation held
by a lower-valued user.  Weak Pareto optimality therefore prevents the lower
value from receiving anything.

\begin{lemma}[Internal-trade priority]
\label{lem:overreport-priority}
Consider a truthful profile of two eligible sell-\(Y\) users.  Denote their
values and budgets by \((V_1,q_1)\) and \((V_2,q_2)\), and their respective
allocations and payments by \((x_1,y_1)\) and \((x_2,y_2)\).  If
\(V_1>V_2\) and \(y_1<q_1\), then
\(x_2=y_2=0\).
\end{lemma}

\begin{proof}
Suppose \(x_2>0\).  Choose a price \(P\in(V_2,V_1)\) and a sufficiently
small \(\delta>0\) such that
\(P\delta\le q_1-y_1\) and \(\delta<x_2\).  Let user~1 give
\(P\delta\) units of \(Y\) to user~2 in exchange for \(\delta\) units of
\(X\).  The utility gains are
\((V_1-P)\delta>0\) and \((P-V_2)\delta>0\), and both users retain
nonnegative holdings.  This contradicts the internal-trade clause of weak
Pareto optimality.  Hence \(x_2=0\), and individual rationality with
\(V_2>0\) forces \(y_2=0\).
\end{proof}

We will repeatedly use two other consequences of individual rationality,
settlement, and weak Pareto optimality.  Write \(Y=y_1+y_2\).  First, if
either user has unspent budget, a small additional payment to the terminal
pool cannot be profitable, so
\begin{equation}
 \begin{aligned}
 y_1<q_1&\quad\Longrightarrow\quad V_1F'(Y)\le1,\\
 y_2<q_2&\quad\Longrightarrow\quad V_2F'(Y)\le1.
 \end{aligned}
 \label{eq:overreport-pool-condition}
\end{equation}
Second, if \(V_{\max}:=\max\{V_1,V_2\}>1\), individual rationality and
exact settlement give
\begin{equation}
 F(Y)=x_1+x_2
 \ge \frac{y_1}{V_1}+\frac{y_2}{V_2}
 \ge\frac{Y}{V_{\max}},
 \qquad\text{and hence}\qquad
 Y\le\bar y(V_{\max}).
 \label{eq:overreport-total-payment-bound}
\end{equation}

We next derive several standard consequences of value incentive compatibility
when the reported budget is held fixed.  In particular, incentive
compatibility implies the continuity properties used later; we do not
separately assume that the mechanism or its allocation and payment rules are
continuous.

\begin{lemma}[Fixed-budget Myerson lemma~\cite{myerson}]
\label{lem:overreport-value-comparisons}
Fix one user's declared budget \(q\) and the other user's report.  Let
\((x(V),y(V))\) be its outcome when its reported value is \(V\), and put
\(\operatorname{Util}(V):=Vx(V)-y(V)\).  On the eligible interval
\([1,\infty)\), the functions \(x\), \(y\), and \(\operatorname{Util}\)
are nondecreasing, and \(\operatorname{Util}\) is
nonnegative, convex, and continuous (with one-sided continuity at \(1\)).
Moreover, the standard single-parameter payment identity
and its normalized-utility form are, for
\(1\le V<W\),
\begin{equation}
 \begin{aligned}
 \operatorname{Util}(W)-\operatorname{Util}(V)
   &=\int_V^W x(t)\,\dif t, %
\qquad
 \frac{\operatorname{Util}(W)}{W}
   -\frac{\operatorname{Util}(V)}{V}
   &=\int_V^W\frac{y(t)}{t^2}\,\dif t.
 \end{aligned}
 \label{eq:overreport-value-envelope}
\end{equation}
\end{lemma}

\begin{proof}
For eligible values \(1\le V<W\), incentive compatibility in both directions gives
\[
 x(W)\ge x(V),
 \qquad
 V\bigl(x(W)-x(V)\bigr)
 \le y(W)-y(V)
 \le W\bigl(x(W)-x(V)\bigr).
\]
Thus allocation and payment are nondecreasing.  Moreover,
\[
 \operatorname{Util}(V)=\sup_{v\ge1}\{Vx(v)-y(v)\},
\]
with the supremum attained at \(v=V\).  It is therefore a finite convex
nondecreasing function, and hence continuous; individual rationality makes
it nonnegative.

The same incentive inequalities bound
\(\operatorname{Util}(W)-\operatorname{Util}(V)\) between
\((W-V)x(V)\) and \((W-V)x(W)\).  Summing over partitions and taking
upper and lower sums gives the first identity in
\eqref{eq:overreport-value-envelope}.  Thus \(\operatorname{Util}\) is
locally absolutely continuous and
\(\operatorname{Util}'(V)=x(V)\) almost everywhere.  Since
\(\operatorname{Util}(V)=Vx(V)-y(V)\), it follows almost everywhere that
\[
 \frac{\dif}{\dif V}\left(\frac{\operatorname{Util}(V)}{V}\right)
 =\frac{Vx(V)-\operatorname{Util}(V)}{V^2}
 =\frac{y(V)}{V^2}.
\]
Integrating proves the second identity.
\end{proof}

The next lemma establishes the spot-price normalization needed below using
only eligible reports: when one user's value equals the initial spot price,
that user receives nothing.
Intuitively, such a user gains nothing from the pool, whose average price
exceeds the spot price, so it could gain only by being subsidized by the
other user; weak Pareto optimality alone does not exclude this when the
other user has exhausted its budget.  The proof instead squeezes the other
user's utility: incentive compatibility, through the normalized-utility
identity of \cref{lem:overreport-value-comparisons}, guarantees it at least
the utility it would obtain by trading alone against the pool, whereas
feasibility caps it at that amount, with a strict loss whenever the
spot-value user receives any allocation.

\begin{lemma}[A spot-value user receives zero]
\label{lem:overreport-spot-zero}
For any budgets \(q,Q\ge0\) and any eligible value \(V\ge1\), consider the
profile
\[
 ((V,q),(1,Q)).
\]
The spot-value user receives \((x_2,y_2)=(0,0)\), while the first user
receives the sole-user optimum
\[
 y_1=\min\{q,y^\star(V)\},
 \qquad
 x_1=F(y_1).
\]
The same conclusion holds with the two user positions interchanged.
\end{lemma}

\begin{proof}
Fix the second report \((1,Q)\), vary the first user's reported value
\(V\ge1\), and write its outcome as \((x_1(V),y_1(V))\) and its reported
utility as
\(\operatorname{Util}_1(V):=Vx_1(V)-y_1(V)\).  At \(V=1\), both values equal
the spot value.  Individual rationality and exact settlement imply
\(F(Y)=x_1+x_2\ge Y\).  Since \(F(Y)<Y\) for \(Y>0\), both users receive
zero and \(\operatorname{Util}_1(1)=0\).

Now let \(V>1\).  If the first user does not exhaust its budget,
\cref{lem:overreport-priority} gives \(x_2=y_2=0\), and
\eqref{eq:overreport-pool-condition} gives
\(VF'(y_1(V))\le1\), or equivalently
\(y_1(V)\ge y^\star(V)\).  If it does exhaust its budget, \(y_1(V)=q\).
Consequently
\[
 y_1(V)\ge\min\{q,y^\star(V)\}.
\]
The normalized-utility identity in
\eqref{eq:overreport-value-envelope} now gives
\begin{align*}
 \frac{\operatorname{Util}_1(V)}{V}
 &\ge\int_1^V\frac{\min\{q,y^\star(t)\}}{t^2}\,\dif t\\
 &=\max_{0\le z\le q}\left\{F(z)-\frac{z}{V}\right\}.
\end{align*}
For the equality, the maximizer on the right is
\(z=\min\{q,y^\star(V)\}\); differentiating both sides with respect to
\(V\), starting from their common value zero at \(V=1\), gives the same
integrand.

Conversely, individual rationality for the spot-value user gives
\(x_2\ge y_2\).  Exact settlement and \(F'\le1\) therefore imply
\begin{align*}
 \frac{\operatorname{Util}_1(V)}{V}
 &=F(y_1+y_2)-x_2-\frac{y_1}{V}\\
 &\le F(y_1+y_2)-y_2-\frac{y_1}{V}\\
 &\le F(y_1)-\frac{y_1}{V}\\
 &\le\max_{0\le z\le q}\left\{F(z)-\frac{z}{V}\right\}.
\end{align*}
The second inequality is strict when \(y_2>0\).  Equality with the lower
bound forces \(y_2=0\), then \(x_2=0\), and strict concavity of \(F\)
forces \(y_1=\min\{q,y^\star(V)\}\).  Interchanging the fixed user
positions gives the same proof; no symmetry assumption is needed.
\end{proof}

Consequently, no eligible report can buy \(X\) below the spot price: a user
with the spot value, which obtains zero utility by
\cref{lem:overreport-spot-zero}, could otherwise profit by imitating that
report.

\begin{corollary}[Spot-price payment bound]
\label{cor:overreport-spot-payment-bound}
Every profile of two eligible users satisfies
\[
 x_1\le y_1
 \qquad\text{and}\qquad
 x_2\le y_2.
\]
In particular, a user that pays zero receives zero.
\end{corollary}

\begin{proof}
Hold user~2's eligible report fixed.  If user~1 had the spot value \(1\)
and its current budget, \cref{lem:overreport-spot-zero} would give it zero
utility.  It can instead submit its current eligible value report without
changing its budget.  Incentive compatibility for the spot-value true type
therefore gives \(0\ge x_1-y_1\).  Interchanging the two user positions
gives \(0\ge x_2-y_2\).
\end{proof}

\subsection{Two Large-Budget Comparisons}

We now establish two complementary facts.  In the first, a large declared
budget belongs to the lower-valued user.  In the second, it belongs to the
higher-valued user.  Together they show that a declared budget can control
the allocation even when most of that budget is never paid.

\Cref{lem:overreport-lower-value} below considers the former case.  It shows
that a lower-valued user with a large budget sets the price for a
budget-constrained higher-valued user, much as a second price would: the
higher-valued user spends its entire budget at a unit price equal to the
lower value.  The lower-valued user takes the remaining supply at a strictly
positive profit, and its payment stays bounded no matter how large its
declared budget is.

\begin{lemma}[Large budget on the lower-valued user]
\label{lem:overreport-lower-value}
Fix \(V_1>1\), \(0<Q_1<\bar y(V_1)\), and
\(Q^{\rm large}>\bar y(V_1)\), and put
\[
 V_{\rm cut}:=\frac{Q_1}{F(Q_1)}=1+Q_1,
\]
so that \(1<V_{\rm cut}<V_1\).
For every \(V_2\in(V_{\rm cut},V_1)\), the profile
\(((V_1,Q_1),(V_2,Q^{\rm large}))\) satisfies
\begin{equation}
 y_1=Q_1,\qquad x_1=\frac{Q_1}{V_2},\qquad
 y_2\le\bar y(V_2)-Q_1,\qquad
 V_2x_2-y_2>0.
 \label{eq:overreport-lower-value-outcome}
\end{equation}
\end{lemma}

\begin{proof}
Fix \(V_2\), vary the first user's value \(v\in[1,V_1]\) while holding its
budget at \(Q_1\), and keep the second report
\((V_2,Q^{\rm large})\) fixed.  Write
\((x_1(v),y_1(v))\) for the first user's outcome and
\(\operatorname{Util}_1(v):=vx_1(v)-y_1(v)\).  By
\eqref{eq:overreport-total-payment-bound}, total payment never exceeds
\(\bar y(V_1)<Q^{\rm large}\), so the second user always has unspent budget.

When \(v<V_2\), internal-trade priority gives
\((x_1(v),y_1(v))=(0,0)\); continuity therefore gives
\(\operatorname{Util}_1(V_2)=0\).
Suppose \(y_1(v_0)<Q_1\) for some \(v_0>V_2\).  Payment monotonicity then
gives \(y_1(v)<Q_1\) throughout a right neighborhood of \(V_2\).  The first
user is the higher-valued user there and also has unspent budget, so priority
gives the second user zero.  Exact settlement and the pool condition imply
\[
 x_1(v)=F(y_1(v)),
 \qquad
 y_1(v)\ge y^\star(v),
 \qquad
 \operatorname{Util}_1(v)
 =\operatorname{Util}^{\rm pool}\bigl(v,y_1(v)\bigr).
\]
The monotone limit
\(y_{\rm lim}:=\lim_{v\downarrow V_2}y_1(v)\) consequently
satisfies
\[
 y^\star(V_2)\le y_{\rm lim}<Q_1<\bar y(V_2).
\]
Continuity of \(\operatorname{Util}_1\) would give
\(\operatorname{Util}^{\rm pool}(V_2,y_{\rm lim})
=\operatorname{Util}_1(V_2)=0\), contradicting the fact that
\(\operatorname{Util}^{\rm pool}(V_2,\cdot)\) is strictly positive between
its two zeros.
Hence \(y_1(v)=Q_1\) for every \(v\in(V_2,V_1]\).

With payment fixed at \(Q_1\), incentive compatibility forces allocation to
be constant on this interval, say \(x_1(v)=x_{\rm const}\).  Taking
\(v\downarrow V_2\) in
\(\operatorname{Util}_1(v)=v x_{\rm const}-Q_1\) gives
\(V_2x_{\rm const}-Q_1=\operatorname{Util}_1(V_2)=0\), so
\(x_{\rm const}=Q_1/V_2\).  This proves the first two claims in
\eqref{eq:overreport-lower-value-outcome}.  Individual rationality for
user~2 and exact settlement give
\[
 F(Q_1+y_2)=\frac {Q_1}{V_2}+x_2\ge\frac{Q_1+y_2}{V_2},
\]
and therefore \(Q_1+y_2\le\bar y(V_2)\).

It remains to show that user~2 has strictly positive utility.  Keep the
first report \((V_1,Q_1)\) fixed, vary user~2's reported value \(v\) while
holding its declared budget at \(Q^{\rm large}\), and write its reported utility
\(vx_2-y_2\) as \(\operatorname{Util}_2(v)\).  If
\(\operatorname{Util}_2(V_2)=0\), nonnegativity and
monotonicity give \(\operatorname{Util}_2(v)=0\) for every
\(v\in(V_{\rm cut},V_2)\).  Choose such a \(v\).  The conclusions already
proved give \(y_1=Q_1\) and \(x_1=Q_1/v\).  Zero utility for user~2 and exact
settlement then give
\[
 F(Q_1+y_2)=\frac{Q_1+y_2}{v},
 \qquad
 y_2=\bar y(v)-Q_1>0.
\]
A true type \((V_2,Q^{\rm large})\) could instead report
\((v,Q^{\rm large})\) and obtain
\((V_2-v)x_2>0\), contradicting incentive compatibility.  Thus
\(\operatorname{Util}_2(V_2)>0\).
\end{proof}

\Cref{lem:overreport-higher-value} below considers the mirror case, in which
the higher-valued user has a large budget.  It shows that this user
effectively trades alone against the pool, while the lower-valued user
receives nothing.  The only thing the mechanism can do to the higher-valued
user is impose a minimum purchase, and incentive compatibility keeps the
resulting overpayment small when the lower-valued user's budget is small.

\begin{lemma}[Large budget on the higher-valued user]
\label{lem:overreport-higher-value}
Fix \(V_1>V_2>1\), \(Q_2>0\), and
\(Q^{\rm large}>\bar y(V_1)\).  For
\(v\in[V_2,V_1]\), consider the profile
\(((v,Q^{\rm large}),(V_2,Q_2))\), and write the first user's outcome as
\((x_1(v),y_1(v))\).  For \(v>V_2\), user~2 receives zero and
\[
 x_1(v)=F(y_1(v)),
 \qquad
 y_1(v)\ge y^\star(v).
\]
The limit \(y_{\rm lim}:=\lim_{v\downarrow V_2}y_1(v)\) exists and satisfies
\begin{equation}
 y_{\rm lim}\ge y^\star(V_2),
 \qquad
 \operatorname{Util}^{\rm pool}(V_2,y_{\rm lim})
 \ge \operatorname{Util}^{\rm pool}\bigl(V_2,y^\star(V_2)\bigr)-V_2Q_2,
 \label{eq:overreport-limit-bound}
\end{equation}
and, for every \(v\in(V_2,V_1]\),
\begin{equation}
 y_1(v)=\max\{y_{\rm lim},y^\star(v)\}.
 \label{eq:overreport-payment-shape}
\end{equation}
\end{lemma}

\begin{proof}
The total-payment bound gives
\(Y\le\bar y(V_1)<Q^{\rm large}\), so the first user
never exhausts its declared budget.  For \(v>V_2\), internal-trade priority
therefore gives user~2 zero.  Settlement and the pool condition yield the
two displayed assertions.  Monotonicity and boundedness of \(y_1\) give
the limit \(y_{\rm lim}\ge y^\star(V_2)\).  If
\(\operatorname{Util}_1(v):=vx_1(v)-y_1(v)\), continuity gives
\[
 \operatorname{Util}_1(V_2)
 =\lim_{v\downarrow V_2}
   \operatorname{Util}^{\rm pool}\bigl(v,y_1(v)\bigr)
 =\operatorname{Util}^{\rm pool}(V_2,y_{\rm lim}).
\]

To lower-bound this utility, choose a finite eligible value
\(v_{\rm aux}\) so large that
\[
 v_{\rm aux}F'\bigl(y^\star(V_2)+Q_2\bigr)>1.
\]
A true type \((V_2,Q^{\rm large})\) can deviate to the report
\((v_{\rm aux},y^\star(V_2))\).  In the resulting profile
\(((v_{\rm aux},y^\star(V_2)),(V_2,Q_2))\), the first user
must pay all \(y^\star(V_2)\): otherwise it would have unspent budget and
\eqref{eq:overreport-pool-condition}, together with
\(Y<y^\star(V_2)+Q_2\), would imply
\[
 1\ge v_{\rm aux}F'(Y)
 \ge v_{\rm aux}F'\bigl(y^\star(V_2)+Q_2\bigr)>1.
\]
By \cref{cor:overreport-spot-payment-bound}, user~2 receives at most
\(x_2\le y_2\le Q_2\).  Exact settlement therefore gives the deviating
user at least \(F(y^\star(V_2))-Q_2\) units of \(X\).  The deviation is affordable
because it pays \(y^\star(V_2)<Q^{\rm large}\), and incentive compatibility yields
\begin{align*}
 \operatorname{Util}^{\rm pool}(V_2,y_{\rm lim})
 &=\operatorname{Util}_1(V_2)\\
 &\ge V_2\bigl(F(y^\star(V_2))-Q_2\bigr)-y^\star(V_2)\\
 &=\operatorname{Util}^{\rm pool}\bigl(V_2,y^\star(V_2)\bigr)-V_2Q_2,
\end{align*}
which proves \eqref{eq:overreport-limit-bound}.  Notice that
\(v_{\rm aux}\) is finite; the proof does not use an infinite value report.

It remains to identify the payment path.  Payment monotonicity already gives
\(y_1(v)\ge y_{\rm lim}\), and the pool condition gives
\(y_1(v)\ge y^\star(v)\).  The function \(y_1\) is continuous on
\((V_2,V_1]\): its one-sided limits exist by monotonicity, and continuity
of
\(\operatorname{Util}_1(v)
=\operatorname{Util}^{\rm pool}\bigl(v,y_1(v)\bigr)\), together with
strict decrease of
\(\operatorname{Util}^{\rm pool}(v,\cdot)\) on
\([y^\star(v),\infty)\), forces those limits to agree.

If \(y_1(v)>\max\{y_{\rm lim},y^\star(v)\}\) at some \(v\), continuity and the
definition of \(y_{\rm lim}\) provide \(v'\in(V_2,v)\) such that
\[
 \max\{y_{\rm lim},y^\star(v)\}<y_1(v')<y_1(v).
\]
Both reports give the first user the entire pool output.  Because
\(\operatorname{Util}^{\rm pool}(v,\cdot)\) strictly decreases above
\(y^\star(v)\), a true value \(v\)
would strictly prefer the report \(v'\), contradicting incentive
compatibility.  This proves \eqref{eq:overreport-payment-shape}.
\end{proof}

\subsection{The Contradiction}

The proof constructs a truthful profile of two eligible users in which each
user has an affordable budget overreport, that is, one whose realized payment
does not exceed its true budget.  By \cref{lem:overreport-higher-value}, the
first user's overreport guarantees it utility \(1\), which is the entire
surplus that the batch can generate.  By \cref{lem:overreport-lower-value},
the second user's overreport guarantees it strictly positive utility.
Incentive compatibility requires the truthful outcome to give each user at
least what its overreport guarantees, and these two guarantees together
exceed the available surplus.

\begin{proof}[Proof of \cref{thm:budget-overreporting-impossibility}]
In normalized units, consider the truthful profile
\[
 (V_1,Q_1)=(4,1),
 \qquad
 (V_2,Q_2)=\left(\frac{100}{49},\frac{2}{49}\right),
\]
and let \(Q^{\rm large}=4\) be the larger budget used in the two deviations below.
Both values are strictly eligible.  Write
\(\operatorname{Util}_1^{\rm tr}:=V_1x_1^{\rm tr}-y_1^{\rm tr}\) and
\(\operatorname{Util}_2^{\rm tr}:=V_2x_2^{\rm tr}-y_2^{\rm tr}\).

\paragraph{The first user can safely overreport its budget.}
Apply \cref{lem:overreport-higher-value}; its hypotheses hold because
\(V_1>V_2>1\) and
\(Q^{\rm large}=4>\bar y(V_1)=3\).  We have
\[
 y^\star(V_2)=\frac37,
 \qquad
 \operatorname{Util}^{\rm pool}\bigl(V_2,y^\star(V_2)\bigr)=\frac9{49},
 \qquad
 \operatorname{Util}^{\rm pool}(V_2,1)=\frac1{49}.
\]
Consequently,
\[
 \operatorname{Util}^{\rm pool}(V_2,y_{\rm lim})
 \ge \frac9{49}-\frac{100}{49}\cdot\frac2{49}
 =\frac{241}{2401}
 >\frac1{49}=\operatorname{Util}^{\rm pool}(V_2,1).
\]
Since \(y_{\rm lim}\ge y^\star(V_2)\) and
\(\operatorname{Util}^{\rm pool}(V_2,\cdot)\) strictly decreases after its
maximizer, this implies \(y_{\rm lim}<1\).  Equation
\eqref{eq:overreport-payment-shape} at \(v=V_1=4\) therefore gives
\[
 y_1(V_1)=\max\{y_{\rm lim},y^\star(4)\}=1.
\]
Thus, if user~1 reports
\((V_1,Q^{\rm large})=(4,4)\), it actually pays exactly its
true budget \(Q_1=1\), receives \(F(1)=1/2\), and obtains utility \(1\).
Incentive compatibility requires
\begin{equation}
 \operatorname{Util}_1^{\rm tr}\ge1.
 \label{eq:overreport-first-utility}
\end{equation}

\paragraph{The second user's truthful utility must be zero.}
Let \(Y^{\rm tr}=y_1^{\rm tr}+y_2^{\rm tr}\).  Exact settlement and
\(V_1>V_2\) give
\begin{align*}
 \operatorname{Util}_1^{\rm tr}+\operatorname{Util}_2^{\rm tr}
 &=V_1F(Y^{\rm tr})-Y^{\rm tr}
   -(V_1-V_2)x_2^{\rm tr}\\
 &\le\operatorname{Util}^{\rm pool}(V_1,Y^{\rm tr})
 \le\operatorname{Util}^{\rm pool}\bigl(V_1,y^\star(V_1)\bigr)=1.
\end{align*}
Together with \eqref{eq:overreport-first-utility} and individual
rationality, this forces \(\operatorname{Util}_2^{\rm tr}=0\).

\paragraph{The second user can also safely overreport its budget.}
Now let user~2 report \((V_2,Q^{\rm large})\), keeping user~1 at
\((V_1,Q_1)\).
The hypotheses of \cref{lem:overreport-lower-value} hold because
\[
 V_{\rm cut}=\frac{Q_1}{F(Q_1)}=2
 <V_2<V_1,
 \qquad
 Q^{\rm large}=4>\bar y(V_1)=3.
\]
The lemma gives user~2 strictly positive utility, while its realized payment
satisfies
\[
 y_2\le\bar y(V_2)-Q_1
 =\frac{51}{49}-1
 =\frac2{49}=Q_2.
\]
This deviation is therefore affordable under its true budget and strictly
improves upon its truthful utility of zero, a contradiction.

\end{proof}

Only the two deviations in the final contradiction overreport a budget.
Every preparatory incentive comparison keeps the budget fixed or reports a
smaller one.  Both overreports keep the value truthful and realize a payment
within the true budget.  This pinpoints the obstruction: a declared budget
can influence priority and allocation even when the mechanism never collects
the overstated portion.

Finally, the theorem also rules out a two-sided mechanism with the same four
properties.  Restrict such a mechanism to the two sell-\(Y\) profiles used
above.  Exact settlement becomes
\(x_1+x_2=F(y_1+y_2)\), both weak-Pareto clauses remain applicable, and
every deviation in the proof is still an available single sell-\(Y\)
deviation.  Thus the restricted mechanism would contradict
\cref{thm:budget-overreporting-impossibility}.  The aggregate true-budget
cap in \cref{def:two-sided-strategy} excludes exactly these two fatal
deviations.

\section*{Acknowledgments}
This work began in 2023, and the final writeup  
is a collaboration between human and artificial intelligence.
Human intelligence was solely responsible for all the conceptual
contributions, and for devising the main AMM adaptive clinching mechanism.
AI proved our conjecture regarding Sybil-robust incentive compatibility. 

Elaine Shi acknowledges Hubert Chan and 
Alireza Farhadi for helpful technical discussions during an earlier phase of the project.  
Elaine Shi is also grateful to Leo Fan and his family for their kindness, generosity, 
and the logistical support that made 
this cross-institutional collaboration possible.

\bibliographystyle{alpha}
\bibliography{refs}

\appendix

\section{Fixed and Online Supply as Constant Targets}
\label{app:constant-targets}

The target-based clinching mechanism of
\cref{subsec:abstract-target-mechanism} separates the clinching mechanics
from the source of supply.  This appendix shows that the two classical
supply models are instances of it: the ordinary 
adaptive clinching auction~\cite{adaptiveclinch,budgetelicitation} for a fixed
divisible supply, 
and adaptive clinching with exogenous online
supply, due to Goel, Mirrokni, and Paes Leme~\cite{clinchonline}.  In
both cases the target ignores the auction state \((X,Y)\) and the clock
price: it is the constant function equal to the available supply.
The
admissibility requirements of \cref{def:admissible-target}
then hold trivially.

The main content of the appendix is a rule-by-rule correspondence showing
that, with a constant target, the rules of \cref{sec:mechanism} reduce to
those of the classical auction.  The only genuine boundary
convention is that the abstract clock starts at a positive price
\(p_\star\), whereas the classical clock starts at zero.  A constant target
also attains its capacity, but the zero-trade completion convention makes
that boundary identical to the classical clock after all supply has been
allocated and creates no exception in the incentive theorem.

\subsection{Fixed Divisible Supply}
\label{app:fixed-supply}

\paragraph{The classical auction.}
We recall the adaptive clinching auction in the notation of
\cref{sec:mechanism}.  A divisible supply \(s>0\) of a single good is offered
to finitely many identities; \cite{budgetelicitation} normalizes \(s=1\), whereas
\cite{clinchonline} keeps \(s\) general.  Identity \(i\) reports a value
\(v_i\ge0\) per unit and a hard budget \(q_i\ge0\); these are \(v_i\) and
\(B_i\) in \cite{budgetelicitation,clinchonline}.  A price clock \(p\) ascends from
zero.  At clock price \(p\), identity \(i\) has cumulative allocation
\(x_i\), cumulative payment \(y_i\), and effective remaining budget \(b_i\),
which equals \(q_i-y_i\) while the identity is in the auction and is set to
zero when it exits.  The unallocated supply is \(S=s-\sum_ix_i\), the demand
of identity \(i\) is \(b_i/p\), and the sets \({\sf Active}\),
\({\sf Exit}\), and \({\sf Clinch}\) are those of
\eqref{eq:active-exit-sets} and \eqref{eq:clinch-set}.  The auction maintains
the supply invariant \eqref{eq:supply-invariant}.  Every trade is charged at
the current clock price.  A clinching identity trades at the rates
\begin{equation}
 \frac{\dif x_i}{\dif p}=\frac Sp,
 \qquad
 \frac{\dif y_i}{\dif p}=S,
 \qquad
 \frac{\dif b_i}{\dif p}=-S,
 \qquad i\in{\sf Clinch},
 \label{eq:classical-clinching-rates}
\end{equation}
and no other identity trades
\cite[Figure~1, Step~(II)]{budgetelicitation},
\cite[Definition~4.2(ii)]{clinchonline}.  When the clock reaches the
value of one or more identities, the stopping condition
\eqref{eq:stopping-test} is tested.  If it holds, every active identity
receives its full demand \(b_i/p\), the exiting identities receive the rest
of the supply, and the auction ends \cite[Figure~1, Step~(I)]{budgetelicitation}.
Otherwise the exiting budgets are cancelled, and each active identity \(i\)
clinches
\begin{equation}
 \Bigl(S-\sum_{k\in{\sf Active}\setminus\{i\}}\frac{b_k}{p}\Bigr)_+
 \label{eq:classical-exit-clinch}
\end{equation}
units at the exit price \cite[Lemma~3.1]{budgetelicitation},
\cite[Definition~4.2(iv)]{clinchonline}.

\begin{definition}[Constant target]
\label{def:constant-target}
Fix a supply \(s>0\) and an initial price \(p_\star>0\).  The
\emph{constant target} is
\begin{equation*}
 T^{\rm fix}_s(p,X,Y):=s
 \qquad\text{for every }(p,X,Y)\in\mathcal D^{\rm fix}_s,
\end{equation*}
on the state domain
\begin{equation*}
 \mathcal D^{\rm fix}_s
 :=\{(p,X,Y):p\ge p_\star,\ 0\le X\le s,\ Y\ge0\},
\end{equation*}
Let \(\sigma^{\rm fix}\) denote the trivial environment description, and
define the constant-target specification
\begin{equation*}
 \mathfrak T^{\rm fix}_s
 :=(p_\star,\sigma^{\rm fix},T^{\rm fix}_s,
    \mathcal D^{\rm fix}_s,s).
\end{equation*}
Its residual supply \(S=s-X\) is the unallocated supply of the classical
auction.  For the unit supply of \cite{budgetelicitation}, the target is the
constant \(1\).
\end{definition}

\begin{proposition}[Admissibility of the constant-target specification]
\label{prop:constant-target-admissible}
For every \(s>0\) and \(p_\star>0\), \(\mathfrak T^{\rm fix}_s\) is an
admissible abstract target-based mechanism.  Its frozen-state clock rate vanishes,
\(\lambda\equiv0\), and its only supply lump is the initial release
\(T^{\rm fix}_s(p_\star,0,0)=s\).  Unlike
the AMM capacity, the capacity \(s\) is attainable:
\(T^{\rm fix}_s=X\) precisely when \(X=s\).
\end{proposition}

\begin{proof}
We verify the five requirements of \cref{def:admissible-target} in order.

\begin{itemize}[leftmargin=6mm,itemsep=1pt]
\item 
\emph{Feasibility and partition independence.}  On the domain,
\(X\le s=T^{\rm fix}_s=\overline X\).  A constant depends on the bidders
through nothing at all; in particular it depends neither on their names nor
on the partition of a budget among identities.
\item 
\emph{Fixed-price invariance.}  If \(0\le\delta\le S=s-X\), then
\(X+\delta\le s\) and \(Y+p\delta\ge0\), so the updated state lies in
\(\mathcal D^{\rm fix}_s\).  The target is unchanged because it is constant.
\item 
\emph{Target decreases with mark-to-market gain.}  Both targets in
\eqref{eq:target-mark-to-market-order} equal \(s\).  The implication therefore
holds, with equality, whatever the two gains are.
\item 
\emph{Clock regularity.}  For fixed \((X,Y)\), the map \(p\mapsto s\) is
constant.  It is nondecreasing, right-continuous, of bounded variation, free
of jumps, and absolutely continuous, and \eqref{eq:target-clock-rate} gives
\(\lambda=0\).
\item 
\emph{Eventual supply release.}  The target equals \(X\) only if
\(X=s=\overline X\).  Thus there is no state satisfying the premise
\(T^{\rm fix}_s=X<\overline X\) of
\eqref{eq:no-subcapacity-absorption}.
\end{itemize}

Finally, \(\mathcal D^{\rm fix}_s\) contains the initial state
\((p_\star,0,0)\) and is closed under raising \(p\) with \((X,Y)\) held
fixed.  The positive initial target is released at \(p_\star\) as the initial
lump, as prescribed after \eqref{eq:initial-state}.
\end{proof}

\begin{proposition}[The constant-target mechanism is adaptive clinching]
\label{prop:constant-target-is-aca}
The target-based clinching mechanism associated with
\(\mathfrak T^{\rm fix}_s\) is the adaptive clinching auction for supply
\(s\) with its clock started at \(p_\star\).  Equivalently, apply the target
substitution \eqref{eq:abstract-target-substitution} with
\((p_\star,T_\sigma,\overline X_\sigma)
=(p_\star,T^{\rm fix}_s,s)\).  More precisely:
\begin{enumerate}
\item \emph{Start.}  Identities with \(v_i<p_\star\) are ineligible.  The
whole supply is released at \(p_\star\) as the initial lump, so \(S=s\), and
the event rules in items~3 and~4 are applied at \(p_\star\), with the
identities having \(v_i=p_\star\) as the exiting identities.
\item \emph{Continuous evolution.}  The derived rate
\eqref{eq:derived-clinching-rate} is \(\gamma=S\).  Clinching identities
trade at the classical rates \eqref{eq:classical-clinching-rates}, no other
identity trades, and, with \(c=|{\sf Clinch}|\),
\begin{equation}
 \frac{\dif S}{\dif p}=-\frac{cS}{p},
 \qquad
 \frac{\dif h}{\dif p}=-S.
 \label{eq:constant-target-aggregate-dynamics}
\end{equation}
In particular, \(S\) remains positive between events.
\item \emph{Nonstopping exits.}  If the stopping test fails, the one-shot
restoration purchase \eqref{eq:one-shot-restoration} of an active identity
\(i\) is the classical clinch \eqref{eq:classical-exit-clinch}.
\item \emph{Stopping.}  The stopping test \eqref{eq:stopping-test} and the
full-demand purchases of the active identities are Step~(I) of
\cite[Figure~1]{budgetelicitation}, and the proportional rule
\eqref{eq:proportional-exit-allocation} specifies how the remainder \(L\) of
\eqref{eq:exit-remnant} is divided among the exiting identities.  At every
stopping price above \(p_\star\), some identity exits and
\begin{equation}
 0\le L\le\frac1p\Bigl(B_E-\max_{i\in{\sf Exit}}b_i\Bigr).
 \label{eq:constant-target-remnant-bound}
\end{equation}
Thus \(L=0\) unless at least two identities exit simultaneously, the exiting
identities can afford all of \(L\), and the stopping allocation sells the
entire residual supply.
\item \emph{Capacity completion and termination.}  After the initial release
the target has no lump.  A nonstopping event saturates capacity exactly when
it leaves a single active identity, which buys the entire residual supply at
that price.  The allocation and payment are then complete, and the clock
follows the zero-trade tail of \cref{lem:capacity-zero-trade-tail} until that
identity exits.  Every run that does not enter such a tail ends with a
stopping allocation.  In particular, before capacity saturation at least two
identities are active whenever a run continues.
\end{enumerate}
\end{proposition}

\begin{proof}
Item~1 is the initialization rule stated after \eqref{eq:initial-state},
together with the eligibility rule that accompanies
\eqref{eq:abstract-target-substitution}.

For item~2, \cref{prop:constant-target-admissible} gives \(\lambda=0\), so
\cref{lem:derived-continuous-rate} yields \(\gamma=S+p\lambda=S\).
Equations \eqref{eq:derived-allocation-payment-ode} and
\eqref{eq:derived-budget-ode} then become
\eqref{eq:classical-clinching-rates}, and
\eqref{eq:abstract-aggregate-dynamics} becomes
\eqref{eq:constant-target-aggregate-dynamics}.  On an event-free interval
beginning at \(p_a\), the first equation integrates to
\(S(p)=S(p_a)(p_a/p)^c\), which is positive whenever \(S(p_a)\) is.

For item~3, after the cancellation \eqref{eq:cancel-exiting-budgets} the
aggregate budget is \(B=\sum_{k\in{\sf Active}}b_k\).  Hence
\(h^\star=\sum_{k\in{\sf Active}}b_k-pS\) and
\begin{equation}
 \frac{b_i-h^\star}{p}
 =S-\sum_{k\in{\sf Active}\setminus\{i\}}\frac{b_k}{p}.
 \label{eq:restoration-is-classical-clinch}
\end{equation}
The purchase \eqref{eq:one-shot-restoration} is made exactly when this
quantity is positive, which is \eqref{eq:classical-exit-clinch}.

For item~4, the test and the purchases
\(\operatorname{Buy}(i,b_i/p;p)\) are literally Step~(I) of
\cite[Figure~1]{budgetelicitation}.  Let \(p>p_\star\) be a stopping price.  Some
identity exits at \(p\).  Indeed, by \cref{lem:event-restoration} the
stopping test is equivalent to \(B_E\ge h\).  If an identity \(i\) is still
active at \(p\), then \(h\ge b_i>0\) by \eqref{eq:water-invariant}, so
\(B_E>0\); if none is, the identities that were active immediately before
\(p\) all exit at \(p\).  Every \(i\in{\sf Exit}\) was active immediately
before \(p\), and a constant target has no lump at \(p\).  The supply
invariant therefore gives
\[
 S\le D_{-i}
 =\sum_{k\in{\sf Active}}\frac{b_k}{p}
  +\sum_{k\in{\sf Exit}\setminus\{i\}}\frac{b_k}{p}.
\]
Subtracting the active demand yields \(L\le(B_E-b_i)/p\) for every
\(i\in{\sf Exit}\), which is
\eqref{eq:constant-target-remnant-bound}; and \(L\ge0\) is the stopping test.
Consequently \(L\le B_E/p\), the proportional rule distributes exactly \(L\),
and the total stopping allocation is
\(\sum_{k\in{\sf Active}}b_k/p+L=S\).

For item~5, a lump is an upward jump of the target, and a constant has none
after its initial release.  By item~2, \(S\) does not vanish between events,
so \(S=0\) can only be produced by the fixed-price purchases of an event.  A
stopping event ends the run.  At a nonstopping event, let \(m\) be the number
of active identities with \(b_i>h^\star\).  If \(m=0\), nothing is bought and
\(S\) stays positive.  If \(m\ge1\), then \(pS=\sum_{k\in{\sf Active}}b_k-h^\star\)
shows that the restoration purchases leave the residual supply
\[
 S-\sum_{i\in{\sf Active}}\frac{(b_i-h^\star)_+}{p}
 =\frac1p\Bigl(\sum_{i\in{\sf Active}:\,b_i\le h^\star}b_i
   +(m-1)h^\star\Bigr).
\]
Since \(h^\star>0\) by \cref{lem:event-restoration} and active budgets are
positive, this vanishes if and only if \({\sf Active}\) is a singleton.
Conversely, a single active identity \(i\) at a nonstopping event has
\(b_i>b_i-pS=h^\star\) and buys \(S\) by
\eqref{eq:restoration-is-classical-clinch}.  Then \(S=0\) and \(X=s\), so the
capacity is saturated and \cref{lem:capacity-zero-trade-tail} applies.  An
event that leaves no active identity passes the stopping test trivially.
Hence, outside the zero-trade tail, at least two identities are active
whenever a run continues.
\end{proof}

The zero-trade completion in item~5 is exactly the classical convention: once
the supply has been exhausted, the clock keeps rising with \(S=0\) until the
surviving identity reaches its value, and nothing further is traded.  Apart
from the start price, the only difference between the two descriptions is the
treatment of identities that exit simultaneously.
\cite{budgetelicitation,clinchonline} process them one at a time in a fixed order.  At
a nonstopping event this produces the same state as the simultaneous
cancellation and one-shot restoration, as noted after
\eqref{eq:one-shot-restoration}.  At a stopping event it divides \(L\) among
the tied exiting identities in an order-dependent way, where the mechanism
of \cref{sec:mechanism} uses the proportional rule.  By item~4 of
\cref{prop:constant-target-is-aca}, the two agree whenever no two identities
exit at the stopping price.

The remaining difference is the start of the clock.  The classical clock
starts at zero, where every positive budget has infinite demand, whereas
\cref{subsec:abstract-target-mechanism} fixes \(p_\star>0\).  By the
\emph{zero-start auction} we mean the auction of items~2--5 above with the
clock started at zero and identities with \(v_i=0\) removed at the outset.
If a single identity of positive value and budget remains, no other identity
has demand, and \eqref{eq:classical-exit-clinch} lets it clinch the whole
supply at price zero.

\begin{lemma}[Starting the clock at zero]
\label{lem:zero-start}
Consider a profile in which the set \(P\) of identities with \(v_i>0\) and
\(q_i>0\) has at least two elements, and put
\begin{equation}
 p^\dagger
 :=\min\Bigl\{\min_{i\in P}v_i,\;
   \frac1s\Bigl(\sum_{i\in P}q_i-\max_{i\in P}q_i\Bigr)\Bigr\}>0.
 \label{eq:first-event-price}
\end{equation}
For every \(p_\star\in(0,p^\dagger)\), the zero-start auction and the
target-based mechanism for \((p_\star,T^{\rm fix}_s,s)\) produce the same
outcome.
\end{lemma}

\begin{proof}
The price \(p^\dagger\) is the first event of the zero-start auction
\cite[proof of Theorem~D.2]{clinchonline}.  Indeed, no identity of \(P\)
exits below \(\min_{i\in P}v_i\), and as long as nothing has been traded,
\[
 D_{-i}=\sum_{j\in P\setminus\{i\}}\frac{q_j}{p}
 \ge\frac1p\Bigl(\sum_{j\in P}q_j-\max_{j\in P}q_j\Bigr)>s=S
 \qquad\text{for }p<p^\dagger,
\]
so the supply invariant is strict and nobody clinches.  At price
\(p_\star<p^\dagger\) the state of the zero-start auction is therefore
\(x_i=y_i=0\), \(b_i=q_i\) for \(i\in P\), and \(S=s\).

This is also the state of the target-based mechanism after its initial
event.  Identities outside \(P\) are ineligible or have no budget.  No
identity of \(P\) has \(v_i=p_\star\), so none exits.  The stopping test
fails because \(\sum_{i\in P}q_i>p_\star s\).  After the initial release,
\(h^\star=\sum_{i\in P}q_i-p_\star s>\max_{i\in P}q_i\), so no restoration
purchase is made.  From \(p_\star\) on, both auctions follow the rules of
\cref{prop:constant-target-is-aca}.
\end{proof}

Thus \(p_\star\) acts as a reserve price, and the classical auction is
recovered for all sufficiently small \(p_\star\): profile by profile, and
simultaneously on any finite family of profiles, such as a merged profile
and a Sybil profile.  The start price is visible only on a profile with a
single identity of positive value and budget.  That identity receives
\(\min\{s,q_i/p_\star\}\) units at the unit price \(p_\star\), which tends to
the whole supply at zero payment as \(p_\star\downarrow0\).

\subsection{Exogenous Online Supply}
\label{app:online-supply}

In the online supply model of
\cite[Definition~3.1 and Example~3.2]{clinchonline}, the supply of the
good is revealed over steps \(t=0,1,\dots,m\) as a nondecreasing sequence of
cumulative supplies
\begin{equation}
 0=s_0\le s_1\le\cdots\le s_m.
 \label{eq:online-supply-path}
\end{equation}
The sequence is exogenous: it may be deterministic, random, or adversarial,
but it does not depend on the reports, and at step \(t\) the auctioneer knows
only \(s_0,\dots,s_t\).  The identities, their reports, and their budgets are
fixed before the first step.  At every step the auctioneer must announce a
cumulative allocation and a cumulative payment for each identity, and it may
only augment those of the previous step.  \emph{Online adaptive clinching}
announces, at step \(t\), the outcome \((x(s_t),y(s_t))\) of the adaptive
clinching auction of \cref{app:fixed-supply} run to completion with supply
\(s_t\).  Between steps \(t-1\) and \(t\) it therefore allocates
\(x(s_t)-x(s_{t-1})\) and charges \(y(s_t)-y(s_{t-1})\).

Let the environment description be the supply revealed so far, and let the target be
its last entry:
\begin{equation*}
\begin{gathered}
 \sigma_t:=(s_0,\dots,s_t),
 \qquad
 T_{\sigma_t}(p,X,Y):=s_t,
 \qquad
 \overline X_{\sigma_t}:=s_t,
 \qquad
 \mathcal D_{\sigma_t}:=\mathcal D^{\rm fix}_{s_t},\\
 \mathfrak T_t
 :=(p_\star,\sigma_t,T_{\sigma_t},
    \mathcal D_{\sigma_t},\overline X_{\sigma_t}).
\end{gathered}
\end{equation*}

\begin{proposition}[Online supply as a family of constant targets]
\label{prop:online-supply-target}
Fix a realized supply path \eqref{eq:online-supply-path} and a step \(t\)
with \(s_t>0\).  Then \(\mathfrak T_t\) is an admissible abstract
target-based mechanism, and the cumulative outcome of online adaptive clinching
after step \(t\), with every clock started at \(p_\star\), is the outcome of
its target-based clinching mechanism.  The description \(\sigma_t\) is common
to all bid profiles.  If \(s_t=0\), no supply
has arrived and the cumulative outcome is zero.
\end{proposition}

\begin{proof}
Since \(T_{\sigma_t}=T^{\rm fix}_{s_t}\), the first two claims are
\cref{prop:constant-target-admissible,prop:constant-target-is-aca}, with
\cref{lem:zero-start} relating the start price \(p_\star\) to a clock started
at zero.  During the run that settles step \(t\), the fixed environment
description is \(\sigma_t\): the step index is not the clock.  Clock
regularity therefore holds with
\(\lambda_{\sigma_t}=0\) and no target lump.  Finally,
\(\sigma_t\) is determined by the supply path alone.  Because the path is
exogenous, every bid profile faces the same \(\sigma_t\), and hence the same
target and the same capacity.  This is what \cref{def:admissible-target}
requires of the capacity, and what
\cref{subsec:one-sided-sybil-proof} requires when it lets the merged run and
the Sybil run share the same environment description.
\end{proof}

\begin{remark}[Two time axes]
\label{rem:two-time-axes}
The online step \(t\) and the clock price \(p\) are different objects.  At
step \(t\) the clock is run afresh from \(p_\star\) against the constant
target \(s_t\); the auction does not pause one ascending clock and resume it
when more supply arrives.  Accordingly, online adaptive clinching is a family
of target-based mechanisms indexed by \(\sigma_t\), one for every prefix of
the supply path, and not a single run with a growing target.  A statement
about one prefix, such as a single-run lemma of
\cref{subsec:single-run-structure} or a comparison between two bid profiles
that face the same \(\sigma_t\), is a statement about one admissible abstract
target-based mechanism.
Whether the increments \(x(s_t)-x(s_{t-1})\) and \(y(s_t)-y(s_{t-1})\) are
nonnegative, so that the family can be implemented online, is a different
kind of statement: it compares the targets \(T_{\sigma_{t-1}}\) and
\(T_{\sigma_t}\) on one fixed bid profile.  This is the supply-monotonicity
theorem \cite[Theorem~4.1]{clinchonline}.  It is neither used nor implied
by the results of \cref{sec:one-sided-proofs}.  On a profile that has at
least two identities of positive value and budget and no tied values, and
with \(p_\star\) below the threshold \eqref{eq:first-event-price} for the
largest supply \(s_m\), which is the smallest threshold along the path, every
run is exactly the auction of \cite{clinchonline}, and that theorem
applies verbatim.
\end{remark}

\begin{remark}[A different notion: supply released along the clock]
\label{rem:clock-indexed-supply}
\Cref{def:admissible-target} also allows the fixed environment description
to contain a schedule indexed by the clock.  The simplest example is an
exogenous release schedule: a
nondecreasing, right-continuous function
\(s\colon[p_\star,\infty)\to[0,\infty)\), not identically zero, whose
continuous part is locally absolutely continuous, together with
\begin{gather*}
 \sigma:=s(\cdot),
 \qquad
 T_\sigma(p,X,Y):=s(p),
 \qquad
 \overline X_\sigma:=\sup_{p}s(p),\\
 \mathcal D_\sigma
 :=\{(p,X,Y):p\ge p_\star,\ 0\le X\le s(p),\ Y\ge0\}.
\end{gather*}
The first three requirements hold because the target ignores \((X,Y)\), and
the domain is closed under raising the clock because \(s\) is
nondecreasing.  Clock regularity is the assumption on \(s\), with
\(\lambda_\sigma=s'(p)\) almost everywhere and with supply lumps at the jumps
of \(s\).  If \(s(p)=X<\sup_qs(q)\), then some \(p'>p\) satisfies
\(s(p')>X\), which verifies \eqref{eq:no-subcapacity-absorption}.  Here a
single clock is never restarted: supply that
arrives at clock price \(p\) is sold at prices at least \(p\), clinchers trade
at the rate \(\gamma=S+ps'(p)\) of \eqref{eq:derived-clinching-rate}, and the
mechanism waits while \(s(p)=X\).  This variant exercises the clauses of
\cref{def:admissible-target} that concern clock-dependent environment data,
target lumps, and waiting on \(T_\sigma=X\), none of which is
used by the AMM target or by a constant target.  It is not the model of
\cite{clinchonline}, in which every prefix is settled by a fresh run, and
we mention it only to illustrate the definition.
\end{remark}

\subsection{Incentives for Fixed and Online Supply}
\label{app:constant-target-incentives}

The general incentive theorem now applies directly at an attained capacity:
capacity saturation begins a zero-trade tail rather than ending the execution
with unused budgets.  Thus no separate terminal comparison is needed.

\begin{corollary}[Fixed- and online-supply incentives]
\label{cor:constant-and-online-target-sybil}
Let \(s>0\) and \(p_\star>0\).  In the target-based mechanism for the
constant target \(T^{\rm fix}_s\), truthful participation as the merged
identity \((V,Q)\), where \(V\ge p_\star\), weakly dominates every
same-direction Sybil strategy \eqref{eq:one-sided-sybil-deviation}.  Hence,
at every step \(t\) with \(s_t>0\) of an exogenous online supply path, the
truthful merged identity weakly dominates every such strategy in cumulative
utility.  The same conclusions hold for the classical zero-start auctions,
with the simultaneous proportional rule at tied stopping exits.
\end{corollary}

\begin{proof}
The constant-target specification is admissible by
\cref{prop:constant-target-admissible}, so the fixed-supply conclusion is an
immediate application of \cref{thm:abstract-one-sided-sybil}.  At online
step \(t\), \cref{prop:online-supply-target} identifies the cumulative
outcome with the same mechanism for the constant target
\(T_{\sigma_t}=T^{\rm fix}_{s_t}\).  The supply prefix \(\sigma_t\) is
exogenous and therefore common to the merged and Sybil runs, so the theorem
applies pointwise to that prefix.

For the zero-start conclusions, assume \(V>0\) and \(Q>0\), since otherwise
no strategy yields positive utility.  If the merged profile and the Sybil
profile each contain at least two identities of positive value and budget,
choose \(p_\star\) below \(V\) and below both thresholds
\eqref{eq:first-event-price}; on a finite online path it may be chosen below
the thresholds for the largest supply, and hence for every prefix.  By
\cref{lem:zero-start}, the two positive-start executions equal their
zero-start counterparts.  In the remaining singleton cases the comparison
is immediate: either no outsider has positive value and budget, and the
merged identity obtains the whole supply at zero payment, or the Sybil group
has no such identity and obtains nothing.
\end{proof}

\end{document}